\documentclass[12pt]{article}
\usepackage[dvipsnames]{xcolor}
\usepackage{float}
\usepackage{amsmath}
\usepackage{amssymb}
\usepackage{amsthm}
\usepackage{thmtools}
\usepackage{graphicx}
\usepackage{setspace}
\usepackage[authoryear]{natbib}
\usepackage{appendix}
\usepackage{grffile}
\usepackage{rotating}
\usepackage{tikz}
\usepackage{multirow}
\usepackage[hang,flushmargin, multiple]{footmisc}
\usepackage[font={small}]{caption}
\usepackage{subcaption}
\usepackage{mathabx}

\usepackage{hyperref}
\hypersetup{
    colorlinks=true,
    linkcolor=Red,
    citecolor=Blue,
    filecolor=magenta,      
    urlcolor=cyan,
}

\newcommand{\E}{\mathbb{E}}

\newtheorem{proposition}{Proposition}[section]
\newtheorem{remark}{Remark}[section]
\newtheorem{theorem}{Theorem}[section]
\newtheorem{lemma}{Lemma}[section]

\newtheorem{definition}{Definition}
\newtheorem{assumption}{Assumption}
\newtheorem{example}{Example}

\allowdisplaybreaks

\usepackage{natbib}

\newcommand{\sumi}{\sum_{i=1}^{N}}
\newcommand{\sumt}{\sum_{t=1}^{T}}

\newcommand{\cW}{\mathcal W}
\newcommand{\cF}{\mathcal F}

\newcommand{\betahat}{\hat \beta}

\newcommand{\doubleY}{\dot{\widecheck Y}}
\newcommand{\doubleW}{\dot{\widecheck W}}

\usepackage{caption}
\usepackage{bbm}
\def\arrvline{\hfil\kern\arraycolsep\vline\kern-\arraycolsep\hfilneg}
\newcommand{\appendixpagenumbering}{
  \break
  \pagenumbering{arabic}
  \renewcommand{\thepage}{\arabic{page}}
}
\begin{document}

\title{\Large{\textbf{Panel Experiments and Dynamic Causal Effects:\\ A Finite Population Perspective}\thanks{ 
We thank Isaiah Andrews, Robert Minton, Karthik Rajkumar and Jonathan Roth for helpful discussions. We thank the editor, co-editor and two anonymous referees for valuable and constructive comments. We especially thank James Andreoni and Larry Samuelson for kindly sharing their data. Finally, we are grateful to Gary Chamberlain for early conversations about this project. Any remaining errors are our own. Rambachan gratefully acknowledges financial support from the NSF Graduate Research Fellowship under Grant DGE1745303.} }}

\author{\Large{Iavor Bojinov} \thanks{Technology and Operations Management Unit, Harvard Business School: \href{mailto:ibojinov@hbs.edu}{ibojinov@hbs.edu}}
\and \Large{Ashesh Rambachan} \thanks{Department of Economics, Harvard University: \href{mailto:asheshr@g.harvard.edu}{asheshr@g.harvard.edu}}
\and \Large{Neil Shephard} \thanks{Department of Economics and Department of Statistics, Harvard University: \href{mailto:shephard@fas.harvard.edu}{shephard@fas.harvard.edu}}}
\date{\today}

\maketitle
\thispagestyle{empty} 
\setcounter{page}{0}

{\singlespacing
\begin{abstract}
In panel experiments, we randomly assign units to different interventions, measuring their outcomes, and repeating the procedure in several periods. Using the potential outcomes framework, we define finite population dynamic causal effects that capture the relative effectiveness of alternative treatment paths. For a rich class of dynamic causal effects, we provide a nonparametric estimator that is unbiased over the randomization distribution and derive its finite population limiting distribution as either the sample size or the duration of the experiment increases. We develop two methods for inference: a conservative test for weak null hypotheses and an exact randomization test for sharp null hypotheses. We further analyze the finite population probability limit of linear fixed effects estimators. These commonly-used estimators do not recover a causally interpretable estimand if there are dynamic causal effects and serial correlation in the assignments, highlighting the value of our proposed estimator.

\medskip

\end{abstract}
}

\newpage
\clearpage
\section{Introduction}

Panel experiments, where we randomly assign units to different interventions, measuring their response and repeating the procedure in several periods, form the basis of causal inference in many areas of biostatistics (e.g., \cite{MurphyEtAl(01)}), epidemiology (e.g., \cite{Robins(86)}), and psychology (e.g., \cite{lillie2011n}). In experimental economics, many authors recognize the benefits of panel-based experiments, for instance \cite{BellemareEtAl(14), BellemareEtAl(16)} highlighted the potentially large gains in power and \cite{CziborEtAl(19)} emphasized that panel-based experiments may help uncover heterogeneity across units. Despite these benefits, panel experiments are used infrequently in part due to the lack of a formal statistical framework and concerns about how the impact of past treatments on subsequent outcomes may induce biases in conventional estimators \citep{CharnessEtAl(12)}. In practice, authors typically assume away this complication by requiring that the outcomes only depend on contemporaneous treatment, what is often called the ``no carryover assumption'' (e.g., \cite{AbadieEtAl(17)}, \cite{AtheyImbens(18)}, \cite{AtheyEtAl(18)-MatrixCompletion}, \cite{ImaiKim(19)}, \cite{ArkhangelskyImbens(19)-DoubleRobustId}, \cite{ImaiKim(20)-twoway}, \cite{ChaisemartinDHaultfoeuille(20)}). Even when researchers allow for carryover effects, they commonly focus on incorporating the uncertainty due to sampling units from some super-population as opposed to the design-based uncertainty, which arises due to the random assignment.\footnote{See \cite{AbadieEtAl(20)} for a discussion of the difference between sampling-based and design-based uncertainty in the cross-sectional setting.}

In this paper, we tackle these challenges by defining a variety of new panel-based dynamic causal estimands without evoking restrictions on the extent to which treatments can impact subsequent outcomes. Our approach builds on the potential outcomes formulation of causal inference and takes a purely design-based perspective on uncertainty, allowing us to be agnostic to the outcomes model \citep{Neyman(23),Kempthorne(55),Cox(58book),Rubin(74)}. Our main estimands are various averages of lag-$p$ dynamic causal effects, which capture how changes in the assignments affect outcomes after $p$ periods. We provide nonparametric estimators that are unbiased over the randomization distribution induced by the random design. By exploiting the underlying Martingale property of our unbiased estimators, we derive their finite population asymptotic distribution as either the number of sample periods, experimental units, or both increases. This is a new technique for proving finite population central limit theorems, which may be broadly useful and of independent interest to researchers.

We develop two methods for conducting nonparametric inference on these dynamic causal effects. The first uses the limiting distribution to perform conservative tests on weak null hypotheses of no average dynamic causal effects. The second provides exact randomization tests for sharp null hypotheses of no dynamic causal effects. We then highlight the usefulness of our framework by deriving the finite population probability limit of commonly used linear estimation strategies, such as the unit fixed effects estimator and the two-way fixed effects estimator. Such estimators are biased for a contemporaneous causal effect whenever there exists carryover effects and serial correlation in the assignment mechanism, underscoring the value of our proposed nonparametric estimator.

Finally, we illustrate our theoretical results in a simulation study and apply our framework to reanalyze a panel-based experiment. The simulation study illustrates our finite population central limit theorems under a variety of assumptions about the underlying potential outcomes and assignment mechanism. We confirm that conservative tests based on the limiting distribution of our nonparametric estimator control size well and have good rejection rates against a variety of alternatives. We finish by reanalyzing a panel experiment conducted in \cite{AndreoniSamuelson(06)}, which studies cooperative behavior in game theory and is a natural application of our methods. Participants in the experiment played a twice-repeated prisoners' dilemma many times, and payoff structure of the game was randomly varied across plays. The sequential nature of the experiment raises the possibility that past assignments may impact future actions as participants learn about the structure of the game over time. For example, the random variation in the payoff structure may induce participants to explore possible strategies. This motivates us to analyze the experiment using our methods that are robust to possible dynamic causal effects. We confirm the authors' original hypothesis that the payoff structure of the twice repeated prisoners' dilemma has significant contemporaneous effects on cooperative behavior. Moreover, we provide suggestive evidence of dynamic causal effects in this experiment --- the payoff structure of previously played games may affect cooperative behavior in the current game, which may be indicative of such learning.

Our design-based framework provides a unified generalization of the finite population literature in cross-sectional causal inference (as reviewed in \cite{ImbensRubin(15)}) and time series experiments \citep{BojinovShephard(19)} to panel experiments. Three crucial contributions differentiate our work from the existing literature. First, we focus on a much richer class of dynamic causal estimands, which answer a broader set of causal questions by summarizing heterogeneity across both units and time periods. Second, we derive two new finite population central limit theorems as the size of the population grows, and as both the duration and population size increase. Third, we compute the bias present in standard linear estimators in the presence of dynamic causal effects and serial correlation in the treatment assignment probabilities.

Our framework is also importantly distinct from foundational work by \cite{Robins(86)} and co-authors, that uses treatment paths for causal panel data analysis and focuses on providing super-population (or sampling-based) inference methods. In contrast, we avoid super-population arguments entirely. Our estimands and inference procedures are conditioned on the potential outcomes and all uncertainty arises solely from the randomness in assignments. Avoiding super-populations arguments is often attractive in panel data applications. For example, a company only operates in a finite number of markets (e.g., states or cities within the United States) and can only conduct advertising or promotional experiments across markets. Such panel experiments are increasingly common in industry \citep[e.g.][]{BojinovEtAl(2020)--HBR, BojinovEtAl(2020)--SwitchBack}.\footnote{Of course, in other applications, super-population arguments may be entirely natural. For example, in the mental healthcare digital experiments of \cite{BoruvkaAlmirallWitkiwitzMurphy(17)}, it is compelling to use sampling-based arguments as the experimental units are drawn from a larger group of patients for whom we wish to make inference on as, if successful, the technology will be broadly rolled out.} In econometrics, \cite{AbadieEtAl(17)} highlight the appeal of this design-based perspective in panel data applications. However, the panel-based potential outcome model developed in that work contains no dynamics as the authors primarily focus on cross-sectional data with an underlying cluster structure. Similarly, \cite{AtheyImbens(18)}, \cite{AtheyEtAl(18)-MatrixCompletion} and \cite{ArkhangelskyImbens(19)-DoubleRobustId} also introduce a potential outcome model for panel data, but assume away carryover effects. \cite{HeckmanEtAl(16)}, \cite{Hull(18)}, and \cite{Han(19)} consider a potential outcome model similar to ours but again rely on super-population arguments to perform inference. Additionally, an influential literature in econometrics focuses on estimating dynamic causal effects in panel data under rich models that allow heterogeneity across units, but does not introduce potential outcomes to define counterfactuals and also relies on super-population arguments for inference (e.g., see \cite{ArellanoBonhomme(16)}, \cite{ArellanoBlundellBonhomme(17)} and the review in \cite{ArellanoBonhomme(12)}).


\paragraph{Notation:} For an integer $t \geq 1$ and a variable $A_{t}$, we write $A_{1:t} := \left(A_{1}, \hdots, A_{t}\right)$. We compactly write index sets as $[N] := \{1, \hdots, N\}$ and $[T] := \{1, \hdots, T\}$. Finally, for a variable $A_{i,t}$ observed over $i \in [N]$ and $t \in [T]$, define its average over $t$ as $\bar A_{i\cdot} := \tfrac{1}{T} \sumt A_{i,t}$, its average over $i$ as $\bar A_{\cdot t} :=\tfrac{1}{N} \sumi A_{i,t}$ and its average over both $i$ and $t$ as $\bar A := \tfrac{1}{NT} \sumt \sumi A_{i,t}$.

\section{Potential outcome panel and dynamic causal effects}\label{section:set-up}

\subsection{Assignment panels and potential outcomes} 
Consider a panel in which $N$ units (e.g., individuals or firms) are observed over $T$ time periods. For each unit $i \in [N]$ and period $t \in [T]$, we allocate an assignment $W_{i, t} \in \cW$. The assignment is a random variable and we assume $|\cW| < \infty$. For a binary assignment $\cW = \{0,1\}$, we refer to ``1'' as treatment and ``0'' as control. 

The \textit{assignment path} for unit $i$ is the sequence of assignments allocated to unit $i$, denoted $W_{i,1:T} := (W_{i,1},...,W_{i,T})^\prime \in \cW^{T}$. The \textit{cross-sectional assignment} at time-$t$ describes all assignments allocated at period $t$, denoted $W_{1:N,t} := (W_{1,t},...,W_{N,t})^\prime \in \cW^{N}.$ The \textit{assignment panel} is the $N\times T$ matrix $W_{1:N,1:T} \in \cW^{N \times T}$ that summarizes the assignments given to all units over the sample period, where $W_{1:N, 1:T} := \left( W_{1:N,1}, \hdots , W_{1:N,T} \right) = \left( W_{1, 1:T}^\prime, \hdots , W_{N, 1:T}^\prime \right)^\prime$.

A \textit{potential outcome} describes what would be observed for a particular unit at a fixed point in time along any assignment path.
\begin{definition}\label{defn: potential outcome}
    The \textbf{potential outcome} for unit-$i$ at time-$t$ along assignment path $w_{i, 1:T} \in \mathcal{W}^{T}$ is written as $Y_{i, t}(w_{i, 1:T})$.
\end{definition}
\noindent In principle, the potential outcome can depend upon the entire assignment path allowing for arbitrary spillovers across time periods. Definition \ref{defn: potential outcome} imposes that there are no treatment spillovers across units \citep{Cox(58book)}.\footnote{The idea of defining potential outcomes as a function of assignment paths first appears in \cite{Robins(86)} and has been further developed in subsequent work such as \cite{Robins(94)}, \cite{RobinsGreenlandHu(99)}, \cite{MurphyEtAl(01)}, \cite{BoruvkaAlmirallWitkiwitzMurphy(17)} and \cite{BlackwellGlynn(18)}.}

\subsection{The potential outcome panel model}
We now define the potential outcomes panel model by restricting the potential outcomes for a unit in a given period not to be affected by future assignments.

\begin{assumption}\label{ass:nonanticipating}
    The potential outcomes are \textbf{non-anticipating} if, for all $i \in [N]$, $t \in [T]$, and $w_{i,1:T}, \tilde{w}_{i,1:T} \in \cW^{T}$,
         $Y_{i,t}(w_{i,1:T}) = Y_{i,t}(\tilde w_{i,1:T})$ whenever  $w_{i,1:t} = \tilde{w}_{i,1:t}$.
\end{assumption}

\noindent Non-anticipation still allows an arbitrary dependence on past and contemporaneous assignments, and arbitrary heterogeneity across units and time periods.\footnote{Allowing for rich heterogeneity in panel data models is often useful in many economic applications. For example, there is extensive heterogeneity across units in income processes \citep{BrowningEjrnaesAlvarez(10)} and the dynamic response of consumption to earnings \citep{ArellanoBlundellBonhomme(17)}. Time-varying heterogeneity is also an important feature. For example, it is a classic point of emphasis in studying human capital formation -- see \cite{ben-porath_production_1967}, \cite{griliches_estimating_1977} and more recently, \cite{cunha_chapter_2006} and \cite{cunha_estimating_2010}.} Under Assumption \ref{ass:nonanticipating}, the potential outcome for unit $i$ at time $t$ only depends on the assignment path for unit $i$ up to time $t$, allowing us to write the potential outcomes as $Y_{i,t}(w_{i,1:t})$. As notation, let ${\textbf Y_{i,t} } = \{ Y_{i,t}(w_{i,1:t}): w_{i,1:t} \in \mathcal{W}^{t}\}$ denote the collection of potential outcomes for unit $i$ at time $t$ and ${\textbf Y_{1:N, 1:T} } = \{ {\textbf Y_{i,t} } \,:\, i\in[N], t\in[T]\}$ denote the collection of potential outcomes for all units across all time periods. Along an assignment panel $w_{1:N, 1:t} \in \mathcal{W}^{N \times t}$ up to time $t$, let $Y_{1:N, 1:t}(w_{1:N, 1:t})$ denote the associated $N \times t$ matrix of outcomes for all units up to time $t$.

To connect the observed outcomes with the potential outcomes, we assume every unit complies with the assignment.\footnote{In some applications, this assumption may be unrealistic. For example, in a panel-based clinical trial, we may worry that patients do not properly adhere to the assignment. In such cases, our analysis can be re-interpreted as focusing on dynamic intention-to-treat (ITT) effects.} For all $i \in [N]$, $t \in [T]$, the observed outcomes for unit $i$ are $y_{i,1:T}^{obs} = Y_{i,1:T}(w_{i,1:T}^{obs})$, where $w_{i,1:T}^{obs}$ is the observed assignment path for unit $i$.

A panel of units, assignments and outcomes in which the units are non-interfering and compliant with the assignments and the outcomes obey Assumption \ref{ass:nonanticipating} is a \textit{potential outcome panel}. For $N = 1$, the potential outcome panel reduces to the potential outcome time series model in \cite{BojinovShephard(19)}. For $T = 1$, the potential outcome panel reduces to the cross-sectional potential outcome model (e.g., \cite{Holland(86)} and \cite{ImbensRubin(15)}).

\subsection{Assignment mechanism assumptions}\label{section:special-cases-treatment-assigment}

We focus on randomized experiments in which the assignment mechanisms for each period only depend on past assignments and observed outcomes, but not on future potential outcomes nor unobserved past potential outcomes.

\begin{definition}\label{ass:assignments} 
    The assignments are \textbf{sequentially randomized} if, for all $t \in [T]$ and any $w_{1:N,1:t-1} \in \mathcal{W}^{N \times (t-1)}$ 
        \begin{equation*}
            \Pr(W_{1:N,t}|W_{1:N,1:t-1}=w_{1:N,1:t-1},{\textbf Y_{1:N,1:T}}) = \Pr(W_{1:N,t}|W_{1:N,1:t-1}=w_{1:N,1:t-1},Y_{1:N,1:t-1}(w_{1:N,1:t-1})).
        \end{equation*}
\end{definition}
\noindent  It is common to focus on sequentially randomized assignments in biostatistics and epidemiology \citep{Robins(86), Murphy(03)}. This is the panel data analogue of an ``unconfounded'' or ``ignorable'' assignment mechanism in the literature on cross-sectional causal inference (as reviewed in Chapter 3 of \cite{ImbensRubin(15)}).\footnote{If the researcher further observes characteristics $X_{i, t}$ that are causally unaffected by the assignments, then the definition of a sequentially randomized assignment mechanism can be modified to additionally condition on past and contemporaneous values of the characteristics $X_{1:N, 1:t}$.} Since future potential outcomes and counterfactual past potential outcomes are unobservable, any feasible assignment mechanism must be sequentially randomized.



An important special case imposes further conditional independence structure across assignments. Let $W_{-i,t} := (W_{1,t},...,W_{i-1,t},W_{i+1,t},...,W_{N,t})$ and $\mathcal{F}_{1:N,t,T}$ be the filtration generated by $W_{1:N,1:t}$ and ${\textbf Y_{1:N,1:T}}$. 

\begin{definition}\label{ass:individualistic}
    The assignments are \textbf{individualistic} for unit $i$ if, for all $t \in [T]$ and any $w_{1:N,1:t-1} \in \mathcal{W}^{N \times (t-1)}$ 
    $$           
        \Pr(W_{i,t}| W_{-i,t}, \mathcal{F}_{1:N,t-1,T}) = \Pr(W_{i,t}|W_{i,1:t-1}=w_{i,1:t-1}, Y_{i,1:t-1}(w_{i,1:t-1})).
    $$
\end{definition}


\noindent An individualistic assignment mechanism further imposes that conditional on its own past assignments and outcomes, the assignment for unit $i$ at time $t$ is independent of the past assignments and outcomes of all other units as well as all other contemporaneous assignments. For example, the Bernoulli assignment mechanism, where  $\Pr(W_{i,t}|W_{-i,t},\mathcal{F}_{1:N,t-1,T}) = \Pr\left(W_{i, t}\right)$  for all $i\in[N]$ and $t\in[T]$, is individualistic. 

\begin{example}\label{ex:adaptive}
    Consider a food delivery firm that is testing the effectiveness of a new pricing policy across ten major U.S. cities \citep{kastelman2018switchback, Sneider2019switchback}. Each city is an experimental unit, and the intervention administers the appropriate pricing policy for a duration of one hour. The outcome is the total revenue generated during each hour of the experiment, $t\in [T]$ and from city $i \in [N]$. The firm wishes to learn the best policy for each city and the best overall policy across all cities. To do so, it may conduct a panel experiment with an individualistic treatment assignment in which the probability a particular pricing policy is administered in a given city over the next hour depends on prior observed revenue in that city in earlier hours of the experiment. 
\end{example}

\begin{remark}
    Many adaptive experimental strategies (such as the one described in Example \ref{ex:adaptive}), in which a series of units are sequentially exposed to random treatments whose probability vary depending on the past observed data, satisfy our individualistic sequentially randomized assignment assumptions \citep[e.g.,][]{robbins1952some,lai1985asymptotically}. Such experiments are widely used by technology companies to quickly discern user preferences in recommendation algorithms \citep{li2010contextual,li2016collaborative} and by academics interested in improving their power against a particular hypothesis \citep{van2008construction}. There has been a growing interest in drawing causal inferences based on the collected data in such adaptive experimental designs \citep{HadadEtAl(21), zhang2020inference}. Since the assignment probabilities are known to the researcher, our results can be viewed as providing finite population techniques for drawing causal conclusions from adaptive experiments. In the special case of our framework where $N = 1$, $t \in [T]$ indexes individuals arriving over time and there no carryover effects, our results in the subsequent section are the finite population analogue of the inference results in \citet{HadadEtAl(21)}.\footnote{The setup with $N=1$ was developed in \cite{BojinovShephard(19)}, but this connection to adaptive experiments has not been previously made.}
    \end{remark}

Our finite population central limit theorems require that the assignment mechanism be individualistic. In a non-individualistic assignment mechanism, the past outcomes of other units may affect the contemporaneous assignment of a given unit, which introduces complex dependence structure across units. A similar difficulty arises in the growing literature on relaxing the non-interference assumptions in cross-sectional experiments, where researchers allow one unit's potential outcomes to depend on another unit's assignments \citep[e.g., see][]{SavjeEtAl(19)}. To derive the asymptotic distribution of causal estimators in such settings, researchers typically require the assignment mechanism to be independent \citep{chin2018central} or at least have only limited dependence structure across units \citep{aronow2017estimating}.  

\subsection{Dynamic causal effects}\label{sect: dynamic causal effects}
A \textit{dynamic causal effect} compares the potential outcomes for unit $i$ at time $t$ along different assignment paths, which we denote by $\tau_{i, t}(w_{i,1:t}, \tilde w_{i,1:t}) := Y_{i, t}(w_{i,1:t}) - Y_{i, t}(\tilde w_{i,1:t})$ for assignment paths $w_{i,1:t}, \tilde w_{i,1:t} \in \mathcal{W}^{t}$. We use these dynamic causal effects to build up causal estimands of interest.

\subsubsection{Lag-$p$ dynamic causal effects and average dynamic causal effects}

Since the number of potential outcomes grows exponentially with the time period $t$, there is a considerable number of possible causal estimands. To make progress, we restrict our attention to a core class, referred to as the \textit{lag-$p$ dynamic causal effects}.

\begin{definition}\label{defn:lagp-dynamic-causal-effect}
    For $0 \leq p < t$ and $\mathbf{w},\tilde{\mathbf{w}} \in \mathcal{W}^{p+1}$, the \textbf{$i,t$-th lag-$p$ dynamic causal effect} is     
    \begin{align*}
        \tau_{i,t}(\mathbf{w}, \tilde {\mathbf{w}};p) := \begin{cases}  \tau_{i, t}(\{w_{i,1:t-p-1}^{obs}, \mathbf{w}\},  \{w_{i,1:t-p-1}^{obs},\tilde{\mathbf{w}}\}) & \text{ if } p<t-1 \\
            \tau_{i,t}(\mathbf{w}, \tilde{\mathbf{w}}) & \text{otherwise}.
        \end{cases}
        \end{align*}
\end{definition}

\noindent The $i,t$-th lag-$p$ dynamic causal effect measures the difference between the outcomes from following assignment path $\mathbf{w}$ from period $t-p$ to $t$ compared to the alternative path $\tilde {\mathbf{w}}$, fixing the assignments for unit $i$ to follow the observed path up to time $t-p-1$.
Generally, when $N>>T$ we recommend setting $p=t-1$, removing the dependence on the observed path.\footnote{In a time series experiment with $N = 1$, \cite{BojinovShephard(19)} introduced defining causal effects that depend on the observed assignment path because most potential outcomes are unobserved since there is only one experimental unit in their setting. In our more general panel experiments setting, an analogous problem arises when $T$ is of a similar order as $N$.}

By further restricting the paths $\mathbf{w}$ and $\tilde{\mathbf{w}}$ to share common features, we obtain the weighted average $i,t$-th lag-$p$ dynamic causal effect.   

\begin{definition}\label{defn:pq-causal-effect}
    For integers $p, q$ satisfying $0 \leq p < t$, $0 < q \leq p+1$, the \textbf{weighted average $i,t$-th lag-$p,q$ dynamic causal effect} is
    \begin{align*}
        \tau_{i,t}^\dagger(\mathbf{w}, \tilde {\mathbf{w}};p,q) := \sum_{\mathbf{v} \in \cW^{p-q+1}} a_{\mathbf{v}} \tau_{i,t}((\mathbf{w}, \mathbf{v}), (\tilde{\mathbf{w}}, \mathbf{v}); p)
        ,
        \end{align*}
        where $\mathbf{w},\tilde{\mathbf{w}} \in \mathcal{W}^{q}$ and $\{a_{\mathbf{v}}\}$ are 
        non-stochastic weights chosen by the researcher that satisfy $\sum_{\mathbf{v} \in \cW^{p-q+1}}a_{\mathbf{v}} = 1$ and $a_\mathbf{v} \geq 0$ for all $\mathbf{v} \in \cW^{p-q+1}$.
\end{definition}
 
\noindent The weighted average $i,t$-th lag-$p,q$ dynamic causal effect summarizes the \textit{ceteris paribus}, average causal effect of switching the assignment path between period $t-p$ and period $t-p+q$ from $\mathbf{w}$ to $\tilde{\mathbf{w}}$ on outcomes at time $t$.\footnote{For a binary assignment, setting $N=q=1$ gives us a special case that was studied in \cite{BojinovShephard(19)}.} In this sense, the weighted average lag-$p, q$ causal effect is a finite-population causal generalization of an impulse response function, which is a common estimand of interest in existing econometric research.\footnote{For time series experiments, \cite{RambachanShephard(20)} show that a particular version of the weighted average lag-$p,1$ causal effect is equivalent to the generalized impulse response function \citep{KoopPesaranPotter(96)}.} Whenever $q=1$, we drop the $q$ from the notation, simply writing $\tau_{i,t}^\dagger(w, \tilde {w};p) := \tau_{i,t}^\dagger(w, \tilde {w};p,1 )$.
 
The main estimands of interest in this paper are averages of the dynamic causal effects that summarize how different assignments impact the experimental units.
\begin{definition}\label{defn:dynamic-estimand}
     For $p < T$ and $\mathbf{w},\tilde{\mathbf{w}} \in \mathcal{W}^{p+1}$, 
     \begin{enumerate}
         \item the \textbf{time-$t$ lag-$p$ average dynamic causal effect} is $\bar{\tau}_{\cdot t}(\mathbf{w}, \tilde{\mathbf{w}};p) := \frac{1}{N} \sum_{i=1}^{N} \tau_{i,t}(\mathbf{w}, \tilde{ \mathbf{w}};p)$.
         
         \item the \textbf{unit-$i$ lag-$p$ average dynamic causal effect} is $\bar{\tau}_{i\cdot}(\mathbf{w}, \tilde{\mathbf{w}};p) := \frac{1}{T-p} \sum_{t=p+1}^{T} \tau_{i,t}(\mathbf{w}, \tilde {\mathbf{w}};p)$.
         
         \item the \textbf{total lag-$p$ average dynamic causal effect} is $\bar{\tau}(\mathbf{w}, \tilde{\mathbf{w}};p) := \frac{1}{N(T-p)} \sum_{t=p+1}^{T} \sum_{i=1}^{N} \tau_{i,t}(\mathbf{w}, \tilde {\mathbf{w}};p)$.
     \end{enumerate}
     
    \noindent These estimands extend to the weighted average $i,t$-th lag-$p$ dynamic causal effect by analogously defining $\bar\tau_{\cdot t}^\dagger(\mathbf{w}, \tilde{\mathbf{w}};p,q)$, $\bar\tau^\dagger_{i\cdot}(\mathbf{w}, \tilde{\mathbf{w}};p,q)$,  and $\bar\tau^\dagger(\mathbf{w}, \tilde{\mathbf{w}};p,q)$.
\end{definition}

We can augment any of the above averages to incorporate non-stochastic weights. For example, we could define $\{c_{i,t}\}_{i=1}^N$ the weights and consider the weighted time-$t$ lag-$p$ average dynamic causal effect $\frac{1}{N} \sum_{i=1}^{N} c_{i,t} \tau_{i,t}(\mathbf{w}, \tilde{ \mathbf{w}};p)$. These weights, for instance, could be used to adjust for different assignment path probabilities up to time $t - p - 1$, which are non-stochastic since the assignment mechanism is known.


\section{Nonparametric estimation and inference}\label{sect:non-param-estimation}

In this section, we develop a nonparametric \cite{HorvitzThompson(52)} type estimator of the $i,t$-th lag-$p$ dynamic causal effects and derive its properties. If the assignment mechanism is individualistic (Definition \ref{ass:individualistic}) and probabilistic (defined below), our proposed estimator is unbiased for the $i,t$-th lag-$p$ dynamic causal effects and its related averages over the assignment mechanism. An appropriately scaled and centered version of our estimator for the average lag-$p$ dynamic causal effects becomes approximately normally distributed as either the number of units or time periods grows large. These limiting results are finite population central limit theorems in the spirit of \cite{Freedman(08)}, and \cite{LiDing(17)}. 

\subsection{Setup: adapted propensity score and probabilistic assignment}
For each $i,t$, and any ${\mathbf w}= (w_1,\dots, w_{p+1}) \in \mathcal{W}^{(p+1)}$, the \textit{adapted propensity score} summarizes the conditional probability of a given assignment path and is given by $p_{i,t-p}({\mathbf w}) := \Pr(W_{i,t-p:t}={\mathbf w}|W_{i,1:t-p-1},Y_{i,1:t}(W_{i,1:t-p-1},{\mathbf w}))$. Even though the assignment mechanism is known, we only observe the outcomes along the realized assignment path $Y_{i,1:t}(w_{i,1:t}^{obs})$, and so it is not possible to compute $p_{i,t-p}({\mathbf w})$ for all assignment paths. However, we can compute the adapted propensity score along the observed assignment path, $p_{i,t-p}( w_{i,t-p:t}^{obs})$ (see Appendix \ref{section: additional theoretical results} for further discussion).

We next assume that the assignment mechanism is \textit{probabilistic}.

\begin{assumption}[Probabilistic Assignment]\label{ass:probtreat1} 
    Consider a potential outcome panel. There exists $C^L, C^U \in (0, 1)$ such that $C^L < p_{i,t-p}({\mathbf w}) < C^U$ for all $i \in [N]$, $t \in [T]$ and ${\mathbf w} \in \mathcal{W}^{(p+1)}$.
    
\end{assumption}

\noindent This is also commonly known as the ``overlap'' or ``common support'' assumption. 

All expectations, denoted by $\E[\cdot]$, are computed with respect to the probabilistic assignment mechanism. We write $\mathcal{F}_{i,t-p-1}$ as the filtration generated by $W_{i,1:t-p-1}$ and $\mathcal{F}_{1:N,t-p-1}$ as the filtration generated by $W_{1:N,1:t-p-1}$. Since we condition on all of the potential outcomes, conditioning on $W_{i,1:t-p-1}$ is the same as conditioning on both $W_{i,1:t-p-1}$ and $Y_{i,1:t-p-1}(W_{i,1:t-p-1})$.


\subsection{Estimation of the $i,t$-th lag-$p$ dynamic causal effect}
For any $\mathbf w, \tilde {\mathbf w} \in \mathcal{W}^{(p+1)}$, the nonparametric estimator of $\tau_{i,t}({\mathbf w}, \tilde{{\mathbf w}};p)$ is
    \begin{equation}
        \hat{\tau}_{i,t}({\mathbf w}, \tilde{{\mathbf w}};p) := \left\{ \frac{Y_{i,t}(w_{i,1:t-p-1}^{obs},{\mathbf w})\mathbbm{1}(w^{obs}_{i,t-p:t}={\mathbf w})}{p_{i,t-p}({\mathbf w})} - \frac{Y_{i,t}(w_{i,1:t-p-1}^{obs},\tilde{{\mathbf w}})\mathbbm{1}(w^{obs}_{i,t-p:t}=\tilde{{\mathbf w}})}{p_{i,t-p}(\tilde{{\mathbf w}})} \right\},
    \end{equation}
where $\mathbbm{1}\{A\}$ is an indicator function for an event $A$. Under individualistic assignments (Definition \ref{ass:individualistic}), the estimator simplifies to  $\hat{\tau}_{i,t}({\mathbf w}, \tilde{{\mathbf w}};p)  = \frac{y_{i,t}^{obs}\{ \mathbbm{1}(w^{obs}_{i,t-p:t}={\mathbf w}) - \mathbbm{1}(w^{obs}_{i,t-p:t}=\tilde{{\mathbf w}})\}}{p_{i,t-p}(w_{i,t-p:t}^{obs})}$.

\begin{theorem}\label{thm:thmMDs}
Consider a potential outcome panel with an assignment mechanism that is individualistic (Definition \ref{ass:individualistic}) and probabilistic (Assumption \ref{ass:probtreat1}). For any ${\mathbf w},\tilde{{\mathbf w}} \in \mathcal{W}^{(p+1)}$, 
    \begin{equation}\label{eqn:expectU}
        \E[\hat{\tau}_{i,t}({\mathbf w},\tilde{{\mathbf w}};p) \,|\, \mathcal{F}_{i,t-p-1}]=\tau_{i,t}({\mathbf w},\tilde{{\mathbf w}};p), \\
    \end{equation}
    \begin{equation}\label{eqn:varU}
         Var(\hat{\tau}_{i,t}({\mathbf w},\tilde{{\mathbf w}};p)|\mathcal{F}_{i,t-p-1}) = \gamma^2_{i,t}({\mathbf w},\tilde{{\mathbf w}}) - \tau_{i,t}({\mathbf w}, \tilde{{\mathbf w}};p)^2 := \sigma_{i,t}^2,
    \end{equation}
    where 
    \begin{equation}
        \gamma^2_{i,t}({\mathbf w},\tilde{{\mathbf w}};p) = \frac{Y_{i,t}(w_{i,1:t-p-1}^{obs},{\mathbf w})^2}{p_{i,t-p}({\mathbf w})} + \frac{Y_{i,t}(w_{i,1:t-p-1}^{obs},\tilde{{\mathbf w}})^2}{p_{i,t-p}(\tilde{{\mathbf w}})}.
    \end{equation}
  Further, for distinct ${\mathbf w},\tilde{{\mathbf w}},\bar{\mathbf w},\hat{\mathbf w}\in \mathcal{W}^{(p+1)}$
$$
Cov(\hat{\tau}_{i,t} ({\mathbf w},\tilde{{\mathbf w}};p),
\hat{\tau}_{i,t}(\bar{\mathbf w},\hat{\mathbf w};p)|\mathcal{F}_{i,t-p-1}) = - \tau_{i,t}({\mathbf w}, \tilde{{\mathbf w}};p)\tau_{i,t}(\bar{\mathbf w}, \hat{\mathbf w};p).
$$  
Finally, $\hat{\tau}_{i,t}({\mathbf w},\tilde{{\mathbf w}})$ and $\hat{\tau}_{j,t}({\mathbf w},\tilde{{\mathbf w}})$ are independent for $i \ne j$ conditional on $\mathcal{F}_{1:N,t-p-1}$.
\end{theorem}

Theorem \ref{thm:thmMDs} states that for every $i,t$, the error in estimating $\tau_{i,t}({\mathbf w}, \tilde{{\mathbf w}};p)$ is a martingale difference sequence through time and conditionally independent across units. The variance of $\hat{\tau}_{i,t}({\mathbf w},\tilde{{\mathbf w}};p)$ depends upon the potential outcomes under both the treatment and counterfactual and is generally not estimable. However, its variance is bounded from above by $\gamma_{i,t}^2({\mathbf w},\tilde{{\mathbf w}};p)$, which we can estimate by $\hat{\gamma}^2_{i,t}({\mathbf w}, \tilde{{\mathbf w}};p)= \frac{(y_{i,t}^{obs})^2\{\mathbbm{1}(w^{obs}_{i,t-p:t}={\mathbf w}) + \mathbbm{1}(w^{obs}_{i,t-p:t}=\tilde{{\mathbf w}})\}}{{p_{i,t-p}(w_{i,t-p:t}^{obs})^2}}$. The following proposition establishes that $\hat{\gamma}^2_{i,t}({\mathbf w}, \tilde{{\mathbf w}};p)$ is an unbiased estimator of $\gamma^2_{i,t}({\mathbf w}, \tilde{{\mathbf w}};p)$ and its error in estimating $\gamma^2_{i,t}({\mathbf w}, \tilde{{\mathbf w}};p)$ is also a martingale difference sequence through time and conditionally independent across units.

\begin{proposition}\label{prop:varbounds}
    Under the setup of Theorem \ref{thm:thmMDs}, $ \E[\hat{\gamma}^2_{i,t}({\mathbf w}, \tilde{{\mathbf w}};p)| \mathcal{F}_{i,t-p-1}] = \gamma^2_{i,t}({\mathbf w}, \tilde{{\mathbf w}};p)$. Additionally, $\hat{\gamma}^2_{i,t}({\mathbf w}, \tilde{{\mathbf w}};p)$ and $\hat{\gamma}^2_{j,t}({\mathbf w}, \tilde{{\mathbf w}};p)$ are independent for $i \neq j$ conditional on $F_{1:N, t-p-1}$.
\end{proposition}

The variance bound $\gamma^2_{i,t}({\mathbf w}, \tilde{{\mathbf w}};p)$ is different from the typical Neyman variance bound, derived under the assumption of a completely randomized experiment \citep[Chapter~5]{ImbensRubin(15)}. In a completely randomized experiment, there is a negative correlation between any two units' assignments since the total number of units assigned to each treatment is fixed. In our setting, all units' assignments are conditionally independent under individualistic assignments, precluding us from exploiting the negative correlation in deriving a bound.

\begin{remark}\label{remark: HT weighted average causal effect} 
Since the weighted average $i,t$-th lag-$p,q$ dynamic causal effects (Definition \ref{defn:pq-causal-effect}) are linear combinations of the $i,t$-th lag-$p$ dynamic causal effects, we can directly apply Theorem \ref{thm:thmMDs} and Proposition \ref{prop:varbounds}. We provide the details for the case when $q=1$. 

For $w, \tilde{w} \in \mathcal{W}$ and ${\mathbf v} \in \mathcal{W}^p$, the nonparametric estimator of $\tau_{i,t}^{\dagger}(w, \tilde{w}; p)$ is 
$$
\hat{{\tau}}_{i,t}^\dagger (w, \tilde{w};p) 
= \sum_{{\mathbf v} \in W^p} a_{\mathbf v}\left\{ \frac{Y_{i,t}(w_{i,1:t-p-1}^{obs}, w,{\mathbf v})\mathbbm{1}(w^{obs}_{i,t-p:t}=(w,{\mathbf v}))}{p_{i,t-p}({w, {\mathbf v}})} - \frac{Y_{i,t}(w_{i,1:t-p-1}^{obs},\tilde{w},{\mathbf v})\mathbbm{1}(w^{obs}_{i,t-p:t}=(\tilde{w}, {\mathbf v}))}{p_{i,t-p}(\tilde{w},{\mathbf v})} \right\}.
$$
Under an individualistic assignment mechanism, this estimator simplifies to $ \hat{\tau}_{i,t}^\dagger(w, \tilde{w};p) = \frac{a_{ w^{obs}_{i,t-p+1:t}} y_{i,t}^{obs}\{ \mathbbm{1}(w^{obs}_{i,t-p}=w)) - \mathbbm{1}(w^{obs}_{i,t-p}=\tilde{w}) \}}{p_{i,t-p}({w}^{obs}_{i,t-p:t})}$. This estimator is unbiased over the randomization distribution, and its variance can be bounded from above. For uniform weights, the rest of the generalizations follow immediately by noticing that we can replace all instances of $\mathbf w$ and $\tilde{\mathbf w}$ with $(\mathbf w,{\mathbf v})$ and $(\tilde{\mathbf w}, {\mathbf v})$. 
\end{remark}

\subsection{Estimation of lag-$p$ average causal effects}

The martingale difference properties of the nonparametric estimator means that the averaged plug-in estimators 
    \begin{align}
        & \hat{\bar{\tau}}_{\cdot t}({\mathbf w}, \tilde{\mathbf w};p) := \frac{1}{N} \sum_{i=1}^N \hat{\tau}_{i,t}({\mathbf w},\tilde{{\mathbf w}};p) \\
        & \hat{\bar{\tau}}_{i \cdot}({\mathbf w}, \tilde{\mathbf w};p) := \frac{1}{(T-p)} \sum_{t=p+1}^T \hat{\tau}_{i,t}({\mathbf w},\tilde{{\mathbf w}};p) \\
        & \hat{\bar{\tau}}({\mathbf w}, \tilde{\mathbf w};p) := \frac{1}{N(T-p)} \sum_{i=1}^N \sum_{t=p+1}^T \hat{\tau}_{i,t}({\mathbf w},\tilde{{\mathbf w}};p)
    \end{align}
are also unbiased for the average causal estimands $\bar{\tau}_{\cdot t}({\mathbf w}, \tilde{\mathbf w};p)$, $\bar{\tau}_{i \cdot }({\mathbf w}, \tilde{\mathbf w};p)$, and $\bar{\tau}({\mathbf w}, \tilde{\mathbf w};p)$, respectively. 
We next derive the limiting distribution of appropriately scaled and centered versions of these averaged estimators.

\begin{theorem}\label{thm:clts}
Consider a potential outcome panel with an individualistic (Definition \ref{ass:individualistic}) and probabilistic assignment mechanism (Assumption \ref{ass:probtreat1}). Further assume that the potential outcomes are bounded.\footnote{Assuming the potential outcomes are bounded is a common simplifying assumption made in deriving finite population central limit theorems. As discussed in \cite{LiDing(17)}, this assumption can often be replaced by a finite-population analogue of the Lindeberg condition in analyses of cross-sectional, randomized experiments.} Then, for any ${\mathbf w},\tilde{{\mathbf w}} \in \mathcal{W}^{(p+1)}$,
    \begin{align*}
        & \frac{\sqrt{N}\{\hat{\bar{\tau}}_{\cdot t}({\mathbf w}, \tilde {\mathbf w};p) - \bar{\tau}_{\cdot t}({\mathbf w}, \tilde {\mathbf w};p) \}}{\sigma_{\cdot t}} \xrightarrow{d} N(0,1) \quad \mbox{as } N \rightarrow \infty, \\ 
        & \frac{\sqrt{T-p}\{\hat{\bar{\tau}}_{i \cdot}({\mathbf w}, \tilde {\mathbf w};p) - \bar{\tau}_{i \cdot }({\mathbf w}, \tilde {\mathbf w};p) \}}{\sigma_{i\cdot}} \xrightarrow{d} N(0,1) \quad \mbox{as } T \rightarrow \infty, \\
        & \frac{\sqrt{N(T-p)}\{\hat{\bar{\tau}}({\mathbf w}, \tilde {\mathbf w};p) - \bar{\tau}({\mathbf w}, \tilde {\mathbf w};p) \}}{\sigma} \xrightarrow{d} N(0,1) \quad \mbox{as } NT \rightarrow \infty,
    \end{align*}
    where $\sigma_{\cdot t}$, $\sigma_{i \cdot}$, and $\sigma$ are the square root of the appropriate averages of $\sigma_{i,t}^2$, defined in \eqref{eqn:varU}.
\end{theorem}

Likewise, for bounded potential outcomes with an individualistic and probabilistic assignment mechanism, the scaled variances are
\begin{align}
    & N\times Var(\hat{\bar{\tau}}_{\cdot t}(w, \tilde w;p)|\mathcal{F}_{1:N,t-p-1}) = \mathbb{E}\left[\sigma_{\cdot t}^2 \,|\, \mathcal{F}_{1:N, t-p-1} \right], \\
    & (T-p)\times Var(\hat{\bar{\tau}}_{i \cdot}(w, \tilde w;p)|\mathcal{F}_{i,0}) =  \mathbb{E}[\sigma_{i\cdot}^2|\mathcal{F}_{i,0}], \\
    & N(T-p)\times Var(\hat{\bar{\tau}}(w, \tilde w;p)|\mathcal{F}_{1:N,0}) =   \mathbb{E}[\sigma^2|\mathcal{F}_{1:N,0}].
\end{align}
\noindent Following the same logic as earlier, we can establish unbiased and consistent estimators of the variance bounds of the averaged estimators.
\begin{proposition}\label{lemma: cons_var}
  Under the setup of Theorem \ref{thm:clts}, for any ${\mathbf w},\tilde{{\mathbf w}} \in \mathcal{W}^{(p+1)}$,
    \begin{align*}
        & \E\left[\left( \frac{1}{N} \sum_{i=1}^N \hat{\gamma}_{i,t}^2({\mathbf w},\tilde{{\mathbf w}};p) \right) \,|\, \cF_{1:N, t-p-1} \right] = \frac{1}{N} \sumi \gamma_{i,t}^2({\mathbf w},\tilde{{\mathbf w}};p), \\
        & \E\left[ \left( \frac{1}{(T-p)} \sum_{t=p+1}^T \hat{\gamma}_{i,t}^2({\mathbf w},\tilde{{\mathbf w}};p) - \frac{1}{(T-p)} \sum_{t=p+1}^T \gamma_{i,t}^2({\mathbf w},\tilde{{\mathbf w}};p) \right) \,|\, \mathcal{F}_{i,0} \right] = 0,  \\
        & \E\left[ \left( \frac{1}{N(T-p)} \sum_{i=1}^N \sum_{t=p+1}^T \hat{\gamma}_{i,t}^2({\mathbf w},\tilde{{\mathbf w}};p) - \frac{1}{N(T-p)} \sum_{i=1}^N \sum_{t=p+1}^T \gamma_{i,t}^2({\mathbf w},\tilde{{\mathbf w}};p) \right) \,|\, \mathcal{F}_{1:N,0} \right] = 0.
    \end{align*}
   Moreover, 
    \begin{align*}
        & \frac{1}{N} \sum_{i=1}^N \hat{\gamma}_{i,t}^2({\mathbf w},\tilde{{\mathbf w}};p) - \frac{1}{N} \sumi \gamma_{i,t}^2({\mathbf w},\tilde{{\mathbf w}};p) \xrightarrow{p} 0 \mbox{ as } N \rightarrow \infty, \\
        & \frac{1}{(T-p)} \sum_{t=p+1}^T \hat{\gamma}_{i,t}^2({\mathbf w},\tilde{{\mathbf w}};p) - \frac{1}{(T-p)} \sum_{t=p+1}^T \gamma_{i,t}^2({\mathbf w},\tilde{{\mathbf w}};p) \xrightarrow{p} 0 \mbox{ as } T \rightarrow \infty, \\
        & \frac{1}{N(T-p)} \sum_{i=1}^N \sum_{t=p+1}^T \hat{\gamma}_{i,t}^2({\mathbf w},\tilde{{\mathbf w}};p) - \frac{1}{N(T-p)} \sum_{i=1}^N \sum_{t=p+1}^T \gamma_{i,t}^2({\mathbf w},\tilde{{\mathbf w}};p) \xrightarrow{p} 0 \mbox{ as } NT \rightarrow \infty.
    \end{align*}
\end{proposition}

\noindent Proposition \ref{lemma: cons_var} shows that increasing the lag $p$ increases our estimator's variance, highlighting an important trade-off: increasing the lag $p$ reduces the dependence on the observed treatment path at the cost of increased variance. Striking the correct balance depends on the context and the design of the experiment.

Theorem \ref{thm:clts} and Proposition \ref{lemma: cons_var} naturally extend to the weighted average $i,t$-th lag-$p,q$ dynamic causal effect from Definition \ref{defn:pq-causal-effect} by using the estimator developed in Remark \ref{remark: HT weighted average causal effect}. 

\subsection{Confidence intervals and testing for lag-$p$ average causal effects}\label{section: testing}

Combining the variance bound estimators in Proposition \ref{lemma: cons_var} with the central limit theorems in Theorem \ref{thm:clts}, we can carry out conservative inference for $\bar{\tau}_{\cdot t}({\mathbf w},\tilde{{\mathbf w}};p)$, $\bar{\tau}_{i \cdot }({\mathbf w},\tilde{{\mathbf w}};p)$ and $\bar{\tau}({\mathbf w},\tilde{{\mathbf w}};p)$. Such techniques can be used to construct conservative confidence intervals or tests of weak null hypotheses that the average dynamic causal effects are zero. For example, these may be $H_0: \bar{\tau}_{i \cdot }({\mathbf w},\tilde{{\mathbf w}};p)=0$ for $i=3$ or $H_0: \bar{\tau}_{\cdot t}({\mathbf w},\tilde{{\mathbf w}};p)=0$ for $t=4$. 

Alternatively, we may construct exact tests for sharp null hypotheses. An example of such a sharp null hypothesis is $H_0: \bar{\tau}_{i,t}({\mathbf w},\tilde{{\mathbf w}};p)=0,$ for all, ${\mathbf w}, \tilde{{\mathbf w}}$, $i\in[N]$ and specific $t=4$. Since all potential outcomes are known under such sharp null hypotheses, we can simulate the assignment path $W_{i,t-p:t}|W_{i,1:t-p-1}^{obs},y_{i,1:t-p-1}^{obs}$ for each unit $i$ and compute $\hat{\tau}_{i,t}({\mathbf w},\tilde{{\mathbf w}};p)$ at each draw. Therefore, we may simulate the exact distribution of any test statistics under the sharp null hypothesis and compute an exact $p$-value for the observed test statistic. These randomization tests only require us to be able to simulate from the randomization distribution of the assignments paths. Therefore, such randomization tests may also be conducted if the treatment assignment mechanism is sequentially randomized (Definition \ref{ass:assignments}).

\section{Estimation in a linear potential outcome panel}\label{section:estimation-linear}

This section explore the properties of commonly used linear estimators, such as the canonical unit fixed-effects estimator and two-way fixed effects estimator, under the potential outcomes panel model. We establish that if there are dynamic causal effects and serial correlation in the treatment assignment mechanism, both the unit fixed-effects estimator and the two-way fixed effects estimator are asymptotically biased for a weighted average of contemporaneous causal effects. In Appendix \ref{section: additional theoretical results}, we consider analyzing the panel experiment as a repeated cross-section, estimating a separate linear model in each period $t$.

Throughout this section, we further assume that the potential outcomes themselves are a linear function of the assignment path.

\begin{definition}\label{defn:linearPanel}
A {\bf linear potential outcome panel} is a potential outcome panel where
$$
Y_{i,t}(w_{i,1:t}) = \beta_{i,t,0} w_{i,t} + \hdots
        + \beta_{i,t,t-1} w_{i,1}
        + \epsilon_{i,t} \quad \forall t \in [T] \text{ and } i \in [N],
$$
and the non-stochastic coefficients $\beta_{i,t,0:t-1}$ and non-stochastic error $\epsilon_{i,t}$ do not depend upon treatments. 
\end{definition}

We adapt notation used in \cite{Wooldridge(05)} for analyzing panel fixed effects models. For a generic random variable $A_{i, t}$, we compactly write the within-period transformed variable as $\dot{A}_{i, t} = A_{i, t} - \bar{A}_{\cdot t}$ and the within-unit transformed variable as $\widecheck{A}_{i, t} = A_{i, t} - \bar{A}_{i \cdot}$. The within-unit and within-period transformed variable is $\dot{\widecheck{A}}_{i,t} = (A_{i, t} - \bar{A}) - (\bar{A}_{\cdot t} - \bar{A}) - (\bar{A}_{i \cdot} - \bar{A})$.

\subsection{Interpreting the unit fixed effects estimator}
Our next result characterizes the finite population probability limit of the unit fixed effects estimator, $\betahat_{UFE} = \sumi \sumt \widecheck Y_{i, t} \widecheck W_{i, t} / \sumi \sumt \widecheck W_{i, t}^2$, under the linear potential outcome panel model. Define $Cov(\widecheck W_{i, t}, \widecheck W_{i, s} | \cF_{1:N, 0, T}) := \widecheck \sigma_{W, i, t, s}$ and $\widecheck{\mu}_{i, t} := \mathbb{E}\left[ \widecheck W_{i, t} | \cF_{1:N, 0,  T} \right]$.

\begin{proposition}\label{prop:unit-fe-additive-general-case}
    Assume a linear potential outcome panel and that the assignment mechanism is individualistic (Definition \ref{ass:individualistic}) with $Var(\widecheck W_{i, t} | \cF_{1:N, 0, T}) := \widecheck \sigma^2_{W, i, t} < \infty$ for each $i \in [N]$, $t \in [T]$. Further assume that as $N \rightarrow \infty$, the following sequences converge non-stochastically:
        \begin{align*}
            & N^{-1} \sum_{i=1}^{N} \beta_{i, t, s} \widecheck \sigma_{W, i, t, s} \rightarrow \widecheck \kappa_{W, \beta, t, s} \quad \forall t \in [T] \, \& \, s \leq t, \\
            & N^{-1} \sumi \widecheck \sigma_{W, i, t}^2 \rightarrow \widecheck \sigma^2_{W, t} \quad \forall t \in [T], \\
            & N^{-1} \sum_{i=1}^{N} \widecheck{Y}_{i,t}({\bf 0}) \widecheck \mu_{i, t} \rightarrow \widecheck \delta_{t} \quad \forall t \in [T]. 
        \end{align*}
    Then, as $N \rightarrow \infty$, 
        \begin{align*}
            \betahat_{UFE} \xrightarrow{p} \frac{\sumt \widecheck \kappa_{W, \beta, t, t}}{\sumt \widecheck \sigma^2_{W, t}} + \frac{\sumt \sum_{s=1}^{t-1} \widecheck \kappa_{W, \beta, t, s}}{\sumt \widecheck \sigma^2_{W, t}} + \frac{\sumt \widecheck \delta_t}{\sumt \widecheck \sigma^2_{W, t}}.
        \end{align*}
\end{proposition}

Proposition \ref{prop:unit-fe-additive-general-case} decomposes the finite population probability limit of the unit fixed effects estimator into three terms. The first term is a weighted average of contemporaneous dynamic causal coefficients, describing how the contemporaneous causal coefficients covary with the within-unit transformed assignments over the assignment mechanism. The second term captures how past causal coefficients covary with the within-unit transformed treatments and arises due to the presence of dynamic causal effects. The last term is an additional error that arises due to the possible relationship between the demeaned counterfactual $\widecheck{Y}_{i, t}(\textbf{0})$ and the average, demeaned treatment assignment. A sufficient condition for the last term to be equal zero is for the counterfactual outcomes to be time invariant $Y_{i,t}({\mathbf 0}) = \alpha_{i}$, in which case $\widecheck{Y}_{i, t}(\textbf{0}) = 0$ for all $i \in [N], t \in [T]$. Therefore, the last term is zero whenever unit fixed effects are correctly summarize the variation in the ``control-only'' counterfactual outcomes across units and time.

Proposition \ref{prop:unit-fe-additive-general-case} is related to yet crucially different from results in \cite{ImaiKim(19)}, which show that the unit fixed effects estimator recover a weighted average of unit-specific contemporaneous causal effects if there are no carryover effects. In contrast, we establish that the unit fixed effects estimator does \textit{not} recover a weighted average of unit-specific contemporaneous causal effects in the presence of carryover effects and persistence in the treatment path assignment mechanism. 

\begin{example}
    Consider a linear outcome panel model with, for all $t > 1$, $Y_{i, t}(w_{i, 1:t}) = \beta_{0} w_{i, t} + \beta_{1} w_{i, t-1} + \epsilon_{i, t}$ and $Y_{i, 1}(w_{i, 1}) = \beta_{0} w_{i, 1} + \epsilon_{i, 1}$ for $t = 1$. Assume $Var(\widecheck W_{i, t} | \cF_{1:N, 0, T}) = \widecheck \sigma^2_{W, t}$ for all $t$ and $Cov(\widecheck W_{i, t}, \widecheck W_{i, t-1} | \cF_{1:N, 0, T}) = \widecheck \sigma_{W, t, t-1}$ for all $t > 1$ are constant across units. In this case, Proposition \ref{prop:unit-fe-additive-general-case} implies
        \begin{equation*}
            \hat{\beta}_{UFE} \xrightarrow{p} \beta_{0} + \beta_{1} \frac{ \sum_{t = 2}^{T} \widecheck{\sigma}_{W, t, t-1} }{ \sum_{t=1}^{T} \widecheck{\sigma}^2_{W, t} } + \frac{\sum_{t=1}^{T} \widecheck{\delta}_t }{ \sum_{t=1}^{T} \widecheck{\sigma}^2_{W, t}  }.
        \end{equation*}
    The unit fixed effects estimator converges in probability to the contemporaneous dynamic causal coefficient $\beta_0$ plus a bias that depends on two terms. The first component of the bias depends on the lag-$1$ dynamic causal coefficient and the covariance between assignments across periods.
\end{example}

\subsection{Interpreting the two-way fixed effects estimator}

Consider the two-way fixed-effect estimator is $\hat{\beta}_{TWFE} = \sum_{i=1}^N \sum_{t=1}^T \dot{\widecheck Y}_{i,t} \dot{\widecheck W}_{i,t}/\sum_{i=1}^N \sum_{t=1}^T \dot{\widecheck W}_{i,t}^2.$ Define $E(\dot{\widecheck W}_{i,t} | \mathcal{F}_{1:N, 0,T}) := \dot{\widecheck \mu}_{i, t}$ and $Cov(\doubleW_{i, t}, \dot{\widecheck W}_{i,s}) := \dot{\widecheck \sigma}_{W, i, t, s}$.

\begin{proposition}\label{prop:two-way-fe-additive-general-case}
    Assume a linear potential outcome panel and assume that the assignment mechanism is individualistic and $Var(\doubleW_{i, t} | \cF_{1:N, 0, T}) := \dot{\widecheck \sigma}^2_{W, i, t} < \infty$ for each $i \in [N]$, $t \in [T]$. Further assume that as $N \rightarrow \infty$, the following sequences converge non-stochastically
    \begin{align*}
        & N^{-1} \sum_{i=1}^{N} \beta_{i, t, s} \dot{\widecheck \sigma}_{W, i, t, s} \rightarrow \dot{\widecheck \kappa}_{W, \beta, t, s} \quad \forall t \in [T] \, \& \, s \leq t, \\
        & N^{-1} \sumi \dot{\widecheck \sigma}_{W, i, t}^2 \rightarrow \dot{\widecheck \sigma}^2_{W, t} \quad \forall t \in [T], \\
        & N^{-1} \sum_{i=1}^{N} \dot{\widecheck Y}_{i, t}(\textbf{0}) \dot{\widecheck \mu}_{i, t} \rightarrow \dot{\widecheck \delta}_{t} \quad \forall t \in [T]. 
    \end{align*}
    Then, as $N \rightarrow \infty$, 
        \begin{align*}
            \betahat_{TWFE} \xrightarrow{p} \frac{ \sumt \dot{\widecheck \kappa}_{W, \beta, t, t}}{\sumt \dot{\widecheck \sigma}^2_{W, t}} + \frac{\sumt \sum_{s=1}^{t-1} \dot{\widecheck \kappa}_{W,\beta, t, s}}{\sumt \dot{\widecheck \sigma}^2_{W, t}} + \frac{\sumt \dot{\widecheck \delta}_t}{\sumt \dot{\widecheck \sigma}^2_{W, t}}
        \end{align*}
\end{proposition}

Similar to Proposition \ref{prop:unit-fe-additive-general-case}, the two-way fixed effects estimand can be decomposed into three components under the linear potential outcome panel model, where the interpretation of each component is similar to the unit fixed effects estimator. A simple sufficient condition for the last term to equal zero is for counterfactual outcome to be additively separable into a time-specific and unit-specific effect, $Y_{i,t}(\mathbf{0}) = \alpha_i + \lambda_{t}$ for all $i \in [N], t \in [T]$. Therefore, the last term is zero whenever unit and time fixed effects are correctly summarize the variation in the ``control-only'' counterfactual outcomes across units and time.

An active literature in econometrics analyzes the two-way fixed effects estimator under various identifying assumptions. For example, \cite{ChaisemartinDHaultfoeuille(20)} rule out carryover effects and decompose the two-way fixed effects estimand under a ``common-trends'' assumption that restricts how the potential outcomes under control evolve over time across groups. \cite{AbrahamSun(20)} decompose the two-way fixed effects estimand in staggered designs (meaning units receive the treatments at some period and forever after) under a common-trends assumption. \cite{BoryusakJaravel(17)}, \cite{AtheyImbens(18)} and \cite{GoodmanBacon(18)} also provide a decomposition of the two-way fixed effects estimand in staggered designs. Proposition \ref{prop:two-way-fe-additive-general-case} provides a decomposition in panel experiments without restrictions on the carryover effects, whereas these existing decompositions are useful in observational settings where other identifying assumptions may be plausible.

\section{Simulation Study}\label{section:simulation}
We conduct a simulation study to investigate the finite sample properties of the asymptotic results presented in Section \ref{sect:non-param-estimation}. These simulations show that the finite population central limit theorems (Theorem \ref{thm:clts}) hold for a moderate number of treatment periods and experimental units. The proposed conservative tests for the weak null of no average dynamic causal effects have correct size and reasonable rejection rates against a range of alternatives.

\subsection{Simulation design} 
We generate the potential outcomes for the panel experiment using an autoregressive model,
    \begin{equation}\label{eqn: simulation design}
        Y_{i,t} = \phi_{i, t, 1} Y_{i, t-1}(w_{i, 1:t-1}) + \hdots, \phi_{i, t, t-1} Y_{i, 1}(w_{i, 1}) + \beta_{i, t, 0} w_{i, t} + \hdots + \beta_{i, t, t-1} w_{i, 1} + \epsilon_{i, t} \quad \forall t > 1,  
    \end{equation}
$Y_{i, 1}(w_{i, 1}) = \beta_{i, 1, 0} w_{i, 1} + \epsilon_{i, 1}$ with $\phi_{i, t, 1} = \phi$, $\phi_{i, t, s} = 0$ for $s > 1$, $\beta_{i, t, 0} = \beta$ and $\beta_{i, t, s} = 0$ for $s > 0$. We vary the choice $\phi$, which governs the persistence of the process, and $\beta$, which governs the size of the contemporaneous causal effects. We vary the probability of treatment $p_{i, t-p}(w) = p(w)$ as well as the distribution of the errors $\epsilon_{i,t}$, which we either sample from a standard normal or Cauchy distribution.

We document the performance of our nonparametric estimators over the randomization distribution, meaning that we first generate the potential outcomes ${\textbf Y_{1:N, 1:T} }$ and simulate over different assignment panels $W_{1:N, 1:T}$, holding the potential outcomes fixed. In the main text, we focus on evaluating the properties of our estimator for the total average dynamic causal effect $\hat{\bar{\tau}}(1, 0; 0)$. Appendix \ref{section: additional simulation results} explores the properties of our estimators for the time-$t$ average $\hat{\bar{\tau}}_{\cdot t}(1, 0; 0)$ and the unit-$i$ average $\hat{\bar{\tau}}_{i \cdot}(1, 0;0)$, as well as our estimators of the lag-$1$ weighted average dynamic causal effects $\hat{\bar{\tau}}^{\dagger}_{\cdot t}(1, 0; 1), \hat{\bar{\tau}}^{\dagger}_{i \cdot}(1, 0; 1)$ and $\hat{\bar{\tau}}^{\dagger}(1, 0; 1)$. 

\begin{figure}[htbp!]
    \centering
    \begin{subfigure}{.5\textwidth}
    \centering
    \includegraphics[width=2.5in, height=2.5in]{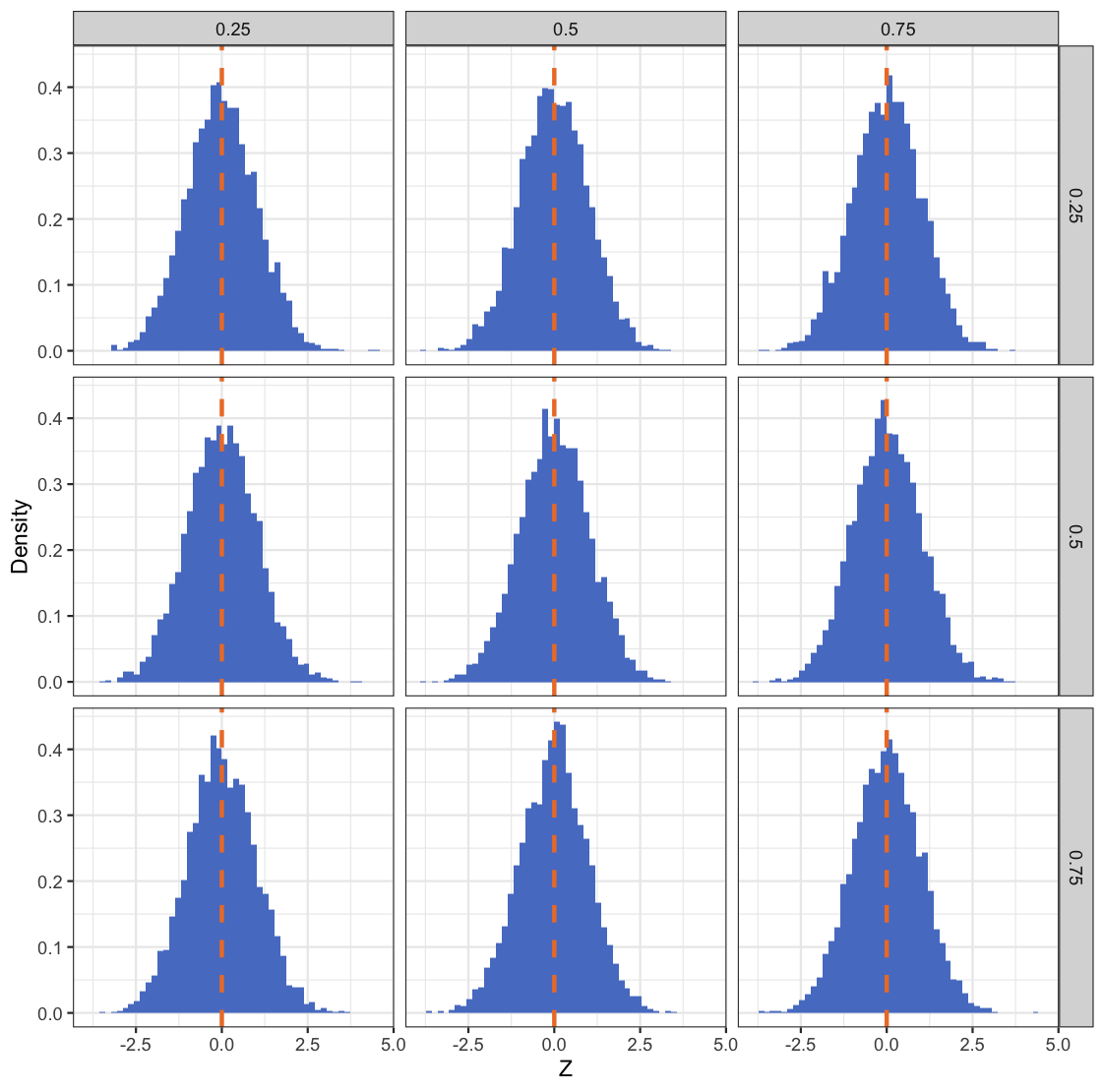}
    \caption{$\epsilon_{i, t} \sim N(0, 1)$, $N = 100$, $T = 10$}
    \end{subfigure}%
    \begin{subfigure}{.5\textwidth}
    \centering
    \includegraphics[width=2.5in, height=2.5in]{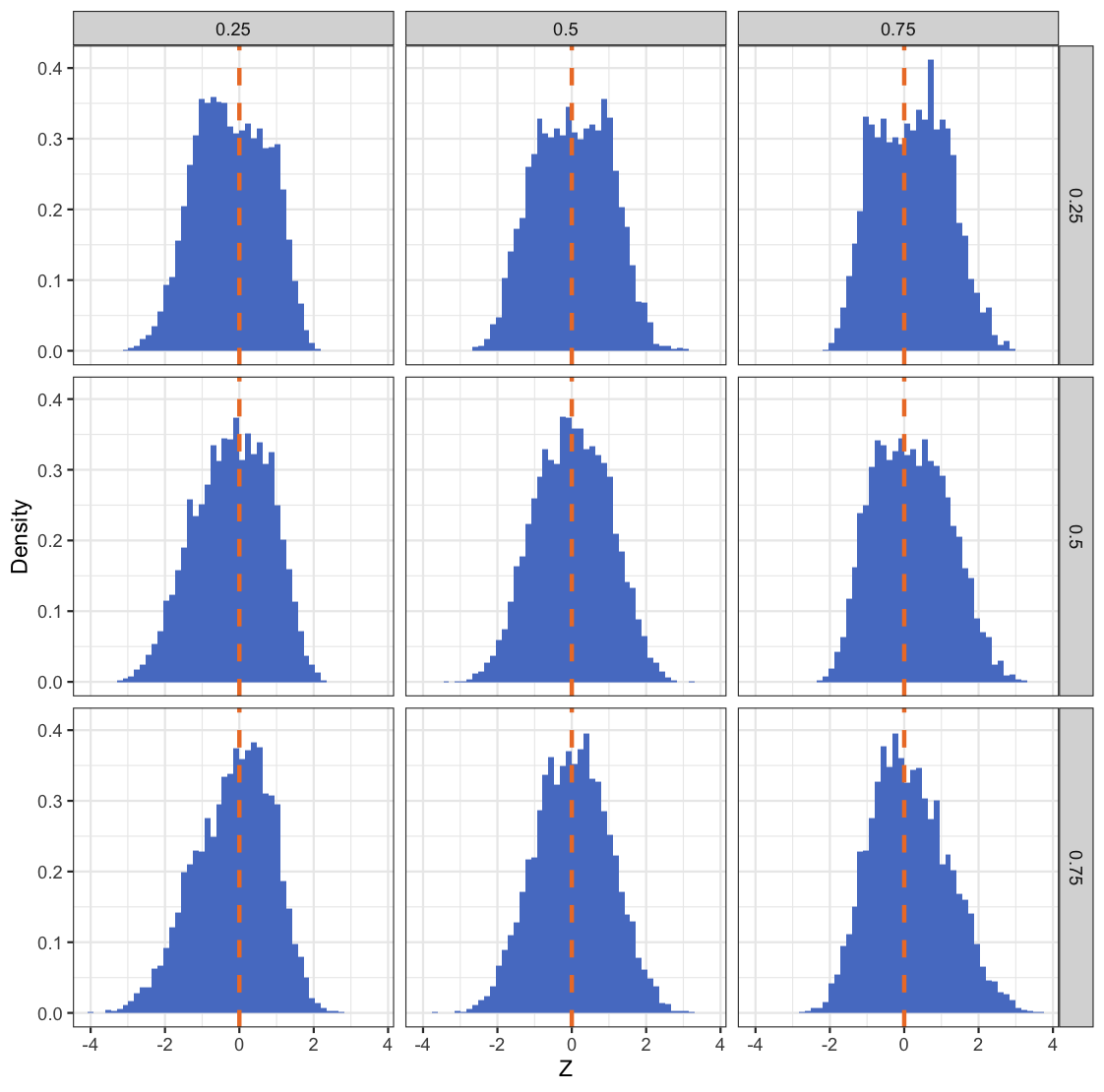}
    \caption{$\epsilon_{i, t} \sim Cauchy$, $N = 500$, $T = 100$}
    \end{subfigure}
    \caption{Simulated randomization distribution for $\hat{\bar{\tau}}(1, 0; 0)$ under different choices of the parameter $\phi$ (defined in (\ref{eqn: simulation design})) and treatment probability $p(w)$. The rows index the parameter $\phi \in \{0.25, 0.5, 0.75\}$. The columns index the treatment probability $p(w) \in \{0.25, 0.5, 0.75\}$. Panel (a) plots the simulated randomization distribution with normally distributed errors $\epsilon_{i, t} \sim N(0, 1)$ and $N = 100, T = 10$. Panel (b) plots the simulated randomization distribution with Cauchy distribution errors $\epsilon_{i, t} \sim Cauchy$ and $N = 500, T = 10$. Results are computed over 5,000 simulations.}
    \label{fig: total average histogram plag0}
\end{figure}

\subsection{Normal approximations and size control}

Figure \ref{fig: total average histogram plag0} plots the randomization distribution for the estimator of the total average dynamic causal effect $\hat{\bar{\tau}}(1, 0; 0)$. We present results for the case with $N = 100,\ T = 10$ and $N = 500,\ T = 100$ (the results are similar when the roles of $N, T$ are reversed). When the errors $\epsilon_{i, t}$ are normally distributed, the randomization distribution quickly converges to a normal distribution. When the errors are Cauchy distributed, the total number of units and time periods must be quite large for the randomization distribution to become approximately normal. There is little difference in the results across the values of $\phi$ and $p(w)$. Appendix \ref{section: additional simulation results} provides quantile-quantile plots of the simulated randomization distributions to further illustrate the quality of the normal approximations. Testing based on the normal asymptotic approximation controls size effectively, staying close to the nominal 5\% level (see Table \ref{table: null rejection prob, total average}).

\begin{table}[htbp!]
    \begin{minipage}{0.5\textwidth}
        \centering
            \begin{tabular}{r l || c c c }
                 & & \multicolumn{3}{c}{$p(w)$} \\
                 & & $0.25$ & $0.5$ & $0.75$ \\
                 \hline \hline
                 \multirow{3}{*}{$\phi$} & $0.25$ & $0.050$ & $0.047$ & $0.048$\\
                 & $0.5$ & $0.052$ & $0.052$ & $0.050$ \\
                 & $0.75$ & $0.050$ & $0.049$ & $0.048$
            \end{tabular}
            \subcaption{$\epsilon_{i, t} \sim N(0, 1)$, $N = 100$, $T = 10$}
        \end{minipage}
        \begin{minipage}{0.5\textwidth}
            \centering
            \begin{tabular}{r l || c c c}
                & & \multicolumn{3}{c}{$p(w)$} \\
                 & & $0.25$ & $0.5$ & $0.75$ \\
                 \hline \hline
                 \multirow{3}{*}{$\phi$} & 0.25 & $0.028$ & $0.029$ & $0.032$ \\
                 & $0.5$ & $0.046$ & $0.039$ & $0.044$ \\
                 & $0.75$ & $0.055$ & $0.044$ & $0.054$ \\
            \end{tabular}
            \subcaption{$\epsilon_{i, t} \sim Cauchy$, $N = 500$, $T = 100$}
        \end{minipage}
        \caption{Null rejection rate for the test of the null hypothesis $H_0: \bar{\tau}(1, 0; 0) = 0$ based upon the normal asymptotic approximation to the randomization distribution of $\hat{\bar{\tau}}(1, 0; 0)$. Panel (a) reports the null rejection probabilities in simulations with $\epsilon_{i, t} \sim N(0, 1)$ and $N = 100$, $T = 10$. Panel (b) reports the null rejection probabilities in simulations with $\epsilon_{i, t} \sim Cauchy$ and $N = 500$, $T = 100$. Results are computed over 5,000 simulations.}
        \label{table: null rejection prob, total average}
\end{table}

\subsection{Rejection rate}
Focusing on simulations with normally distributed errors, we next investigate the rejection rate of statistical tests based on the normal asymptotic approximations. To do so, we generate potential outcomes ${\textbf Y_{1:N, 1:T}}$ under different values of $\beta$, which governs the magnitude of the contemporaneous causal effect. As we vary $\beta = \{-1, -0.9, \hdots, 0.9, 1\}$, we also vary the parameter $\phi \in \{0.25, 0.5, 0.75\}$ and probability of treatment $p(w)\in \{0.25,0.5,0.75\}$ to investigate how rejection varies across a range of parameter values. We report the fraction of tests that reject the null hypothesis of zero average dynamic causal effects.

Figure \ref{fig: total average power plot} plots rejection rate curves against the weak null hypotheses $H_0: \bar{\tau}(1, 0; 0) = 0$ and $H_0: \bar{\tau}^{\dagger}(1, 0; 1) = 0$ as the parameter $\beta$ varies for different choices of the parameter $\phi$ and treatment probability $p(w)$. The rejection rate against $H_0: \bar{\tau}(1, 0; 0) = 0$ quickly converges to one as $\beta$ moves away from zero across a range of simulations, indicating that the conservative variance bound still leads to informative tests. When $\phi = 0.25$, the rejection rate against $H_0: \bar{\tau}^{\dagger}(1, 0; 1) = 0$ is relatively low -- lower values of $\phi$ imply less persistence in the causal effects across periods. When $\phi = 0.75$, there is substantial persistence in the causal effects across periods and we observe that the rejection rate curves looks similar. 

\begin{figure}[htbb!]
    \centering
    \includegraphics[width=5in, height=2.5in]{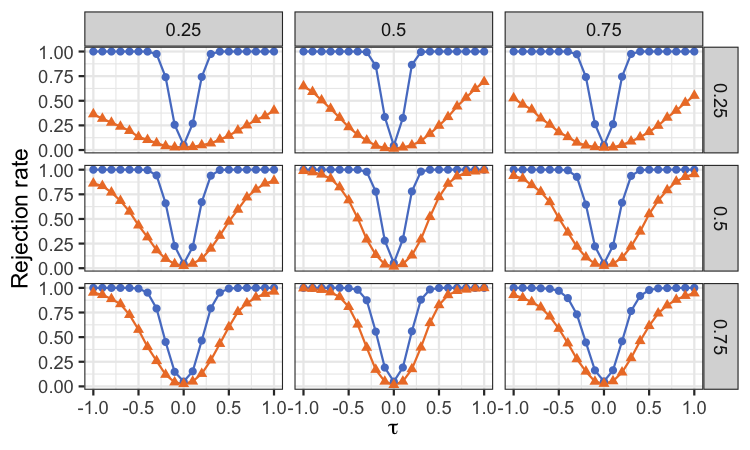}
    \caption{Rejection probabilities for a test of the null hypothesis $H_0: \bar{\tau}(1, 0;0) = 0$ and $H_0: \bar{\tau}^{\dagger}(1, 0; 1) = 0$ as the parameter $\beta$ varies under different choices of the parameter $\phi$ and treatment probability $p(w)$. The rejection rate curve against $H_0: \bar{\tau}(1,0;0) = 0$ is plotted in blue and the rejection rate curve against $H_0: \bar{\tau}^{\dagger}(1, 0; 1) = 0$ is plotted in orange. The rows index the parameter $\phi \in \{0.25, 0.5, 0.75\}$. The columns index the treatment probability $p(w) \in \{0.25, 0.5, 0.75\}$. The simulations are conducted with normally distributed errors $\epsilon_{i, t} \sim N(0, 1)$ and $N = 100$, $T = 10$. Results are averaged over $5000$ simulations.}
    \label{fig: total average power plot}
\end{figure}

Appendix \ref{section: additional simulation results} analyzes the rejection rate curves against the weak null hypothesis on the time-$t$ average dynamic causal effects with $N = 100$ units and the unit-$i$ average dynamic causal effect with $T = 100$ time periods. The conservative tests can have low power against these unit-specific or time period-specific weak null hypotheses in small experiments with few units or few time periods. Unless researchers are analyzing a panel experiment with a large cross-sectional or time dimension, we recommend that researchers focus on analyzing total lag-$p$ dynamic causal effects, which enables them to improve power by pooling information across both units and time periods.

\section{Empirical application in experimental economics}\label{sect: empirical illustration}
We apply our methods to reanalyze a panel experiment from \cite{AndreoniSamuelson(06)} that tests a game-theoretic model of ``rational cooperation'' and  studied how variation in the payoff structure of a two-player, twice-played prisoners' dilemma affects the choices of players. 

The payoffs of the game were determined by two parameters $x_1, x_2 \geq 0$ such that $x_1 + x_2 = 10$. In each period, both players simultaneously select either $C$ (cooperate) or $D$ (defect) and subsequently received the payoffs associated with these choices. Table \ref{tab: andreoni-samuelson stage game} summarizes the payoff structure. Let $\lambda = \frac{x_2}{x_1 + x_2} \in [0,1]$ govern the relative payoffs between the two periods of the prisoners' dilemma; when $\lambda = 0$, all payoffs occurred in period one and when $\lambda = 1$, all payoffs occurred in period two. The authors develop a model of rational cooperation that predicts when $\lambda$ is large, players will cooperate more often in period one compared to when $\lambda$ is small.  

\begin{table}[!htbp]
    \centering
    \begin{tabular}{c c c}
         \arrvline{} & $C$ & $D$ \\
         \hline \hline
        $C$  \arrvline{} & $(3x_1, 3x_1)$ & $(0, 4x_1)$  \\
        $D$  \arrvline{} & $(4x_1, 0)$  & $(x_1, x_1)$ \\ \\
        \multicolumn{3}{c}{Period one}
        \end{tabular}\quad
        \begin{tabular}{c c c}
        \arrvline{} & $C$ & $D$ \\
        \hline \hline
        $C$  \arrvline{} & $(3x_2, 3x_2)$  & $(0, 4x_2)$  \\
        $D$  \arrvline{} & $(4x_2, 0)$  & $(x_2, x_2)$ \\ \\
        \multicolumn{3}{c}{Period two}
        \end{tabular}  
    \caption{Stage games from twice-played prisoners' dilemma in the experiment conducted by \cite{AndreoniSamuelson(06)}, where the parameters satisfy $x_1, x_2 \geq 0$, $x_1 + x_2 = 10$ and $\lambda = \frac{x_1}{x_1 + x_2}$. The choice $C$ denotes ``cooperate'' and the choice $D$ ``defect.''}
    \label{tab: andreoni-samuelson stage game}
\end{table}

To investigate this hypothesis, \cite{AndreoniSamuelson(06)} conducted a panel-based experiment. In each session of the experiment, 22 subjects were recruited to play 20 rounds of the twice-played prisoners' dilemma in Table \ref{tab: andreoni-samuelson stage game}. In each round, participants were randomly matched into pairs, and each pair was then randomly assigned $\lambda \in \{0, 0.1, \hdots, 0.9, 1\}$ with equal probability. The authors conducted the experiment over five sessions for a total sample of 110 participants and we observe 2200 choices total. 

\begin{table}[t]
    \centering
    \begin{tabular}{r || c c | c }
         & \multicolumn{2}{c}{Counts} & \\
         & $0$ & $1$ & Mean \\
         \hline \hline
         Observed treatment, $W_{i,t}$ & $1136$ & $1064$ & $0.484$ \\
         Observed outcome, $Y_{i,t}$ & $521$ & $1679$ & $0.763$
    \end{tabular}
    \caption{Summary statistics for the experiment in \cite{AndreoniSamuelson(06)}. The treatment $W_{i, t}$ equals one when the assigned value of $\lambda$ is larger than $0.6$. The outcome $Y_{i, t}$ equals one whenever the participant cooperates in period one of the twice-repeated prisoners' dilemma. There are 110 participants and we observe 2220 choices total.}
    \label{tab: summary statistics}
\end{table}

This panel experiment is a natural application of our methods. The sequential nature of the experiment raises the possibility that past assignments may impact future actions as participants learn about the structure of the game over time. For example, random variation in the payoff structure may induce players to explore the strategy space. Additionally, the authors originally analyzed the experiment using regression models with unit-level fixed effects, which may be biased in the presence of dynamic causal effects even if the potential outcomes are linear as discussed in Section \ref{section:estimation-linear}.

In our analysis, the outcome of interest $Y$ is an indicator that equals one whenever the participant cooperated in period one of the stage game, $N = 110$, and $T = 20$. The assignment $W \in \mathcal{W} = \{0,1\}$ is binary and equals one whenever the assigned value $\lambda$ is greater than $0.6$, meaning that the payoffs are more concentrated in period two than period one of the stage game. We binarize the assignment in this manner to keep its cardinality (and therefore the number of possible assignment paths) manageable, while continuing to test the authors' core prediction on cooperative behavior. For a given pair of subjects, the assignment mechanism is Bernoulli with probability $p = 5/11$ for treatment and $p=6/11$ for control.\footnote{One potential complication that may arise from the subjects playing against each other in the stage game is possible spillovers or interference across units. The impact of such spillovers is, however, unlikely to be substantial as the matches are anonymous, and no players play each other more than once. We ignore this concern in our analysis.} Table \ref{tab: summary statistics} summarizes the observed assignments and observed outcomes in the experiment.

\subsection{Inference on total lag-$p$ weighted average dynamic causal effects}

We analyze the total lag-$p$ weighted average causal effect $\bar{\tau}^{\dagger}(1, 0; p)$ for $p = 0, 1, 2, 3$, which pools information across all units and time periods to investigate dynamic causal effects.\footnote{Appendix \ref{section: additional empirical results} investigates unit-specific and period-specific weighted average lag-$p$ dynamic causal effects. Since there are only $N = 110$ units and $T = 20$ periods in the experiment, these estimates are noisier than our estimates of the total lag-$p$ weighted average dynamic causal effects.} Based on the conservative test in Section \ref{section: testing}, the weak null hypothesis $\bar{\tau}^{\dagger}(1, 0; 0) = 0$ can be soundly rejected, indicating that the treatment has a positive contemporaneous effect on cooperation in period one of the stage game and confirming the hypothesis of \cite{AndreoniSamuelson(06)}. Table \ref{tab: pvalues for total plag, cooperation} summarizes these estimates of the total lag-$p$ weighted average causal effects. Interestingly, the point estimates are positive at $p = 1, 2, 3$, suggesting there may be dynamic causal effects on cooperative behavior across rounds of the twice-repeated prisoners' dilemma. For example, the treatment may induce participants to learn about the value of cooperation, thereby producing persistent effects.

\begin{table}[h]
    \centering
    \begin{tabular}{l || c c c c}
         &  \multicolumn{4}{c}{lag-$p$} \\
         & $0$ & $1$ & $2$ & $3$ \\ 
         \hline \hline
        Point estimate, $\hat{\bar{\tau}}^{\dagger}(1,0;p)$ & $0.285$ & $0.058$ & $0.134$ & $0.089$ \\
        Conservative p-value & $0.000$ & $0.226$ & $0.013$ & $0.126$ \\
        Randomization p-value & $0.000$ & $0.263$ & $0.012$ & $0.114$ \\
    \end{tabular}
    \caption{Estimates of the total lag-$p$ weighted average dynamic causal effect for $p = 0, 1, 2, 3$. The conservative p-value reports the p-value associated with testing the weak null hypothesis of no average dynamic causal effects, $H_0: \bar{\tau}^{\dagger}(1,0;p) = 0$, using the conservative estimator of the asymptotic variance of the nonparametric estimator (Theorem \ref{thm:clts}). The randomization p-value reports the p-value associated with randomization test of the sharp null of dynamic causal effects, $H_0: \tau_{i, t}(w, \tilde w;p) = 0$ for all $i \in [N], t \in [T]$. The randomization p-values are constructed based on 10,000 draws.}
    \label{tab: pvalues for total plag, cooperation}
\end{table}

We further investigate these results using randomization tests based on the sharp null of no dynamic causal effects. We construct the randomization distribution for the nonparametric estimator of the total lag-$p$ weighted average dynamic causal effect $\hat{\bar{\tau}}^{\dagger}(1, 0; p)$ for $p = 0, 1, 2, 3$ under the sharp null hypothesis of no lag-$p$ dynamic dynamical causal effects for all units and time periods; $H_0: \tau_{i, t}(w, \tilde w;p) = 0$ for all $i \in [N]$, $t \in [T]$.\footnote{When simulating the randomization distribution, we redraw assignment paths in a manner that respects the realized pairs of subjects in the experiment, meaning that subjects that are paired in the same round receive the same assignment.} Table \ref{tab: pvalues for total plag, cooperation} summarizes randomization p-values for the total lag-$p$ weighted average causal effects.  The p-value for the randomization test at $p = 0$ is approximately zero, strongly rejecting the sharp null of no contemporaneous dynamic causal effects for all units and again confirming the hypothesis of \cite{AndreoniSamuelson(06)}. 


\section{Conclusion}\label{sect:conclusion}
This paper developed a potential outcome model for studying dynamic causal effects in a panel experiment. We defined new panel-based dynamic causal estimands such as the lag-$p$ dynamic causal effect and introduced an associated nonparametric estimator. Our proposed estimator is unbiased for lag-$p$ dynamic causal effects over the randomization distribution, and we derived its finite population asymptotic distribution. We developed tools for inference on these dynamic causal effects -- a conservative test for weak nulls and an exact randomization test for sharp nulls. We showed that the linear unit fixed effects estimator and two-way fixed effects estimator are asymptotically biased for the contemporaneous causal effects in the presence of dynamic causal effects and persistence in the assignment mechanism. Finally, we illustrated our results through a simulation study and analyzed a panel experiment on rational cooperation in games.

\clearpage
\newpage
\singlespacing
\bibliographystyle{aea}
\bibliography{Bibliography.bib}

@STRING{E="Econometrica"}

@STRING{American="The American Statistican"}

@STRING{Annals="Annals of Statistics"}

@article{lillie2011n,
  title={The n-of-1 clinical trial: the ultimate strategy for individualizing medicine?},
  author={Lillie, Elizabeth O and Patay, Bradley and Diamant, Joel and Issell, Brian and Topol, Eric J and Schork, Nicholas J},
  journal={Personalized Medicine},
  volume={8},
  number={2},
  pages={161--173},
  year={2011},
  publisher={Future Medicine}
}

@article{kastelman2018switchback,
  title={Switchback tests and randomized experimentation under network effects at DoorDash},
  author={Kastelman, David and Ramesh, Raghav},
  journal={URL: https://medium.com/@DoorDash/switchback-tests-and-randomized-experimentation-under-network-effects-at-doordash-f1d938ab7c2a},
  year={2018}
}

@article{Sneider2019switchback,
  title={Experiment Rigor for Switchback Experiment Analysis},
  author={Sneider, Carla and Tang, Yixin},
  journal={URL: https://doordash.engineering/2019/02/20/experiment-rigor-for-switchback-experiment-analysis/},
  year={2018}
}

@article{aronow2017estimating,
  title={Estimating average causal effects under general interference, with application to a social network experiment},
  author={Aronow, Peter M and Samii, Cyrus},
  journal={The Annals of Applied Statistics},
  volume={11},
  number={4},
  pages={1912--1947},
  year={2017},
  publisher={Institute of Mathematical Statistics}
}

@article{chin2018central,
  title={Central limit theorems via Stein's method for randomized experiments under interference},
  author={Chin, Alex},
  journal={arXiv preprint arXiv:1804.03105},
  year={2018}
}

@article{van2008construction,
  title={The construction and analysis of adaptive group sequential designs},
  author={van der Laan, Mark J},
  year={2008},
  publisher={bepress}
}

@inproceedings{li2010contextual,
  title={A contextual-bandit approach to personalized news article recommendation},
  author={Li, Lihong and Chu, Wei and Langford, John and Schapire, Robert E},
  booktitle={Proceedings of the 19th international conference on World wide web},
  pages={661--670},
  year={2010}
}

@article{zhang2020inference,
  title={Inference for Batched Bandits},
  author={Zhang, Kelly W and Janson, Lucas and Murphy, Susan A},
  journal={arXiv preprint arXiv:2002.03217},
  year={2020}
}

@inproceedings{li2016collaborative,
  title={Collaborative filtering bandits},
  author={Li, Shuai and Karatzoglou, Alexandros and Gentile, Claudio},
  booktitle={Proceedings of the 39th International ACM SIGIR conference on Research and Development in Information Retrieval},
  pages={539--548},
  year={2016}
}

@article{lai1985asymptotically,
  title={Asymptotically efficient adaptive allocation rules},
  author={Lai, Tze Leung and Robbins, Herbert},
  journal={Advances in Applied Mathematics},
  volume={6},
  number={1},
  pages={4--22},
  year={1985},
  publisher={Academic Press}
}

@article{robbins1952some,
  title={Some aspects of the sequential design of experiments},
  author={Robbins, Herbert},
  journal={Bulletin of the American Mathematical Society},
  volume={58},
  number={5},
  pages={527--535},
  year={1952},
  publisher={Citeseer}
}

@incollection{cunha_chapter_2006,
	title = {Chapter 12 {Interpreting} the {Evidence} on {Life} {Cycle} {Skill} {Formation}},
	volume = {1},
	url = {http://www.sciencedirect.com/science/article/pii/S1574069206010129},
	urldate = {2019-07-23},
	booktitle = {Handbook of the {Economics} of {Education}},
	publisher = {Elsevier},
	author = {Cunha, Flavio and Heckman, James J. and Lochner, Lance and Masterov, Dimitriy V.},
	editor = {Hanushek, E. and Welch, F.},
	month = jan,
	year = {2006},
	doi = {10.1016/S1574-0692(06)01012-9},
	pages = {697--812}
}

@article{cunha_estimating_2010,
	title = {Estimating the {Technology} of {Cognitive} and {Noncognitive} {Skill} {Formation}},
	volume = {78},
	copyright = {© 2010 The Econometric Society},
	issn = {1468-0262},
	url = {https://onlinelibrary.wiley.com/doi/abs/10.3982/ECTA6551},
	doi = {10.3982/ECTA6551},
	language = {en},
	number = {3},
	urldate = {2019-07-23},
	journal = {Econometrica},
	author = {Cunha, Flavio and Heckman, James J. and Schennach, Susanne M.},
	year = {2010},
	pages = {883--931}
}

@article{ben-porath_production_1967,
	title = {The {Production} of {Human} {Capital} and the {Life} {Cycle} of {Earnings}},
	volume = {75},
	issn = {0022-3808},
	url = {https://www.jstor.org/stable/1828596},
	urldate = {2019-07-23},
	journal = {Journal of Political Economy},
	author = {Ben-Porath, Yoram},
	year = {1967},
	pages = {352--365}
}

@article{griliches_estimating_1977,
	title = {Estimating the {Returns} to {Schooling}: {Some} {Econometric} {Problems}},
	volume = {45},
	issn = {0012-9682},
	shorttitle = {Estimating the {Returns} to {Schooling}},
	url = {https://www.jstor.org/stable/1913285},
	doi = {10.2307/1913285},
	urldate = {2019-07-23},
	journal = {Econometrica},
	author = {Griliches, Zvi},
	year = {1977},
	pages = {1--22}
}

\clearpage  
\appendix

\appendixpagenumbering

\renewcommand{\thefigure}{A\arabic{figure}}
\setcounter{figure}{0}

\renewcommand{\thetable}{A\arabic{table}}
\setcounter{table}{0}

\clearpage
\begin{center}
\vspace{-1.8cm}{\Large \textbf{Panel-Based Experiments and Dynamic Causal Effects:\\ A Finite Population Perspective}}\medskip \\
	\Large \textbf{Online Appendix} \bigskip \\
\large Iavor Bojinov \hspace{0.3cm} Ashesh Rambachan \hspace{0.3cm} Neil Shephard
\end{center}

\section{Proofs of main results}

\subsection*{Proof of Theorem \ref{thm:thmMDs}}

We begin the proof with a Lemma that will be used later on.

\begin{lemma}\label{ref:thmZ1}
  Assume a potential outcome panel with an assignment mechanism that is individualistic (Definition \ref{ass:individualistic}) and probabilistic (Assumption \ref{ass:probtreat1}). Define, for any ${\mathbf w} \in \mathcal{W}^{(p+1)}$, the random function
$
Z_{i,t-p:t}({\mathbf w}) := p_{i,t-p}({\mathbf w})^{-1} \mathbbm{1}\{W_{i,t-p:t}={\mathbf w}\}
$. Then, over the assignment mechanism, $\E(Z_{i,t-p:t}({\mathbf w})|\mathcal{F}_{i,t-p-1}) = 1$ and $Var(Z_{i,t-p:t}({\mathbf w})|\mathcal{F}_{i,t-p-1}) = p_{i,t-p}({\mathbf w})^{-1}(1 - p_{i,t-p}({\mathbf w})),$ and 
$Cov(Z_{i,t-p:t}({\mathbf w}),Z_{i,t-p:t}(\tilde{{\mathbf w}})|\mathcal{F}_{i,t-p-1}) = -1$ for all ${\mathbf w} \ne \tilde{{\mathbf w}}$.
  Moreover, $Z_{i,t-p:t}({\mathbf w})$ and $Z_{j,t-p:t}({\mathbf w})$ are, conditioning on $\mathcal{F}_{1:N,t-p-1}$, independent for $i \ne j$.
    \begin{proof}
    The expectation is by construction, the variance comes from the variance of a Bernoulli trial.  The conditional independence is by the individualistic assignment assumption.  
    \end{proof}
\end{lemma} 

For any ${\mathbf w},\tilde{{\mathbf w}} \in \mathcal{W}^{(p+1)}$, let $u_{i,t-p}({\mathbf w},\tilde{{\mathbf w}};p) = \hat{\tau}_{i,t}({\mathbf w},\tilde{{\mathbf w}};p) - \tau_{i,t}({\mathbf w},\tilde{{\mathbf w}};p)$ be the estimation error. Now 
$$u_{i,t-p}({\mathbf w},\tilde{{\mathbf w}};p) = Y_{i,t}(w_{i,1:t-p-1}^{obs},{\mathbf w})(Z_{i,t-p:t}({\mathbf w})-1) 
- Y_{i,t}(w_{i,1:t-p-1}^{obs},\tilde{\mathbf w})(Z_{i,t-p:t}(\tilde{\mathbf w})-1).$$

Hence the conditional expectation is zero by Lemma \ref{ref:thmZ1}. Then,

\begin{align*}
&Var(u_{i,t-p}({\mathbf w},\tilde{{\mathbf w}};p)|\mathcal{F}_{i,t-p-1}) = Y_{i,t}(w_{i,1:t-p-1}^{obs},{\mathbf w})^2Var(Z_{i,t-p:t}({\mathbf w})|\mathcal{F}_{i,t-p-1}) \\
&+ Y_{i,t}(w_{i,1:t-p-1}^{obs},\tilde{{\mathbf w}})^2Var(Z_{i,t-p:t}(\tilde{{\mathbf w}})|\mathcal{F}_{i,t-p-1}) \\
& -2 Y_{i,t}(w_{i,1:t-p-1}^{obs},{\mathbf w}) Y_{i,t}(w_{i,1:t-p-1}^{obs},\tilde{{\mathbf w}})
Cov(Z_{i,t-p:t}(\tilde{{\mathbf w}}),Z_{i,t-p:t}(\tilde{{\mathbf w}}|\mathcal{F}_{i,t-p-1}) \\
& = Y_{i,t}(w_{i,1:t-p-1}^{obs},{\mathbf w})^2 p_{i,t-p}({\mathbf w})^{-1}(1 - p_{i,t-p}({\mathbf w})) \\
& + Y_{i,t}(w_{i,1:t-p-1}^{obs},\tilde{{\mathbf w}})^2 p_{i,t-p}(\tilde{{\mathbf w}})^{-1}(1 - p_{i,t-p}(\tilde{{\mathbf w}})) \\
& - 2 Y_{i,t}(w_{i,1:t-p-1}^{obs},{\mathbf w}) Y_{i,t}(w_{i,1:t-p-1}^{obs},\tilde{{\mathbf w}}). 
\end{align*}

Simplifying gives the result on the variance of the estimation error. Then,
\begin{align*}
&Cov(u_{i,t-p}({\mathbf w},\tilde{\mathbf w};p),
u_{i,t-p}(\bar{\mathbf w},\hat{\mathbf w};p)|\mathcal{F}_{i,t-p-1}) \\
&= Y_{i,t}(w_{i,1:t-p-1}^{obs},{\mathbf w}) Y_{i,t}({w}_{i,1:t-p-1}^{obs},\bar{\mathbf w}) Cov(Z_{i,t-p:t}({\mathbf w}),Z_{i,t-p:t}(\bar{\mathbf w})|\mathcal{F}_{i,t-p-1}) \\
&- Y_{i,t}(w_{i,1:t-p-1}^{obs},{\mathbf w}) Y_{i,t}({w}_{i,1:t-p-1}^{obs},\hat{\mathbf w}) Cov(Z_{i,t-p:t}({\mathbf w}),Z_{i,t-p:t}(\hat{\mathbf w})|\mathcal{F}_{i,t-p-1}) \\
&- Y_{i,t}(w_{i,1:t-p-1}^{obs},\tilde{\mathbf w}) Y_{i,t}({w}_{i,1:t-p-1}^{obs},\bar{\mathbf w}) Cov(Z_{i,t-p:t}(\tilde{\mathbf w}),Z_{i,t-p:t}(\bar{\mathbf w})|\mathcal{F}_{i,t-p-1}) \\
& Y_{i,t}(w_{i,1:t-p-1}^{obs},\tilde{\mathbf w}) Y_{i,t}({w}_{i,1:t-p-1}^{obs},\hat{\mathbf w}) Cov(Z_{i,t-p:t}(\tilde{\mathbf w}),Z_{i,t-p:t}(\hat{\mathbf w})|\mathcal{F}_{i,t-p-1}) \\
& = -Y_{i,t}(w_{i,1:t-p-1}^{obs},{\mathbf w}) Y_{i,t}({w}_{i,1:t-p-1}^{obs},\bar{\mathbf w})  
+ Y_{i,t}(w_{i,1:t-p-1}^{obs},{\mathbf w}) Y_{i,t}({w}_{i,1:t-p-1}^{obs},\hat{\mathbf w})  \\
&+ Y_{i,t}(w_{i,1:t-p-1}^{obs},\tilde{\mathbf w}) Y_{i,t}({w}_{i,1:t-p-1}^{obs},\bar{\mathbf w})  
- Y_{i,t}(w_{i,1:t-p-1}^{obs},\tilde{\mathbf w}) Y_{i,t}({w}_{i,1:t-p-1}^{obs},\hat{\mathbf w})\\ 
\end{align*}
Finally, conditional independence of the errors follows due to the individualistic assignment of treatments. $\Box$ 

\subsection*{Proof of Proposition \ref{prop:varbounds}}

The proof of this result is analogous to the proof of Theorem \ref{thm:thmMDs}. We state the analogue of Lemma \ref{ref:thmZ1} for completeness.

\begin{lemma}\label{lemma:zforvariance}
    Assume a potential outcome panel with an assignment mechanism that is individualistic (Definition \ref{ass:individualistic}) and probabilistic (Assumption \ref{ass:probtreat1}). Define, for any ${\mathbf w} \in \mathcal{W}^{(p+1)}$, the random function $V_{i,t-p:t}({\mathbf w}) := p_{i,t-p}({\mathbf w})^{-2} \mathbbm{1}\{W_{i,t-p:t}={\mathbf w}\}$. Then, over the assignment mechanism, $\E(V_{i,t-p:t}({\mathbf w})|\mathcal{F}_{i,t-p-1}) = p_{i,t-p}({\mathbf w})^{-1}$ and $Var(V_{i,t-p:t}({\mathbf w})|\mathcal{F}_{i,t-p-1}) = p_{i,t-p}({\mathbf w})^{-3}(1 - p_{i,t-p}({\mathbf w})),$ and $Cov(V_{i,t-p:t}({\mathbf w}),V_{i,t-p:t}(\tilde{{\mathbf w}})|\mathcal{F}_{i,t-p-1}) = p_{i,t-p}({\mathbf w})^{-1} p_{i,t-p}(\tilde {\mathbf w})^{-1}$ for all ${\mathbf w} \ne \tilde{{\mathbf w}}$. Moreover, $V_{i,t-p:t}({\mathbf w})$ and $V_{j,t-p:t}({\mathbf w})$ are, conditioning on $\mathcal{F}_{1:N,t-p-1}$, independent for $i \ne j$.
\end{lemma}

For any ${\mathbf w},\tilde{{\mathbf w}} \in \mathcal{W}^{(p+1)}$, let $v_{i,t-p}({\mathbf w},\tilde{{\mathbf w}};p) = \hat{\gamma}_{i,t}^2({\mathbf w},\tilde{{\mathbf w}};p) - \gamma_{i,t}^2({\mathbf w},\tilde{{\mathbf w}};p)$ be the estimation error. Now 
$$v_{i,t-p}({\mathbf w},\tilde{{\mathbf w}};p) = Y_{i,t}(w_{i,1:t-p-1}^{obs},{\mathbf w})^2 (V_{i,t-p:t}({\mathbf w}) - p_{i,t-p}({\mathbf w})^{-1}) + Y_{i,t}(w_{i,1:t-p-1}^{obs},\tilde{\mathbf w})^2 ( V_{i,t-p:t}(\tilde{\mathbf w}) - \tilde p_{i,t-p}({\mathbf w})^{-1}).$$
Therefore, the conditional expectation is zero by Lemma \ref{lemma:zforvariance}. The conditional independence of the errors follows due to the individualistic assignment of the treatments. $\Box$

\subsection*{Proof of Theorem \ref{thm:clts}}

Only the third results requires a new proof. The first result is a reinterpretation of the classic cross-sectional result using a triangular array central limit theorem, for the usual Lindeberg condition must hold due to the bounded potential outcomes and the treatments being probabilistic. The second result follows from results in \cite{BojinovShephard(19)}, who use a martingale difference array central limit theorem.  

The third result, which holds for $NT$ going to infinity, can be split into three parts.  For $NT$ to go to infinity we must have either: (i) $T$ goes to infinity with $N$ finite, (ii) $N$ goes to infinity with $T$ finite, or (iii) both $N$ and $T$ go to infinity.  In the case (i), we apply the martingale difference CLT but now we have preaveraged the cross-sectional errors over the $N$ terms for each time period. The preaverage is still a martingale difference, so the technology is the same. In the case (ii) we preaverage the time aspect. Then we are back to a standard triangular array CLT. As we have both (i) and (ii), then (iii) must hold. $\Box$

\subsection*{Proof of Proposition \ref{lemma: cons_var}}
The unbiasedness statements follow directly from Proposition \ref{prop:varbounds}. The proofs of the consistency statements are analogous to the proof of Theorem \ref{thm:clts}. The first result follows from an application of the triangular array law of law of large numbers, which may be applied due to the bounded potential outcomes and the treatments being probabilistic. The second statement follows from an application of a martingale difference sequence law of large numbers (Theorem 2.13 in \cite{HallHeyde(80)}). The third statement can be again proved in three cases: (i) $T$ goes to infinity with $N$ finite, (ii) $N$ goes to infinity with $T$ finite, or (iii) both $N$ and $T$ go to infinity as in the proof of Theorem \ref{thm:clts} and applying the appropriate law of large numbers. $\Box$

\subsection*{Proof of Proposition \ref{prop:unit-fe-additive-general-case}}

Begin by writing the observed outcomes as 
            \begin{align*}
                Y_{i, t} = Y_{i, t}(\textbf{0}) + \sum_{s=1}^{t} \beta_{i, t, t-s} W_{i, s}.
            \end{align*}
        Similarly, write $\bar{Y}_{i \cdot} = \bar{Y}_{i\cdot}(\textbf{0}) + \overline{\beta W}_{i\cdot}$, where  $\overline{\beta W}_{i\cdot} = \frac{1}{T} \sumt \sum_{s=1}^{t} \beta_{i, t, t-s} W_{i, s}$. The transformed outcome can be then written as
            \begin{align*}
                \widecheck Y_{i, t} = \sum_{s=1}^{t} \beta_{i, t, t-s} W_{i, s} - \overline{\beta W}_{i\cdot} + \widecheck{Y}_{i,t}({\bf 0}).
            \end{align*}
        Consider the numerator of the unit fixed effects estimator. Substituting in, we arrive at
            \begin{align*}
                & \frac{1}{NT} \sumi \sumt \widecheck Y_{i, t} \widecheck W_{i, t} = \frac{1}{NT} \sumi \sumt \beta_{i, t, 0} W_{i, t} \widecheck W_{i, t} + \frac{1}{NT} \sumi \sumt \left( \sum_{s=1}^{t-1} \beta_{i, t, t-s} W_{i, s} \widecheck W_{i, t} \right) + \frac{1}{NT} \sumi \sumt \widecheck{Y}_{i,t}({\bf 0}) \widecheck W_{i, t} \\
                & =  \frac{1}{T} \sumt \left( \frac{1}{N} \sumi \beta_{i, t, 0} W_{i, t} \widecheck W_{i, t} \right) + \frac{1}{T} \sumt \sum_{s=1}^{t-1} \left( \frac{1}{N} \sumi \beta_{i, t, t-s} W_{i, s} \widecheck W_{i, t} \right) + \frac{1}{T} \sumt \left( \frac{1}{N} \sumi \widecheck{Y}_{i,t}({\bf 0}) \widecheck W_{i,t} \right).
            \end{align*}
        Therefore, for fixed $T$ as $N \rightarrow \infty$,
            \begin{align*}
                \frac{1}{T} \sumt \left( \frac{1}{N} \sumi \beta_{i, t, 0} W_{i, t} \widecheck W_{i, t} \right) &\xrightarrow{p} \frac{1}{T} \sumt \widecheck \kappa_{W, \beta, t, t}, \\
                \frac{1}{T} \sumt \sum_{s=1}^{t-1} \left( \frac{1}{N} \sumi \beta_{i, t, t-s} W_{i, s} \widecheck W_{i, t} \right) &\xrightarrow{p} \frac{1}{T} \sumt \sum_{s=1}^{t-1} \widecheck \kappa_{W, \beta, t, s}, \\
                 \frac{1}{T} \sumt \left( \frac{1}{N} \sumi \widecheck{Y}_{i,t}({\bf 0}) \widecheck W_i \right) &\xrightarrow{p} \frac{1}{T} \sumt \widecheck \delta_t.
            \end{align*}
        Similarly, the denominator converges to $\frac{1}{NT} \sumt \sumi \widecheck W_{i, t}^2 \xrightarrow{p} \frac{1}{T} \sumt \widecheck \sigma^2_{W, t}$. The result then follows by Slutsky. $\Box$

\subsection*{Proof of Proposition \ref{prop:two-way-fe-additive-general-case}}

Begin by writing 
            \begin{align*}
                Y_{i,t} = Y_{i,t}(\textbf{0}) + \sum_{s=1}^{t} \beta_{i, t, t-s} W_{i, s}.
            \end{align*}
Then, $\bar Y_{\cdot t} = \bar Y_{\cdot t}(\textbf{0}) + \overline{\beta W}_{\cdot t}$, $\bar Y_{i\cdot} = \bar{Y}_{i\cdot}(\textbf{0}) + \overline{\beta W}_{i\cdot}$ and $\bar Y = \bar Y(\textbf{0}) + \overline{\beta W}$. Therefore, 
            \begin{align*}
                \dot{\widecheck Y}_{i, t} = \dot{\widecheck Y}_{i, t}(\textbf{0}) + \left( \sum_{s=1}^{t} \beta_{i, t, t-s} W_{i, s} - \overline{\beta W} \right) - \left( \overline{\beta W}_{\cdot t} - \overline{\beta W} \right) - \left( \overline{\beta W}_{i\cdot} - \overline{\beta W} \right).
            \end{align*}
Consider the numerator of the unit fixed effects estimator. Substituting in,
            \begin{align*}
                \frac{1}{NT} \sumi \sumt \doubleY_{i,t} \doubleW_{i,t} = \frac{1}{NT} \sumi \sumt \beta_{i, t, 0} W_{i,t} \doubleW_{i,t} + \frac{1}{NT} \sumi \sumt \sum_{s=1}^{t-1} \beta_{i, t, t-s} W_{i, s} \doubleW_{i,t} + \frac{1}{NT} \sumi \sumt \dot{\widecheck Y}_{i, t}(\textbf{0}) \doubleW_{i,t} \\
                = \frac{1}{T} \sumt \left( \frac{1}{N} \sumi \beta_{i, t, 0} W_{i,t} \doubleW_{i,t} \right) + 
                \frac{1}{T} \sumt \left( \frac{1}{N} \sumi \sum_{s=1}^{t-1}  \beta_{i, t, t-s} W_{i, s} \doubleW_{i,t} \right) + \frac{1}{T} \sumt \left( \frac{1}{N} \sumi \dot{\widecheck Y}_{i, t}(\textbf{0}) \doubleW_{i,t} \right).
            \end{align*}
        Therefore, 
            \begin{align*}
                & \frac{1}{N} \sumi \beta_{i, t, 0} W_{i,t} \doubleW_{i,t} \xrightarrow{p} \dot{\widecheck \kappa}_{W, \beta, t, t}, \\
                & \frac{1}{N} \sumi \sum_{s=1}^{t-1}  \beta_{i, t, t-s} W_{i, s} \doubleW_{i,t} \xrightarrow{p} \sum_{s=1}^{t-1} \dot{\widecheck \kappa}_{W, \beta, t, s}, \\
                & \frac{1}{N} \sumi \dot{\widecheck Y}_{i, t}(\textbf{0}) \doubleW_{i,t} \xrightarrow{p} \dot{\widecheck \delta}_{t}.
            \end{align*}
        A similar argument applies to the denominator and the result follows. $\Box$
        
\section{Additional theoretical results}\label{section: additional theoretical results}

\subsection{Prediction decomposition of the adapted propensity score}
Recall the definition of the adapted propensity score in Section \ref{sect:non-param-estimation}
\begin{equation*}
    p_{i,t-p}({\mathbf w}) := \Pr(W_{i,t-p:t}={\mathbf w}|W_{i,1:t-p-1},Y_{i,1:t}(W_{i,1:t-p-1},{\mathbf w})).
\end{equation*}
\noindent The adapted propensity score can be decomposed using individualistic assignment (Definition \ref{ass:individualistic}) and the prediction decomposition. 

\begin{lemma}\label{lemma:pred} 
For a potential outcome panel satisfying individualistic assignment (Definition \ref{ass:individualistic}) and any ${\mathbf w} \in \mathcal{W}^{(p+1)}$, the adapted propensity score can be factorized as  
\begin{align*}
 p_{i,t-p}({\mathbf w}) = &\Pr(W_{i,t-p} = w_1| W_{i,1:t-p-1}, Y_{i,1:t-p-1}(W_{i,1:t-p-1})) \\
&\times \prod_{s=1}^{p} \Pr(W_{i,t-p+s}=w_{s+1}|
W_{i,1:t-p-1},W_{i,t-p:t-p+s-1}={\mathbf w}_{1:s},
Y_{i,1:t-p+s-1}(W_{i,1:t-p-1},{\mathbf w}_{1:s})).
\end{align*}
\end{lemma}
\begin{proof} Use the prediction decomposition for assignments, given all outcomes,  
\begin{align*}
 p_{i,t-p}({\mathbf w}) = &\Pr(W_{i,t-p} = w_1| W_{i,1:t-p-1}, Y_{i,1:t}(W_{i,1:t-p-1},\mathbf{w})) \\
&\times \prod_{s=1}^{p} \Pr(W_{i,t-p+s}=w_{s+1}|
W_{i,1:t-p-1},W_{i,t-p:t-p+s-1}={\mathbf w}_{1:s},
Y_{i,1:t}(W_{i,1:t-p-1},\mathbf{w})).
\end{align*}
and then simplify using the individualistic assignment of treatments.  
\end{proof}

Even though the assignment mechanism is known, we only observe the outcomes along the realized assignment path $Y_{i,1:t}(w_{i,1:t}^{obs})$, and so it is not possible to use Lemma \ref{lemma:pred} to compute $p_{i,t-p}({\mathbf w})$ for all assignment path. We can, however, compute the adapted propensity score along the observed assignment path, $p_{i,t-p}( w_{i,t-p:t}^{obs})$, since the associated outcomes are observed.

\subsection{Estimation as a repeated cross-section}
Denote $\dot Y_{1:N, t} = (\dot Y_{1,t},...,\dot Y_{N,t})^{\prime}$, $\dot W_{i,1:t} = (W_{i,t} - \bar{W}_{\cdot t}, W_{i,t-1} - \bar{W}_{\cdot t-1}, ..., W_{i,1} - \bar{W}_{\cdot 1})^{\prime}$ and $\dot W_{1:N, t} = (\dot W_{1,1:t},...,\dot W_{N,1:t})^{\prime}$. The least squares coefficient in the regression of $\dot Y_{1:N, t}$ on $\dot W_{1:N, t}$ is $\hat{{\boldsymbol \beta}}_{1:N, t} = (\dot W_{1:N,t}^{\prime} \dot W_{1:N, t})^{-1} \dot W_{1:N,t}^{\prime} \dot Y_{1:N,t}$. Proposition \ref{thm:estAdditive} derives the finite population limiting distribution of $\hat{{\boldsymbol \beta}}_{1:N, t}$ as the number of units grows large. 

\begin{proposition}\label{thm:estAdditive} 
    Assume a potential outcome panel and consider the ``control'' only path, for $0\in \cW$ let $\tilde w_{i,1:t}={\bf 0}$. Let $\dot \mu_{i, t}$ be the $t \times 1$ vector whose $u$-th element is $E\left[ \dot W_{i, t-(u-1)} \,|\, \cF_{1:N, 0, T} \right]$ and $\Omega_{i, t}$ be the $t \times t$ matrix whose $u,v$-th element is $Cov(\dot W_{i,t-(u-1)},\dot W_{i,t-(v-1)}|\mathcal{F}_{1:N,0,T})$. Additionally assume that: 
\begin{enumerate}
    \item The potential outcome panel is linear (Definitions \ref{defn:linearPanel}) and homogeneous with ${\boldsymbol \beta_{it}} \equiv {\boldsymbol \beta}_t = \left(\beta_{t, 0}, \hdots, \beta_{t, t-1} \right)$ for all $t$.
    
    \item $W_{i,1:t}$ is an individualistic stochastic assignment path and, over the randomization distribution, $Var(W_{i,t}|\mathcal{F}_{1:N,0,T})=\sigma^2_{W,i,t}<\infty$ for each $i \in [N]$, $t \in [T]$.
    
    \item As $N \rightarrow \infty$,
    \begin{enumerate} 
       \item Non-stochastically, $N^{-1} \sum_{i=1}^N \Omega_{i,t} \rightarrow \Gamma_{2,t},$ where $\Gamma_{2,t}$ is positive definite. 
    
       \item $N^{-1/2} \sum_{i=1}^N (\dot W_{i,1:t} - \dot\mu_{i,t}) \dot{Y}_{i,t}({\bf 0})|\mathcal{F}_{1:N,0,T} \xrightarrow{d} N(0,\Gamma_{1,t}).$
    
       \item Non-stochastically, $N^{-1} \sum_{i=1}^N \dot{Y}_{i,t}({\bf 0}) \dot \mu_{i,t}\rightarrow \dot{\delta}_t$.
    \end{enumerate} 
    \end{enumerate} 
    Then, over the randomization distribution, as $N \rightarrow \infty$, 
    $$
        \sqrt{N}(\hat{{\boldsymbol \beta}}_{1:N, t} - {\boldsymbol \beta}_{t} - \Gamma_{2,t}^{-1} \dot{\delta}_t) |\mathcal{F}_{1:N,0,T} \xrightarrow{d} N(0,\Gamma_{2,t}^{-1} \Gamma_{1,t} \Gamma_{2,t}^{-1}).
    $$
    \begin{proof}
Under linear potential outcomes,
$$
Y_{i,t}(W_{i,1:t}) - Y_{i,t}(\tilde W_{i,1:t}) = \sum_{s=0}^{t-1} \beta_{i,t,s} (W_{i,t-s} - \tilde W_{i,t-s}). 
$$
Focus on the counterfactual $\tilde W_{i,1:t}=\bf{0}$, then
\begin{equation*}
    Y_{i,t} = Y_{i,t}(W_{i,1:t}) = \bar Y_{\cdot t}({\bf 0}) + \sum_{s=0}^{t-1}\beta_{i,t,s} W_{i,t-s} + \dot{Y}_{i,t}({\bf 0}).
\end{equation*}

\noindent Therefore, the within-period transformed outcome equals
$$
\dot{Y}_{i,t} = Y_{i,t}-\bar Y_{\cdot t} = \sum_{s=0}^{t-1} \{ \beta_{i,t,s} W_{i,t-s} -  \frac{1}{N}\sum_{j=1}^N\beta_{j,t,s}W_{j,t-s} \}
+ \dot{Y}_{i,t}({\bf 0}).
$$
Further imposing homogeneity, it simplifies to
$$
\dot{Y}_{i,t} = \sum_{s=0}^{t-1} \{ \beta_{t,s} (W_{i,t-s} -  \frac{1}{N}\sum_{j=1}^N W_{j,t-s} \}
+ \dot{Y}_{i,t}({\bf 0}).
$$
Stacking everything across units, this becomes $\dot Y_{1:N, t} = \dot W_{1:N, t} \beta_t +  \dot{Y}_{1:N, t}({\bf 0})$, and so the linear projection coefficient is given by
$$
\hat{\beta}_{t} = (\dot W_{1:N, t} ^{\prime} \dot W_{1:N, t})^{-1} \dot W_{1:N, t} ^{\prime} \dot Y_{1:N, t} = \beta_{t}
+ (\dot W_{1:N, t} ^{\prime} \dot W_{1:N, t})^{-1} \dot W_{1:N, t} ^{\prime} \dot{Y}_{1:N, t}({\bf 0}). 
$$
The important unusual point here is that $\dot{Y}_{1:N, t}({\bf 0})$ is non-stochastic and that $\dot W_{1:N, t}$ is random, exactly the opposite of the case often discussed in the statistical analysis of linear regression. Now 
$$
\frac{1}{N} \dot W_{1:N, t}^{\prime} \dot W_{1:N, t} =  \frac{1}{N} \sum_{i=1}^N \dot W_{i,1:t} \dot W_{i,1:t}^{\prime},
$$
and 
$$
\frac{1}{N} \sum_{i=1}^N \dot W_{i,1:t} \dot{Y}_{i,t}({\bf 0}) =  \frac{1}{N} \sum_{i=1}^N (\dot W_{i,1:t}-\dot\mu_{i,t}) \dot{Y}_{i,t}({\bf 0}) + \frac{1}{N} \sum_{i=1}^N \mu_{i,t} \dot{Y}_{i,t}({\bf 0}).
$$
Then, under the assumption of individualistic assignment (Definition \ref{ass:individualistic}), 
$$
 \frac{1}{N}\sum_{i=1}^N \dot W_{i,1:t} \dot W_{i,1:t}^{\prime}|\mathcal{F}_{1:N,0,T} \xrightarrow{p} \Gamma_{2,t}, 
$$
recalling $\dot{Y}_{i,t}({\bf 0})$ is non-stochastic and applying Assumptions 3(b) and 3(c), then Slutsky's theorem delivers the result. $\Box$
    \end{proof}
\end{proposition}

\section{Additional simulation results}\label{section: additional simulation results}

\subsection{Additional simulations for the estimator of the total average dynamic causal effects}

\paragraph{Quantile-quantile plot for the normal approximation:} Figure \ref{fig: total qqplot plag0} provides quantile-quantile plots of the simulated randomization distribution for the estimator $\hat{\bar{\tau}}(1, 0; 0)$ presented in Section \ref{section:simulation} of the main text. 

\begin{figure}[htbp!]
    \centering
    \begin{subfigure}{.5\textwidth}
    \centering
    \includegraphics[width=2.5in, height=2.5in]{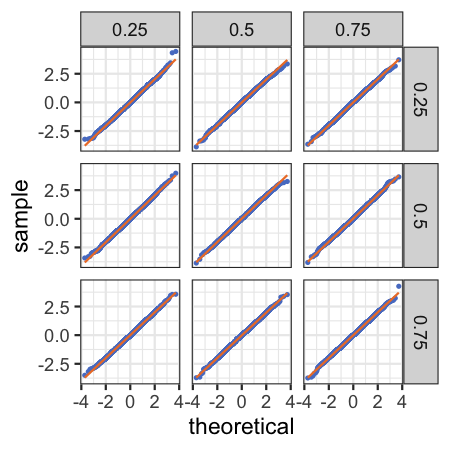}
    \caption{$\epsilon_{i, t} \sim N(0, 1)$, $N = 100$, $T = 10$}
    \end{subfigure}%
    \begin{subfigure}{.5\textwidth}
    \centering
    \includegraphics[width=2.5in, height=2.5in]{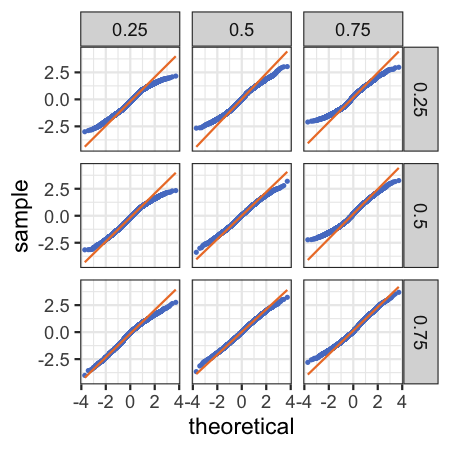}
    \caption{$\epsilon_{i, t} \sim Cauchy$, $N = 500$, $T = 100$}
    \end{subfigure}
    \caption{Quantile-quantile plots for the simulated randomization distribution for $\hat{\bar{\tau}}(1, 0; 0)$ under different choices of the parameter $\phi$ and treatment probability $p(w)$. The quantile-quantile plots compare the quantiles of the simulated randomization distribution (y-axis) against the quantiles of a standard normal random variable (x-axis). The 45 degree line is plotted in solid orange. The rows index the parameter $\phi \in \{0.25, 0.5, 0.75\}$, and the columns index the treatment probability $p(w) \in \{0.25, 0.5, 0.75\}$. Panel (a) plots the quantile-quantile plots for simulated randomization distribution with normally distributed errors $\epsilon_{i, t} \sim N(0, 1)$ and $N = 100, T = 10$. Panel (b) plots the quantile-quantile plots simulated randomization distribution with Cauchy distribution errors $\epsilon_{i, t} \sim Cauchy$ and $N = 500, T = 100$. Results are computed over 5,000 simulations. See Section \ref{section:simulation} of the main text for further details.}
    \label{fig: total qqplot plag0}
\end{figure}

\paragraph{Simulation results for the estimator of the lag-$1$ total weighted average dynamic causal effect, $\bar{\tau}^\dagger(1, 0; 1)$:} We now present simulation results that analyze the properties of our estimator for the lag-$1$ total weighted average dynamic causal effect, $\hat{\bar{\tau}}^\dagger(1, 0; 1)$. We choose the weights to $a_{\mathbf v}$ to place equal weight on the future treatment paths. Figure \ref{fig: total average histogram plag1} plots the simulated randomization distribution for $\hat{\bar{\tau}}^\dagger(1, 0; 1)$ and Figure \ref{fig: total average qqplot plag1} plots the associated quantile-quantile plot. We observe that the normal approximation remains accurate for lagged dynamic causal effects.

\begin{figure}[htbp!]
    \centering
    \begin{subfigure}{.5\textwidth}
    \centering
    \includegraphics[width=2.5in, height=2.5in]{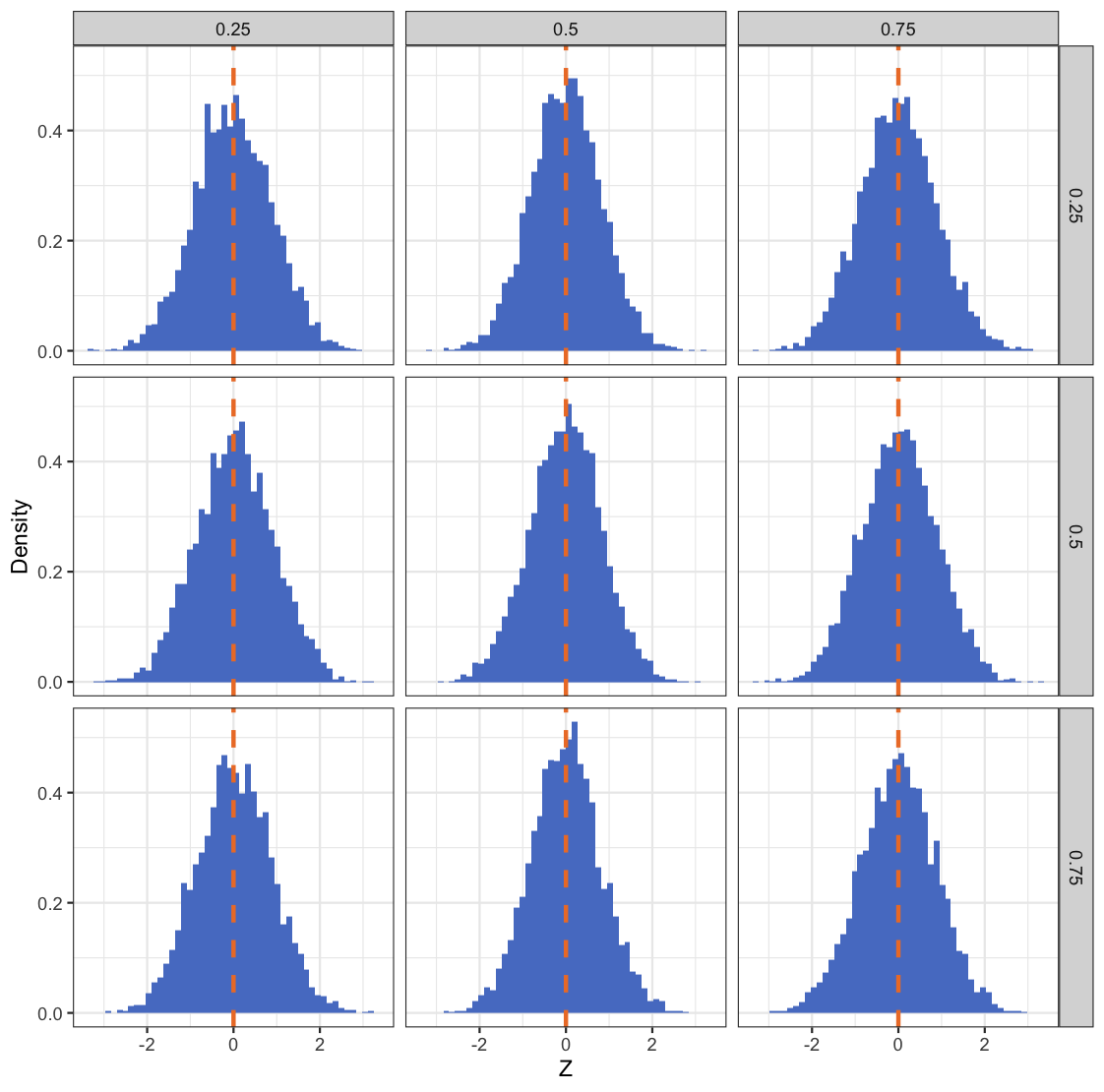}
    \caption{$\epsilon_{i, t} \sim N(0, 1)$, $N = 100$, $T = 10$}
    \end{subfigure}%
    \begin{subfigure}{.5\textwidth}
    \centering
    \includegraphics[width=2.5in, height=2.5in]{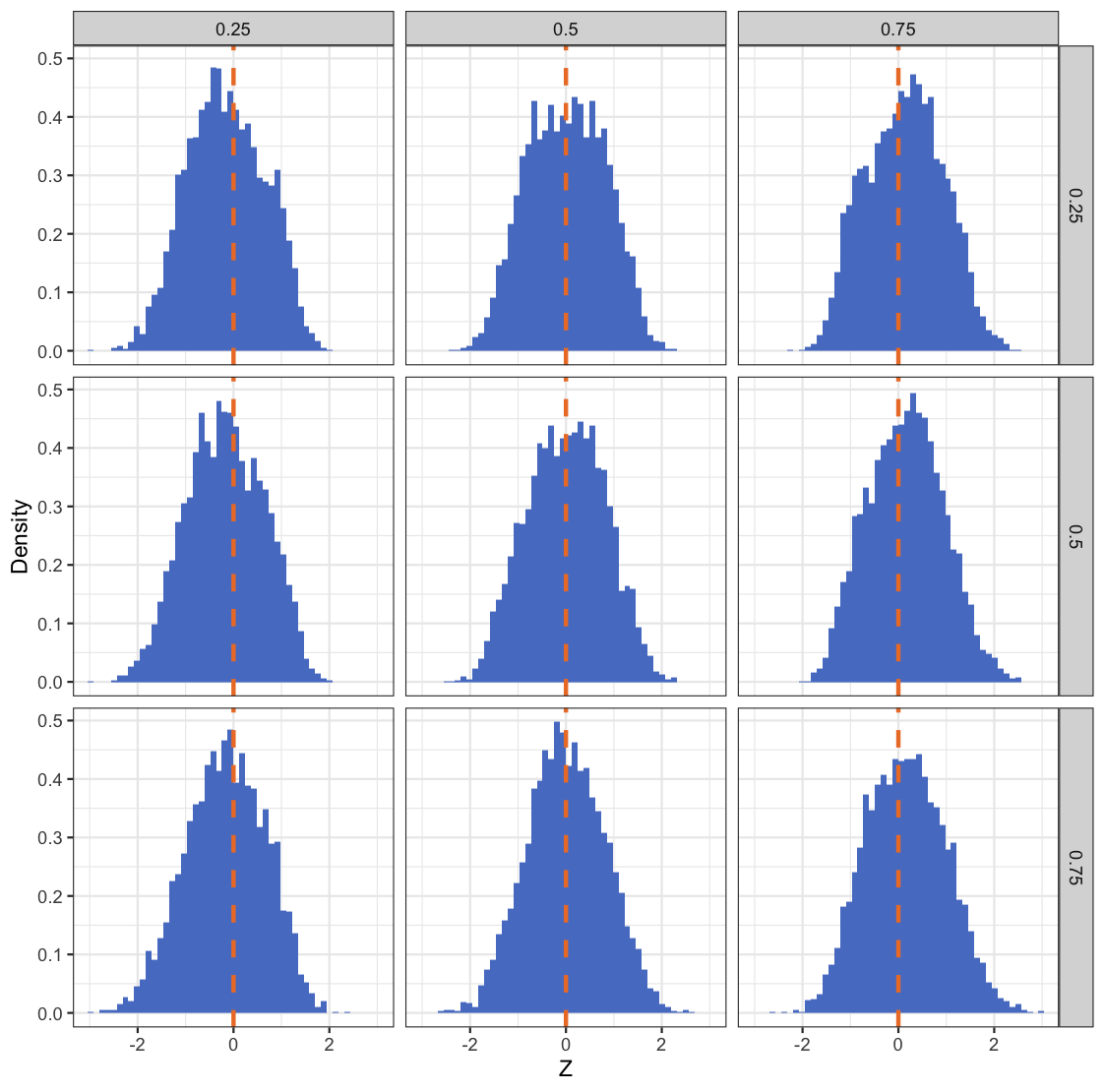}
    \caption{$\epsilon_{i, t} \sim Cauchy$, $N = 500$, $T = 100$}
    \end{subfigure}
    \caption{Simulated randomization distribution for $\hat{\bar{\tau}}^{\dagger}(1, 0; 1)$ under different choices of the parameter $\phi$ and treatment probability $p(w)$. The rows index the parameter $\phi \in \{0.25, 0.5, 0.75\}$, and the columns index the treatment probability $p(w) \in \{0.25, 0.5, 0.75\}$. Panel (a) plots the simulated randomization distribution with normally distributed errors $\epsilon_{i, t} \sim N(0, 1)$ and $N = 100, T = 10$. Panel (b) plots the simulated randomization distribution with Cauchy distribution errors $\epsilon_{i, t} \sim Cauchy$ and $N = 500, T = 10$. Results are computed over 5,000 simulations. See Section \ref{section:simulation} of the main text for further details.}
    \label{fig: total average histogram plag1}
\end{figure}

\begin{figure}[htbp!]
    \centering
    \begin{subfigure}{.5\textwidth}
    \centering
    \includegraphics[width=2.5in, height=2.5in]{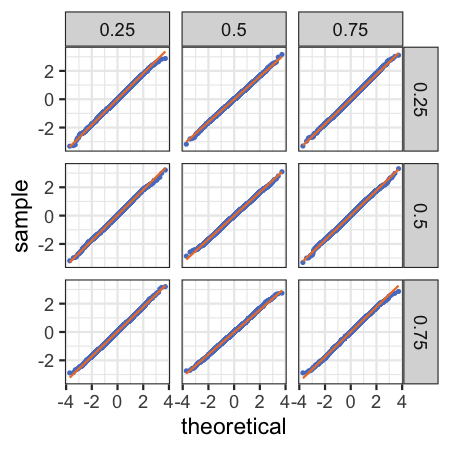}
    \caption{$\epsilon_{i, t} \sim N(0, 1)$, $N = 100, T = 10$}
    \end{subfigure}%
    \begin{subfigure}{.5\textwidth}
    \centering
    \includegraphics[width=2.5in, height=2.5in]{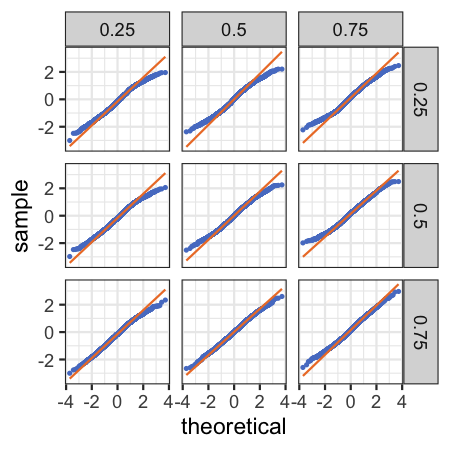}
    \caption{$\epsilon_{i, t} \sim Cauchy$, $N = 500, T = 100$}
    \end{subfigure}
    \caption{Quantile-quantile plots for the simulated randomization distribution for $\hat{\bar{\tau}}^{\dagger}(1, 0; 1)$ under different choices of the parameter $\phi$ and treatment probability $p(w)$. The quantile-quantile plots compare the quantiles of the simulated randomization distribution (y-axis) against the quantiles of a standard normal random variable (x-axis). The 45 degree line is plotted in solid orange. The rows index the parameter $\phi \in \{0.25, 0.5, 0.75\}$, and the columns index the treatment probability $p(w) \in \{0.25, 0.5, 0.75\}$. Panel (a) plots the quantile-quantile plots for simulated randomization distribution with normally distributed errors $\epsilon_{i, t} \sim N(0, 1)$ and $T = 1000$. Panel (b) plots the quantile-quantile plots simulated randomization distribution with Cauchy distribution errors $\epsilon_{i, t} \sim Cauchy$ and $T = 50,000$. Results are computed over 5,000 simulations. See Section \ref{section:simulation} of the main text for further details.}
    \label{fig: total average qqplot plag1}
\end{figure}

\newpage
\clearpage
\subsection{Simulations for the estimator of the time-$t$ average dynamic causal effects}

We present simulation results for our estimator of the time-$t$ average dynamic causal effect, $\hat{\bar{\tau}}_{\cdot, t}(1, 0; 0)$, with $N = 100$ units when the potential outcomes are generated with normally distributed errors and $N = 50,000$ with Cauchy distributed errors.

\paragraph{Normal approximations and size control:} Figure \ref{fig: unit average histogram plag0} plots the randomization distribution for the estimator of the contemporaneous time-$t$ average dynamic causal effect, $\hat{\bar{\tau}}_{\cdot t}(1, 0;0)$, under the null hypothesis of $\beta = 0$ for different combinations of the parameter $\phi\in\{0.25, 0.5, 0.75\}$ and treatment probability $p(w)\in\{0.25, 0.5, 0.75\}$. When the errors $\epsilon_{i, t}$ are normally distributed, the randomization distribution quickly converges to a normal distribution -- the normal approximation is accurate when there are only $N = 100$ units in the experiment. As expected, when the errors are Cauchy distributed, the number of units must be quite large for the randomization distribution to become approximately normal. There is little difference in the results across the values of $\phi$ and $p(w)$. Figure \ref{fig: unit average qqplot plag0} provides quantile-quantile plots of the simulated randomization distributions to further illustrate the quality of the normal approximations. Testing based on the normal asymptotic approximation controls size effectively, staying close to the nominal 5\% level (the exact rejection rates for the null hypothesis, $H_0: \bar{\tau}_{\cdot t}(1, 0;0) = 0$ are reported in Table \ref{table: null rejection prob, unit average}). 

\begin{figure}[!htbp]
    \centering
    \begin{subfigure}{.5\textwidth}
    \centering
    \includegraphics[width=2.5in, height=2.5in]{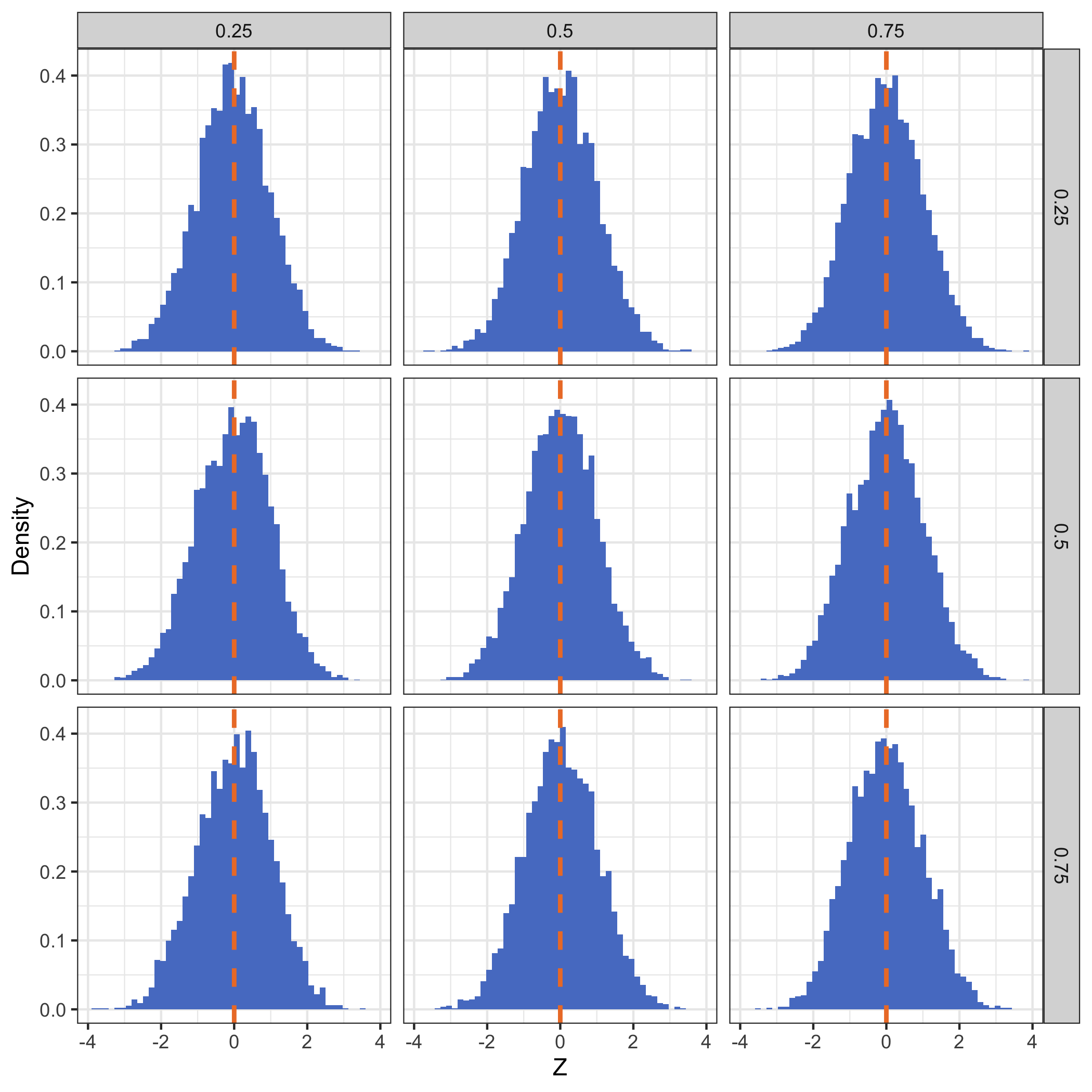}
    \caption{$\epsilon_{i, t} \sim N(0, 1)$, $N = 100$}
    \end{subfigure}%
    \begin{subfigure}{.5\textwidth}
    \centering
    \includegraphics[width=2.5in, height=2.5in]{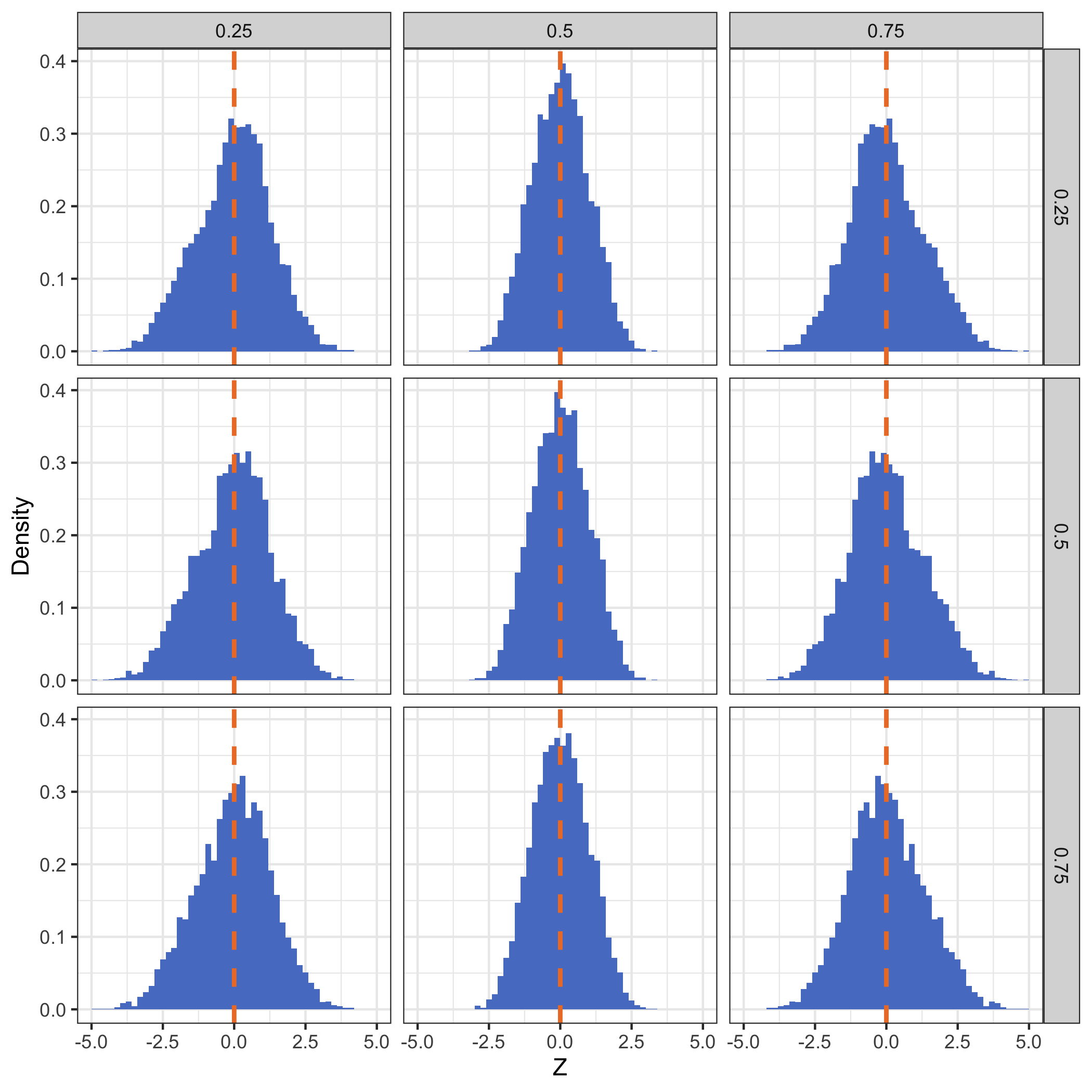}
    \caption{$\epsilon_{i, t} \sim Cauchy$, $N = 50,000$}
    \end{subfigure}
    \caption{Simulated randomization distribution for $\hat{\bar{\tau}}_{\cdot t}(1, 0;0)$ under different choices of the parameter $\phi$ and treatment probability $p(w)$. The rows index the parameter $\phi \in \{0.25, 0.5, 0.75\}$  and the columns index the treatment probability $p(w) \in \{0.25, 0.5, 0.75\}$. Panel (a) plots the simulated randomization distribution with normally distributed errors $\epsilon_{i, t} \sim N(0, 1)$ and $N = 100$. Panel (b) plots the simulated randomization distribution with Cauchy distribution errors $\epsilon_{i, t} \sim Cauchy$ and $N = 50,000$. Results are computed over 5,000 iterations. See Section \ref{section:simulation} of the main text for further details on the simulation design.}
    \label{fig: unit average histogram plag0}
\end{figure}

\begin{table}[htbp!]
    \begin{minipage}{0.5\textwidth}
        \centering
            \begin{tabular}{r l || c c c }
                 & & \multicolumn{3}{c}{$p(w)$} \\
                 & & $0.25$ & $0.5$ & $0.75$ \\
                 \hline \hline
                 \multirow{3}{*}{$\phi$} & $0.25$ & $0.046$ & $0.048$ & $0.048$ \\
                 & $0.5$ & $0.049$ & $0.049$ & $0.050$ \\
                 & $0.75$ & $0.050$ & $0.049$ & $0.045$ \\
            \end{tabular}
            \subcaption{$\epsilon_{i, t} \sim N(0, 1), N = 100$}
        \end{minipage}
        \begin{minipage}{0.5\textwidth}
            \centering
            \begin{tabular}{r l || c c c}
                & & \multicolumn{3}{c}{$p(w)$} \\
                 & & $0.25$ & $0.5$ & $0.75$ \\
                 \hline \hline
                 \multirow{3}{*}{$\phi$} & $0.25$ & $0.031$ & $0.031$ & $0.034$ \\
                 & $0.5$ & $0.048$ & $0.039$ & $0.043$ \\
                 & $0.75$ & $0.052$ & $0.047$ & $0.057$
            \end{tabular}
            \subcaption{$\epsilon_{i, t} \sim Cauchy, N = 50,000$}
        \end{minipage}
        \caption{Null rejection rate for the test of the null hypothesis $H_0: \bar{\tau}_{\cdot t}(1, 0; 0) = 0$ based upon the normal asymptotic approximation to the randomization distribution of $\hat{\bar{\tau}}_{\cdot t}(1, 0; 0)$. Panel (a) reports the null rejection probabilities in simulations with $\epsilon_{i, t} \sim N(0, 1)$ and $N = 100$. Panel (b) reports the null rejection probabilities in simulations with $\epsilon_{i, t} \sim Cauchy$ and $N = 50,000$. Results are computed over 5,000 simulations. See Section \ref{section:simulation} of the main text for further details on the simulation design.}
        \label{table: null rejection prob, unit average}
\end{table}

\begin{figure}[htbp!]
    \centering
    \begin{subfigure}{.5\textwidth}
    \centering
    \includegraphics[width=2.5in, height=2.5in]{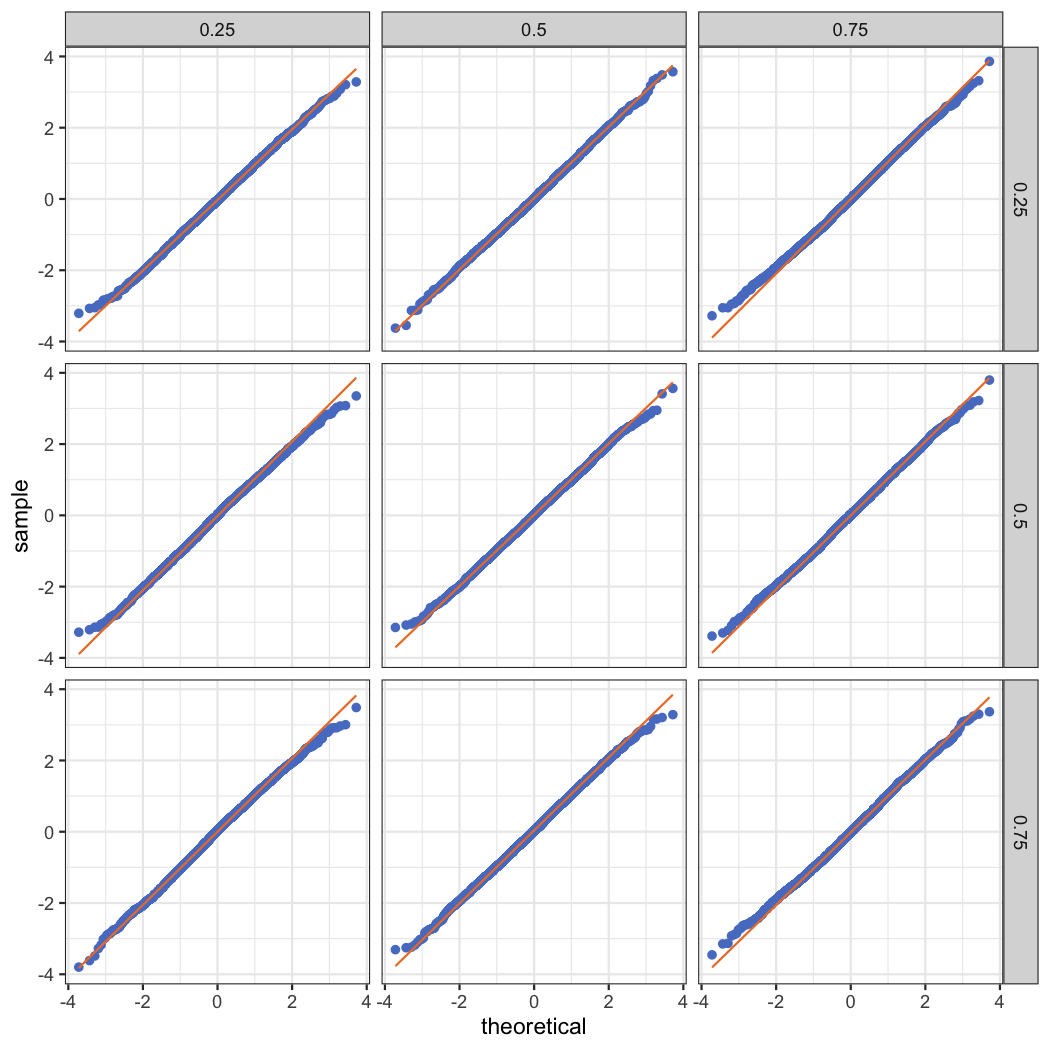}
    \caption{$\epsilon_{i, t} \sim N(0, 1)$, $N = 100$}
    \end{subfigure}%
    \begin{subfigure}{.5\textwidth}
    \centering
    \includegraphics[width=2.5in, height=2.5in]{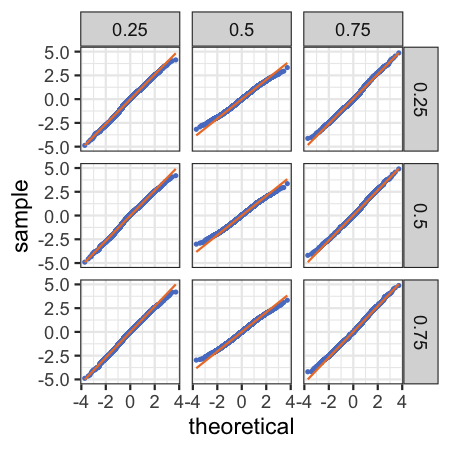}
    \caption{$\epsilon_{i, t} \sim Cauchy$, $N = 50,000$}
    \end{subfigure}
    \caption{Quantile-quantile plots for the simulated randomization distribution for $\hat{\bar{\tau}}_{\cdot t}(1, 0; 0)$ under different choices of the parameter $\phi$ and treatment probability $p(w)$. The quantile-quantile plots compare the quantiles of the simulated randomization distribution (y-axis) against the quantiles of a standard normal random variable (x-axis). The 45 degree line is plotted in solid orange. The rows index the parameter $\phi \in \{0.25, 0.5, 0.75\}$, and the columns index the treatment probability $p(w) \in \{0.25, 0.5, 0.75\}$. Panel (a) plots the quantile-quantile plots for simulated randomization distribution with normally distributed errors $\epsilon_{i, t} \sim N(0, 1)$ and $N = 100$. Panel (b) plots the quantile-quantile plots simulated randomization distribution with Cauchy distribution errors $\epsilon_{i, t} \sim Cauchy$ and $N = 50,000$. Results are computed over 5,000 simulations. See Section \ref{section:simulation} of the main text for further details on the simulation design.}
    \label{fig: unit average qqplot plag0}
\end{figure}

\paragraph{Rejection rates:} Figure \ref{fig: unit average power plot, n100} plots rejection rate curves against the null hypotheses as the parameter $\beta$ varies for different choices of the parameter $\phi$ and treatment probability $p(w)$ in simulations with $N = 100$ units. For $p = 0$, the rejection rate against $H_0: \bar{\tau}_{\cdot t}(1, 0; 0) = 0$ quickly converges to one as $\beta$ moves away from zero across a range of simulations. This is encouraging as it indicates that the conservative variance bound still leads to informative tests. However, when $p = 1$, the persistence of the causal effects $\phi$ has an important effect on the power of our tests. In particular, when $\phi = 0.25$, the rejection rate against $H_0: \bar{\tau}^{\dagger}_{\cdot t}(1, 0; 1) = 0$ is quite low for all values of $\beta$ -- lower values of $\phi$ imply less persistence in the causal effects across periods. When $\phi = 0.75$, there is substantial persistence across periods and observe that the rejection rate curves improve for $p = 1$. Additionally, Figure \ref{fig: unit average power plot, n1000} shows the same power plots for $N = 1000$ units. We again observe that power is relatively low for low values of $\phi$, but when $\phi = 0.75$, the rejection rate curves for $p = 0, 1$ appear similar. This suggests that detecting dynamic causal effects requires larger sample sizes.

\begin{figure}[!htbp]
    \centering
    \includegraphics[width=5in, height=2.5in]{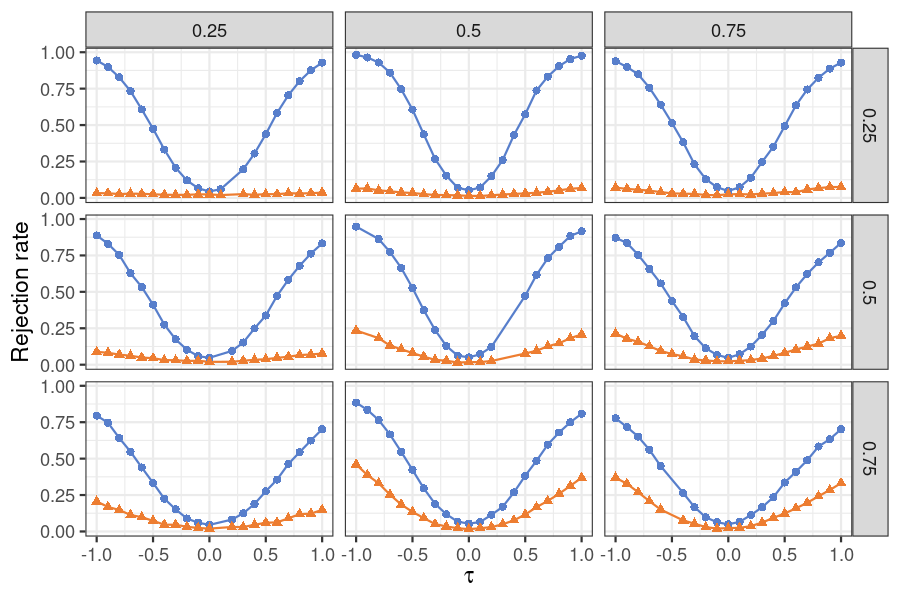}
    \caption{Rejection probabilities for a test of the null hypothesis $H_0: \bar{\tau}_{\cdot t}(1, 0;0) = 0$ and $H_0: \bar{\tau}^{\dagger}_{\cdot t}(1, 0; 1) = 0$ as the parameter $\beta$ varies under different choices of the parameter $\phi$ and treatment probability $p(w)$. The rejection rate curve against $H_0: \bar{\tau}_{\cdot t}(1, 0;0) = 0$ is plotted in blue and the rejection rate curve against $H_0: \bar{\tau}^{\dagger}_{\cdot t}(1, 0; 1) = 0$ is plotted in orange. The rows index the parameter $\phi \in \{0.25, 0.5, 0.75\}$, and the columns index the treatment probability $p(w) \in \{0.25, 0.5, 0.75\}$. The simulations are conducted with normally distributed errors $\epsilon_{i, t} \sim N(0, 1)$ and $N = 100$. Results are averaged over $5000$ simulations. See Section \ref{section:simulation} of the main text for further details on the simulation design.}
    \label{fig: unit average power plot, n100}
\end{figure}

\begin{figure}[!htbp]
    \centering
    \includegraphics[width=5in, height=2.5in]{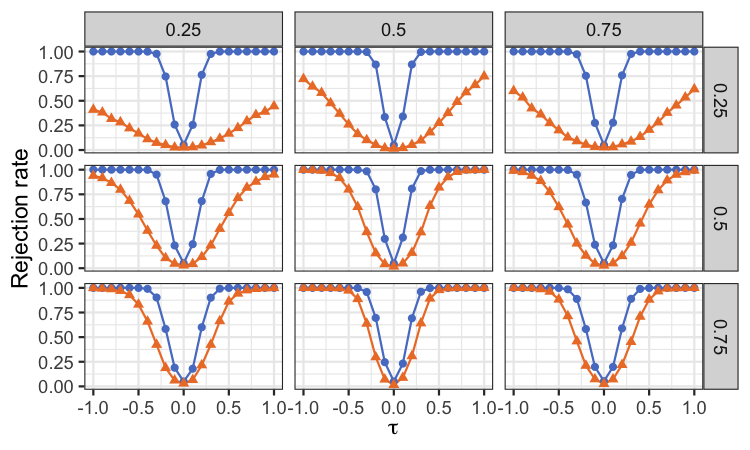}
    \caption{Rejection probabilities for a test of the null hypothesis $H_0: \bar{\tau}_{\cdot t}(1, 0;0) = 0$ and $H_0: \bar{\tau}^{\dagger}_{\cdot t}(1, 0; 1) = 0$ as the parameter $\beta$ varies under different choices of the parameter $\phi$ and treatment probability $p(w)$. The rejection rate curve against $H_0: \bar{\tau}_{\cdot t}(1, 0;0) = 0$ is plotted in blue and the rejection rate curve against $H_0: \bar{\tau}^{\dagger}_{\cdot t}(1, 0; 1) = 0$ is plotted in orange. The rows index the parameter $\phi \in \{0.25, 0.5, 0.75\}$, and the columns index the treatment probability $p(w) \in \{0.25, 0.5, 0.75\}$. The simulations are conducted with normally distributed errors $\epsilon_{i, t} \sim N(0, 1)$ and $N = 1000$. Results are averaged over $5000$ simulations. See Section \ref{section:simulation} of the main text for further details on the simulation design.}
    \label{fig: unit average power plot, n1000}
\end{figure}

\paragraph{Simulation results for the estimator of the lag-$1$, time-$t$ weighted average dynamic causal effect, $\bar{\tau}^\dagger_{\cdot, t}(1, 0; 1)$:} We now present simulation results that analyze the properties of our estimator for the lag-$1$ total weighted average dynamic causal effect, $\hat{\bar{\tau}}^\dagger_{\cdot, t}(1, 0; 1)$. We choose the weights to $a_{\mathbf v}$ to place equal weight on the future treatment paths. Figure \ref{fig: unit average histogram plag1} plots the simulated randomization distribution for $\hat{\bar{\tau}}^\dagger_{\cdot, t}(1, 0; 1)$ and Figure \ref{fig: unit average qqplot plag1} plots the associated quantile-quantile plot. We observe that the normal approximation remains accurate for lagged dynamic causal effects.

\begin{figure}[htbp!]
    \centering
    \begin{subfigure}{.5\textwidth}
    \centering
    \includegraphics[width=2.5in, height=2.5in]{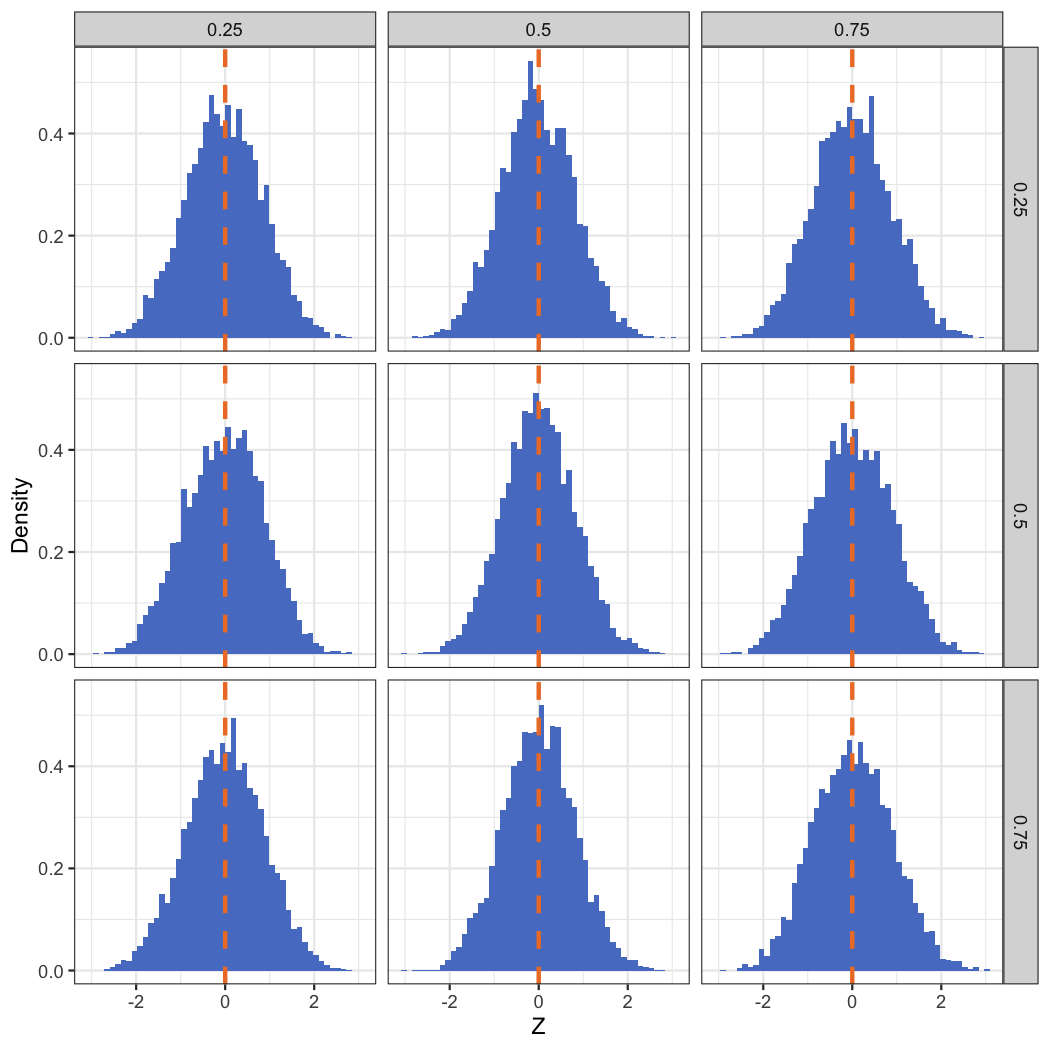}
    \caption{$\epsilon_{i, t} \sim N(0, 1)$, $N = 100$}
    \end{subfigure}%
    \begin{subfigure}{.5\textwidth}
    \centering
    \includegraphics[width=2.5in, height=2.5in]{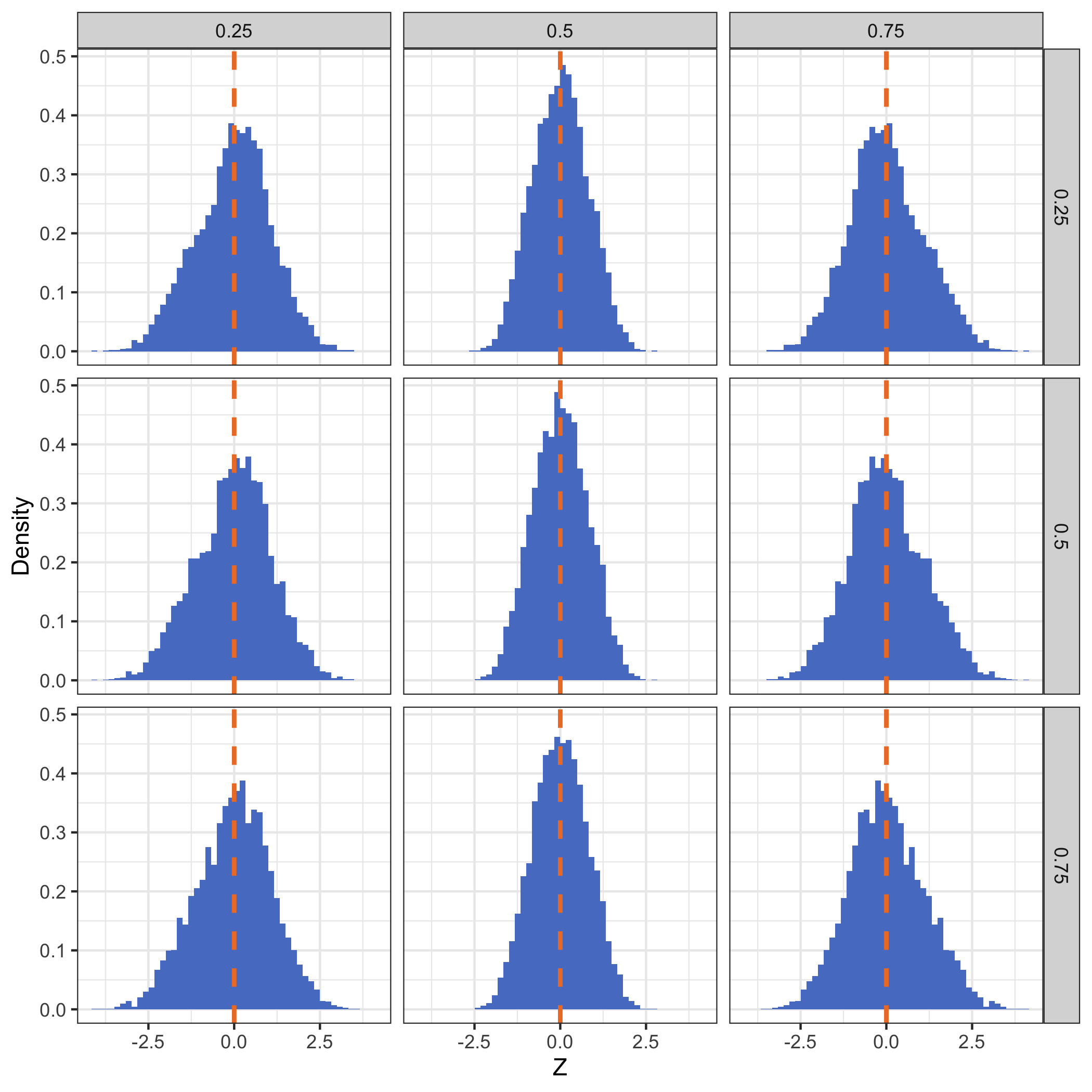}
    \caption{$\epsilon_{i, t} \sim Cauchy$, $N = 50,000$}
    \end{subfigure}
    \caption{Simulated randomization distribution for $\hat{\bar{\tau}}^{\dagger}_{\cdot t}(1, 0; 1)$ under different choices of the parameter $\phi$ and treatment probability $p(w)$. The rows index the parameter $\phi \in \{0.25, 0.5, 0.75\}$, and the columns index the treatment probability $p(w) \in \{0.25, 0.5, 0.75\}$. Panel (a) plots the simulated randomization distribution with normally distributed errors $\epsilon_{i, t} \sim N(0, 1)$ and $N = 100$. Panel (b) plots the simulated randomization distribution with Cauchy distribution errors $\epsilon_{i, t} \sim Cauchy$ and $N = 50,000$. Results are computed over 5,000 simulations. See Section \ref{section:simulation} of the main text for further details on the simulation design.}
    \label{fig: unit average histogram plag1}
\end{figure}

\begin{figure}[htbp!]
    \centering
    \begin{subfigure}{.5\textwidth}
    \centering
    \includegraphics[width=2.5in, height=2.5in]{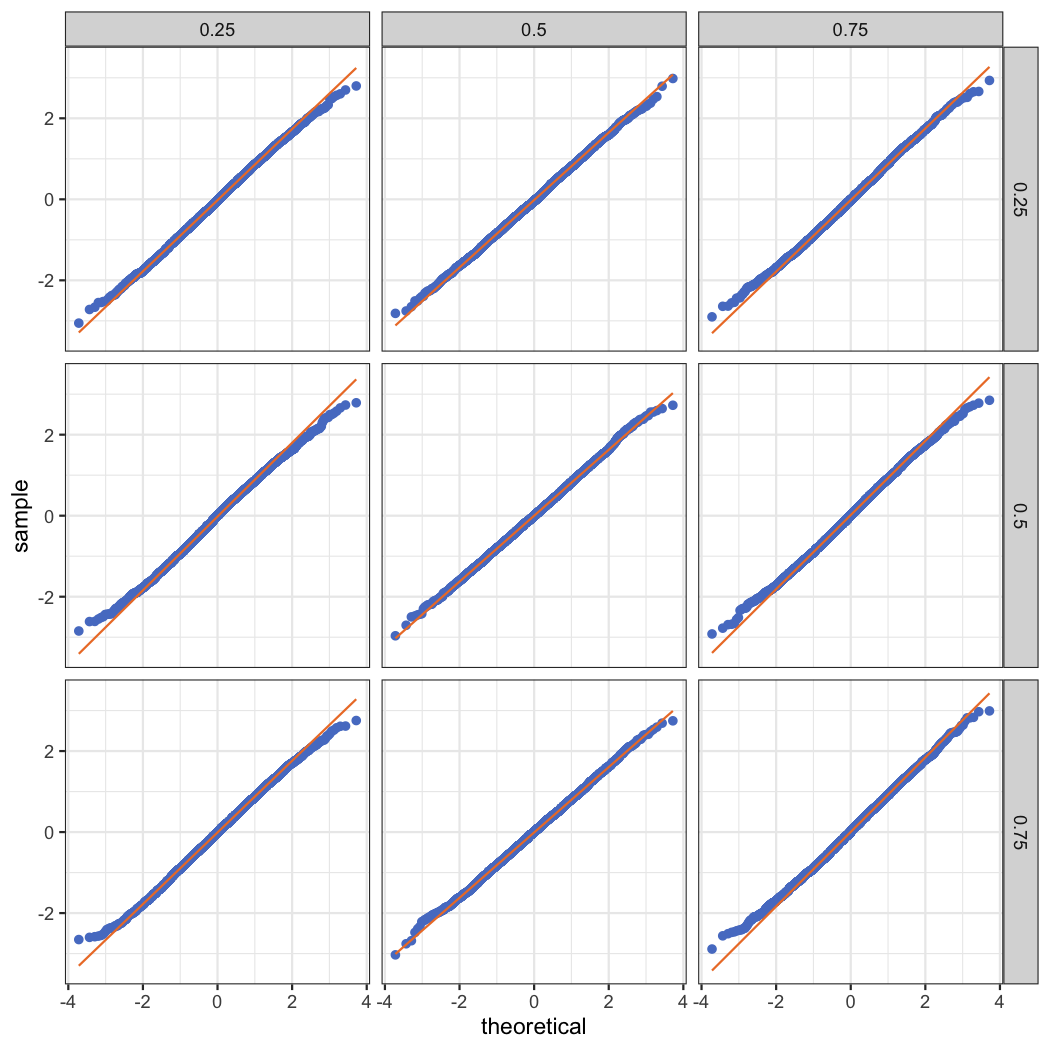}
    \caption{$\epsilon_{i, t} \sim N(0, 1)$, $N = 100$}
    \end{subfigure}%
    \begin{subfigure}{.5\textwidth}
    \centering
    \includegraphics[width=2.5in, height=2.5in]{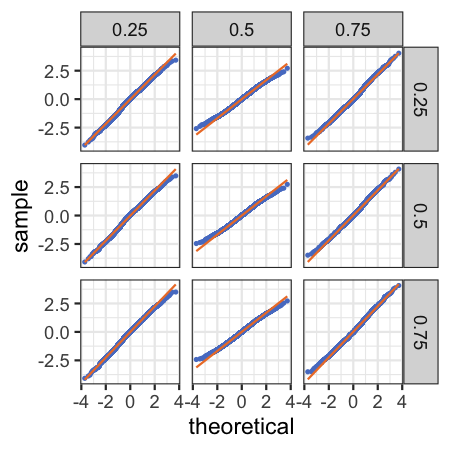}
    \caption{$\epsilon_{i, t} \sim Cauchy$, $N = 50,000$}
    \end{subfigure}
    \caption{Quantile-quantile plots for the simulated randomization distribution for $\hat{\bar{\tau}}^{\dagger}_{\cdot t}(1, 0; 1)$ under different choices of the parameter $\phi$ and treatment probability $p(w)$. The quantile-quantile plots compare the quantiles of the simulated randomization distribution (y-axis) against the quantiles of a standard normal random variable (x-axis). The 45 degree line is plotted in solid orange. The rows index the parameter $\phi \in \{0.25, 0.5, 0.75\}$, and the columns index the treatment probability $p(w) \in \{0.25, 0.5, 0.75\}$. Panel (a) plots the quantile-quantile plots for simulated randomization distribution with normally distributed errors $\epsilon_{i, t} \sim N(0, 1)$ and $N= 1000$. Panel (b) plots the quantile-quantile plots simulated randomization distribution with Cauchy distribution errors $\epsilon_{i, t} \sim Cauchy$ and $N = 50,000$. Results are computed over 5,000 simulations. See Section \ref{section:simulation} of the main text for further details on the simulation design.}
    \label{fig: unit average qqplot plag1}
\end{figure}

\newpage
\clearpage
\subsection{Simulations for the estimator of the unit-$i$ average dynamic causal effects}

We present simulation results for our estimator of the unit-$i$ average dynamic causal effect, $\hat{\bar{\tau}}_{i, \cdot}(1, 0; 0)$, with $T = 100$ time periods when the potential outcomes are generated with normally distributed errors and $T = 50,000$ with Cauchy distributed errors.

\paragraph{Normal approximations and size control:} Figure \ref{fig: time average histogram plag0} plots the randomization distribution for $\hat{\bar{\tau}}_{i \cdot}(1, 0;0)$. We see a similar pattern as before---when the errors are normally distributed, the randomization distribution converges quickly to a normal distribution, but it takes longer to do so when the errors are heavy-tailed. Figure \ref{fig: time average qqplot plag0} provides quantile-quantile plots of the simulation randomization distributions to further illustrate the quality of the normal approximations. The null rejection rates for the hypothesis, $H_0: \bar{\tau}_{i,\cdot}(1, 0;0) = 0$ are reported in Table \ref{table: null rejection prob, time average} and, again, the test controls size well across a wide range of parameters. 

\begin{figure}[!htbp]
    \centering
    \begin{subfigure}{.5\textwidth}
    \centering
    \includegraphics[width=2.5in, height=2.5in]{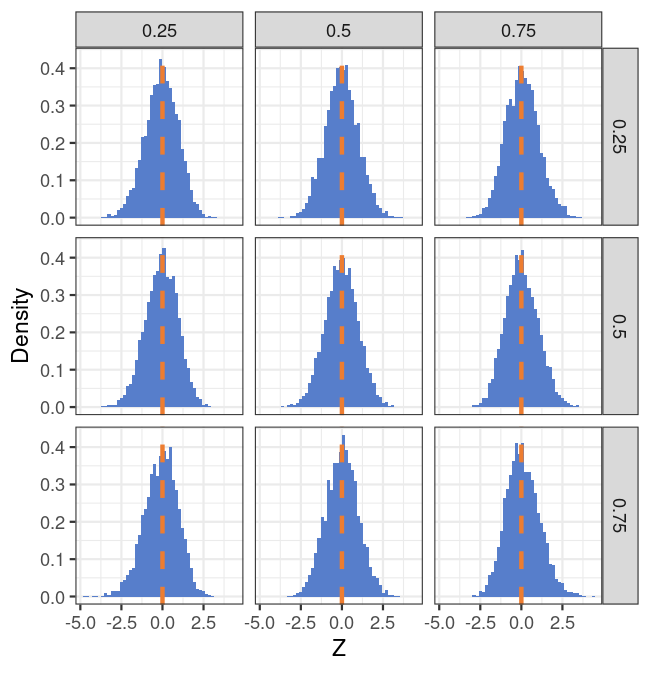}
    \caption{$\epsilon_{i, t} \sim N(0, 1)$, $T = 100$}
    \end{subfigure}%
    \begin{subfigure}{.5\textwidth}
    \centering
    \includegraphics[width=2.5in, height=2.5in]{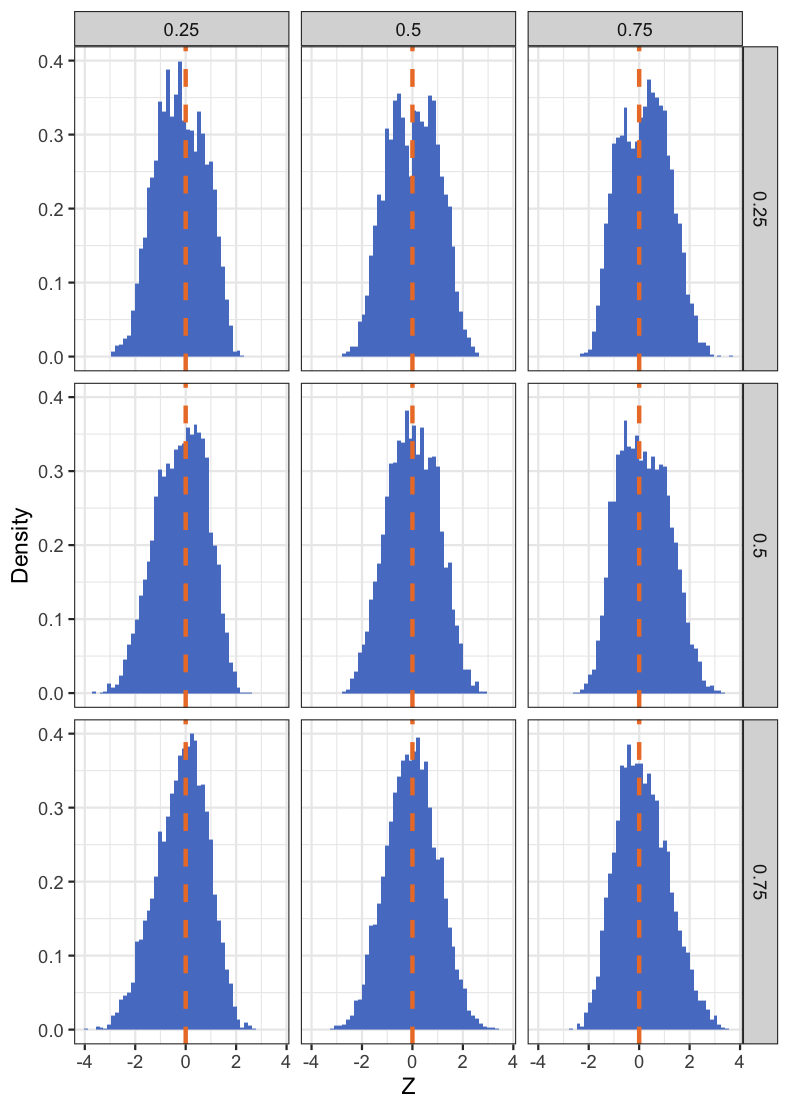}
    \caption{$\epsilon_{i, t} \sim Cauchy$, $T = 50,000$}
    \end{subfigure}
    \caption{Simulated randomization distribution for $\hat{\bar{\tau}}_{i \cdot}(1, 0;0)$ under different choices of the parameter $\phi$ and treatment probability $p(w)$. The rows index the parameter $\phi \in \{0.25, 0.5, 0.75\}$, and the columns index the treatment probability $p(w) \in \{0.25, 0.5, 0.75\}$. Panel (a) plots the simulated randomization distribution with normally distributed errors $\epsilon_{i, t} \sim N(0, 1)$ and $T = 100$. Panel (b) plots the simulated randomization distribution with Cauchy distribution errors $\epsilon_{i, t} \sim Cauchy$ and $T = 50,000$. Results are computed over 5,000 simulations. See Section \ref{section:simulation} of the main text for further details on the simulation design.}
    \label{fig: time average histogram plag0}
\end{figure}

\begin{table}[htbp!]
    \begin{minipage}{0.5\textwidth}
        \centering
            \begin{tabular}{r l || c c c }
                 & & \multicolumn{3}{c}{$p(w)$} \\
                 & & $0.25$ & $0.5$ & $0.75$ \\
                 \hline \hline
                 \multirow{3}{*}{$\phi$} & $0.25$ & $0.052$ & $0.047$ & $0.054$ \\
                 & $0.5$ & $0.049$ & $0.049$ & $0.048$ \\
                 & $0.75$ & $0.058$ & $0.046$ & $0.054$
            \end{tabular}
            \subcaption{$\epsilon_{i, t} \sim N(0, 1), T = 100$}
        \end{minipage}
        \begin{minipage}{0.5\textwidth}
            \centering
            \begin{tabular}{r l || c c c}
                & & \multicolumn{3}{c}{$p(w)$} \\
                 & & $0.25$ & $0.5$ & $0.75$ \\
                 \hline \hline
                 \multirow{3}{*}{$\phi$} & $0.25$ & $0.031$ & $0.031$ & $0.034$ \\
                 & $0.5$ & $0.048$ & $0.039$ & $0.043$ \\
                 & $0.75$ & $0.052$ & $0.047$ & $0.057$
            \end{tabular}
            \subcaption{$\epsilon_{i, t} \sim Cauchy, T = 50,000$}
        \end{minipage}
        \caption{Null rejection rate for the test of the null hypothesis $H_0: \bar{\tau}_{i \cdot }(1, 0; 0) = 0$ based upon the normal asymptotic approximation to the randomization distribution of $\hat{\bar{\tau}}_{i \cdot}(1, 0; 0)$. Panel (a) reports the null rejection probabilities in simulations with $\epsilon_{i, t} \sim N(0, 1)$ and $T = 100$. Panel (b) reports the null rejection probabilities in simulations with $\epsilon_{i, t} \sim Cauchy$ and $T = 50,000$. Results are computed over 5,000 simulations. See Section \ref{section:simulation} of the main text for further details on the simulation design.}
        \label{table: null rejection prob, time average}
\end{table}

\begin{figure}[htbp!]
    \centering
    \begin{subfigure}{.5\textwidth}
    \centering
    \includegraphics[width=2.5in, height=2.5in]{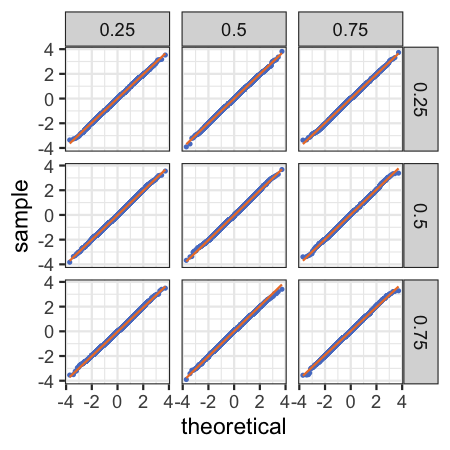}
    \caption{$\epsilon_{i, t} \sim N(0, 1)$, $T = 100$}
    \end{subfigure}%
    \begin{subfigure}{.5\textwidth}
    \centering
    \includegraphics[width=2.5in, height=2.5in]{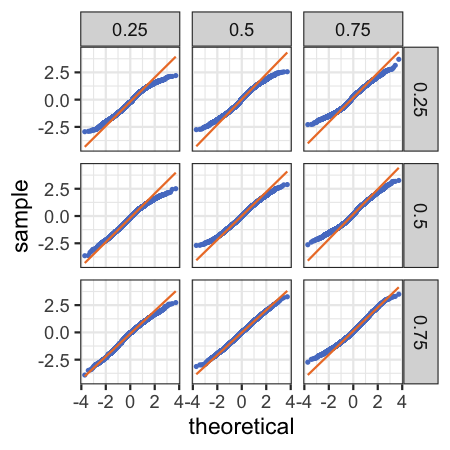}
    \caption{$\epsilon_{i, t} \sim Cauchy$, $T = 50,000$}
    \end{subfigure}
    \caption{Quantile-quantile plots for the simulated randomization distribution for $\hat{\bar{\tau}}_{i \cdot}(1, 0; 0)$ under different choices of the parameter $\phi$ and treatment probability $p(w)$. The quantile-quantile plots compare the quantiles of the simulated randomization distribution (y-axis) against the quantiles of a standard normal random variable (x-axis). The 45 degree line is plotted in solid orange. The rows index the parameter $\phi \in \{0.25, 0.5, 0.75\}$, and the columns index the treatment probability $p(w) \in \{0.25, 0.5, 0.75\}$. Panel (a) plots the quantile-quantile plots for simulated randomization distribution with normally distributed errors $\epsilon_{i, t} \sim N(0, 1)$ and $T = 100$. Panel (b) plots the quantile-quantile plots simulated randomization distribution with Cauchy distribution errors $\epsilon_{i, t} \sim Cauchy$ and $T = 50,000$. Results are computed over 5,000 simulations. See Section \ref{section:simulation} of the main text for further details on the simulation design.}
    \label{fig: time average qqplot plag0}
\end{figure}

\paragraph{Rejection rates:}

Next, we investigate the rejection rate of the statistical test based on the normal asymptotic approximation for $H_0: \bar{\tau}^{\dagger}_{i \cdot}(1, 0; 0) = 0$ and $H_0: \bar{\tau}^{\dagger}_{i \cdot}(1, 0; 1) = 0$, plotting the rejection rates in Figure \ref{fig: time average power plot, T100}. For $p = 0$,  Once again, we observe that the rejection rate against $H_0: \bar{\tau}^{\dagger}_{i \cdot}(1, 0;0) = 0$ has good power properties across a range of simulations. However, once again for $p = 1$, our conservative test has low power and the persistence of the causal effects $\phi$ has an important effect on the power of our tests. Additionally, Figure \ref{fig: time average power plot, T1000} shows the same power plots for $T = 1000$ time periods. In this case, we observe that the conservative test has good power against the weak null of no unit-$i$ average dynamic causal effects for both $p = 0, 1$. This suggests that detecting unit-$i$ average dynamic causal effects requires a long time dimension in the panel experiment.

\begin{figure}[!htbp]
    \centering
    \includegraphics[width=5in, height=2.5in]{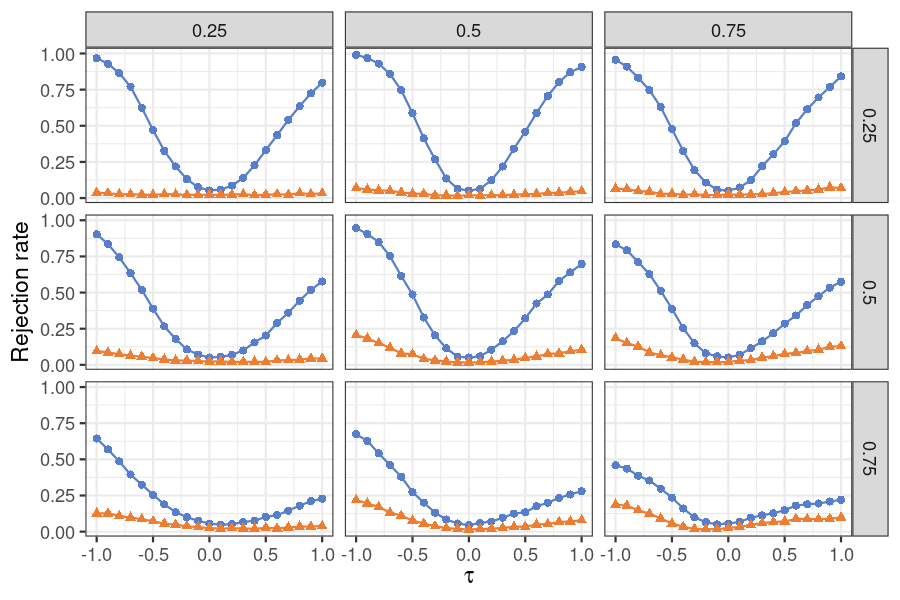}
    \caption{Rejection probabilities for a test of the null hypothesis $H_0: \bar{\tau}^{\dagger}_{i \cdot}(1, 0; 0) = 0$ and $H_0: \bar{\tau}^{\dagger}_{i \cdot}(1, 0; 1) = 0$ as the parameter $\beta$ varies under different choices of the parameter $\phi$ and treatment probability $p(w)$. The rejection rate curve against $H_0: \bar{\tau}^{\dagger}_{i \cdot}(1, 0; 0) = 0$ is plotted in blue and the rejection rate curve against $H_0: \bar{\tau}^{\dagger}_{i \cdot}(1, 0; 1) = 0$ is plotted in orange. The rows index the parameter $\phi \in \{0.25, 0.5, 0.75\}$, and the columns index the treatment probability $p(w) \in \{0.25, 0.5, 0.75\}$. The simulations are conducted with normally distributed errors $\epsilon_{i, t} \sim N(0, 1)$ and $T = 100$. Results are averaged over $5000$ simulations. See Section \ref{section:simulation} of the main text for further details on the simulation design.}
    \label{fig: time average power plot, T100}
\end{figure}

\begin{figure}[!htbp]
    \centering
    \includegraphics[width=5in, height=2.5in]{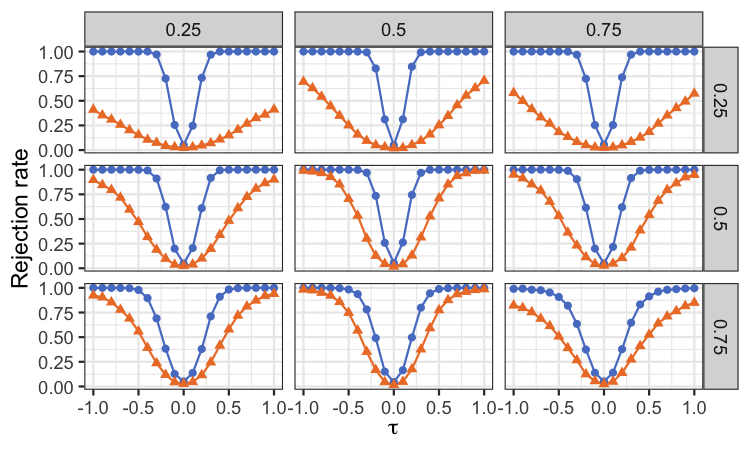}
    \caption{Rejection probabilities for a test of the null hypothesis $H_0: \bar{\tau}^{\dagger}_{i \cdot}(1, 0; 0) = 0$ and $H_0: \bar{\tau}^{\dagger}_{i \cdot}(1, 0; 1) = 0$ as the parameter $\beta$ varies under different choices of the parameter $\phi$ and treatment probability $p(w)$. The rejection rate curve against $H_0: \bar{\tau}^{\dagger}_{i \cdot}(1, 0; 0) = 0$ is plotted in blue and the rejection rate curve against $H_0: \bar{\tau}^{\dagger}_{i \cdot}(1, 0; 1) = 0$ is plotted in orange. The rows index the parameter $\phi \in \{0.25, 0.5, 0.75\}$, and the columns index the treatment probability $p(w) \in \{0.25, 0.5, 0.75\}$. The simulations are conducted with normally distributed errors $\epsilon_{i, t} \sim N(0, 1)$ and $T = 1000$. Results are averaged over $5000$ simulations. See Section \ref{section:simulation} of the main text for further details on the simulation design.}
    \label{fig: time average power plot, T1000}
\end{figure}

\paragraph{Simulation results for the estimator of the lag-$1$, unit-$i$ weighted average dynamic causal effect, $\bar{\tau}^\dagger_{i, \cdot}(1, 0; 1)$:} We now present simulation results that analyze the properties of our estimator for the lag-$1$ total weighted average dynamic causal effect, $\hat{\bar{\tau}}^\dagger_{i, \cdot}(1, 0; 1)$. We choose the weights to $a_{\mathbf v}$ to place equal weight on the future treatment paths. Figure \ref{fig: unit average histogram plag1} plots the simulated randomization distribution for $\hat{\bar{\tau}}^\dagger_{i, \cdot}(1, 0; 1)$ and Figure \ref{fig: unit average qqplot plag1} plots the associated quantile-quantile plot. We observe that the normal approximation remains accurate for lagged dynamic causal effects.

\begin{figure}[htbp!]
    \centering
    \begin{subfigure}{.5\textwidth}
    \centering
    \includegraphics[width=2.5in, height=2.5in]{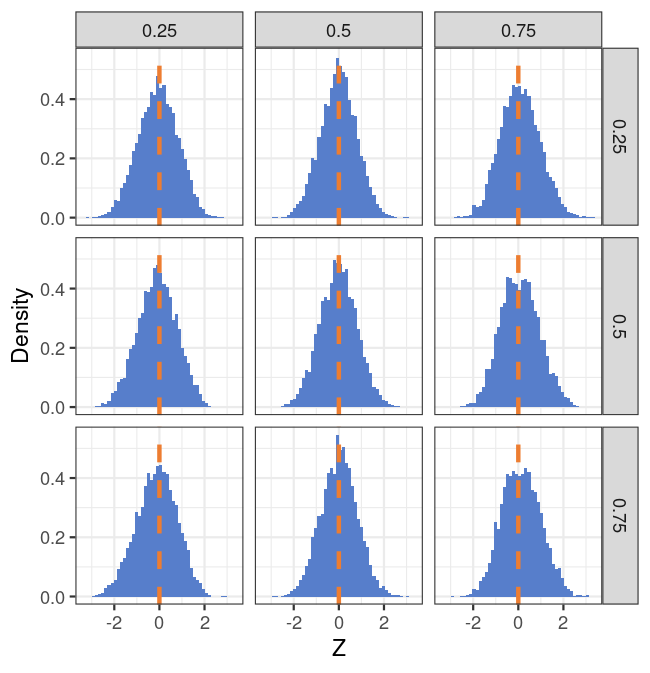}
    \caption{$\epsilon_{i, t} \sim N(0, 1)$, $T = 100$}
    \end{subfigure}%
    \begin{subfigure}{.5\textwidth}
    \centering
    \includegraphics[width=2.5in, height=2.5in]{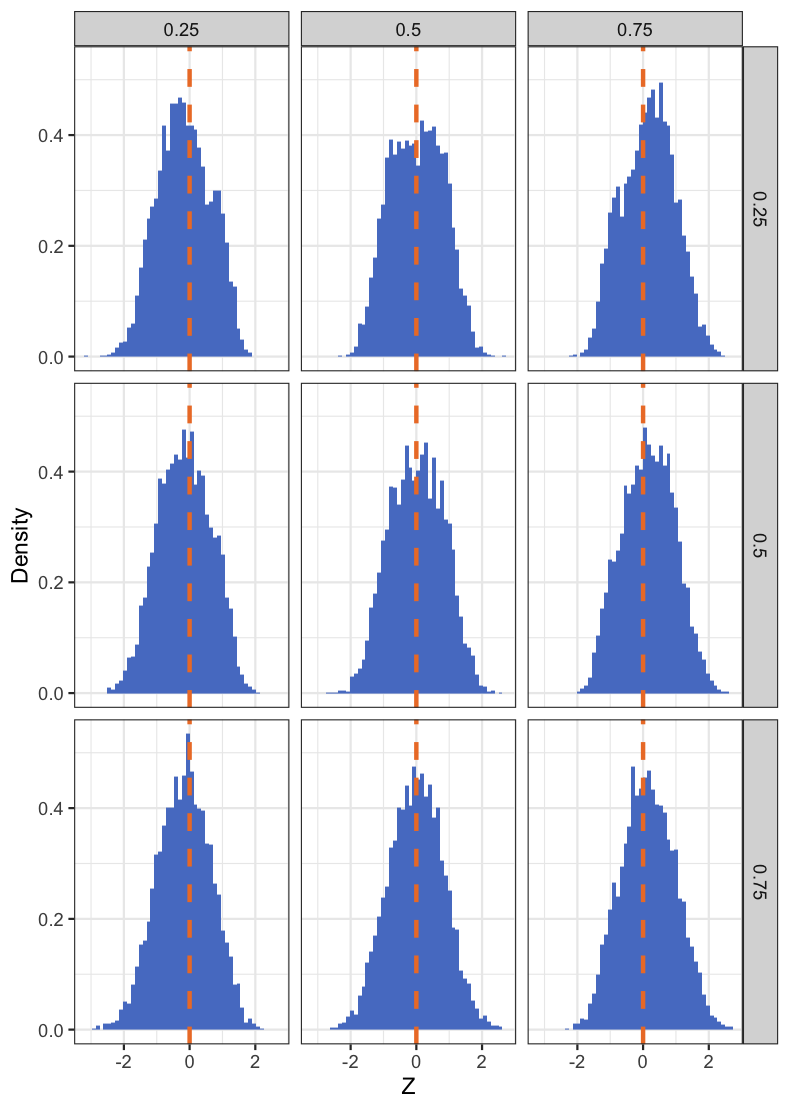}
    \caption{$\epsilon_{i, t} \sim Cauchy$, $T = 50,000$}
    \end{subfigure}
    \caption{Simulated randomization distribution for $\hat{\bar{\tau}}^{\dagger}_{i \cdot}(1, 0; 1)$ under different choices of the parameter $\phi$ and treatment probability $p(w)$. The rows index the parameter $\phi \in \{0.25, 0.5, 0.75\}$, and the columns index the treatment probability $p(w) \in \{0.25, 0.5, 0.75\}$. Panel (a) plots the simulated randomization distribution with normally distributed errors $\epsilon_{i, t} \sim N(0, 1)$ and $T = 100$. Panel (b) plots the simulated randomization distribution with Cauchy distribution errors $\epsilon_{i, t} \sim Cauchy$ and $T = 50,000$. Results are computed over 5,000 simulations. See Section \ref{section:simulation} of the main text for further details on the simulation design.}
    \label{fig: time average histogram plag1, T1000}
\end{figure}

\begin{figure}[htbp!]
    \centering
    \begin{subfigure}{.5\textwidth}
    \centering
    \includegraphics[width=2.5in, height=2.5in]{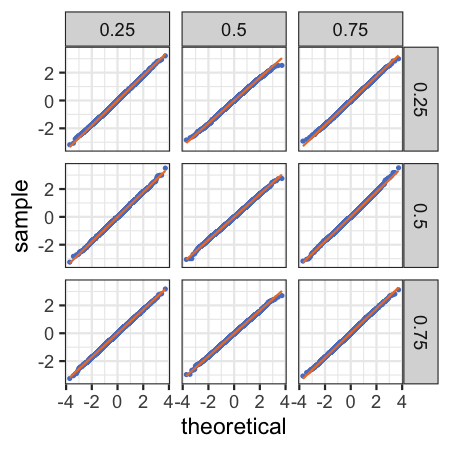}
    \caption{$\epsilon_{i, t} \sim N(0, 1)$, $T = 1000$}
    \end{subfigure}%
    \begin{subfigure}{.5\textwidth}
    \centering
    \includegraphics[width=2.5in, height=2.5in]{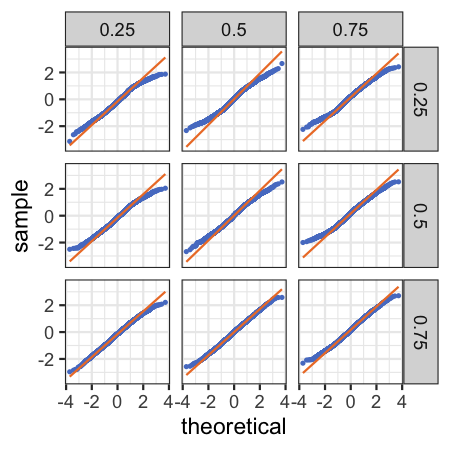}
    \caption{$\epsilon_{i, t} \sim Cauchy$, $T = 50,000$}
    \end{subfigure}
    \caption{Quantile-quantile plots for the simulated randomization distribution for $\hat{\bar{\tau}}^{\dagger}_{i \cdot}(1, 0; 1)$ under different choices of the parameter $\phi$ and treatment probability $p(w)$. The quantile-quantile plots compare the quantiles of the simulated randomization distribution (y-axis) against the quantiles of a standard normal random variable (x-axis). The 45 degree line is plotted in solid orange. The rows index the parameter $\phi \in \{0.25, 0.5, 0.75\}$, and the columns index the treatment probability $p(w) \in \{0.25, 0.5, 0.75\}$. Panel (a) plots the quantile-quantile plots for simulated randomization distribution with normally distributed errors $\epsilon_{i, t} \sim N(0, 1)$ and $T = 1000$. Panel (b) plots the quantile-quantile plots simulated randomization distribution with Cauchy distribution errors $\epsilon_{i, t} \sim Cauchy$ and $T = 50,000$. Results are computed over 5,000 simulations. See Section \ref{section:simulation} of the main text for further details on the simulation design.}
    \label{fig: time average qqplot plag1}
\end{figure}

\newpage
\clearpage
\section{Additional empirical results}\label{section: additional empirical results}
\subsection{Analysis of unit and time-specific average dynamic causal effects}

We estimate unit-specific average dynamic causal effects in the panel experiment conducted by \cite{AndreoniSamuelson(06)}. We focus on two randomly selected units in the experiment and construct estimates of their average $i,t$-th lag-$0$ dynamic causal effect, $\tau_{i,t}(1, 0;0)$ (Definition \ref{defn:pq-causal-effect}). Figure \ref{fig: unit plots, cooperation} shows the nonparametric estimates $\hat{\tau}_{i,t}(1,0;0)$ for $t \in [T]$, for the two units. The figure also contains the nonparametric estimate of the average unit-$i$ lag-$0$ dynamic causal effect $\bar \tau_{i\cdot}(1,0;0) = \frac{1}{T} \sum_{t=1}^{T} \hat{\tau}_{i,t}(1,0;0)$. The result shows that the point estimate of the average unit-$i$ lag-$0$ dynamic causal effect is positive for both units, suggesting that a larger value of $\lambda$ in the current game increases the likelihood of cooperation for both units. Since each unit only plays a total of twenty rounds, the estimated variance of these unit-specific estimators is quite large.  

\begin{figure}[htbp!]
    \centering
    \includegraphics[width = 5in, height = 3in]{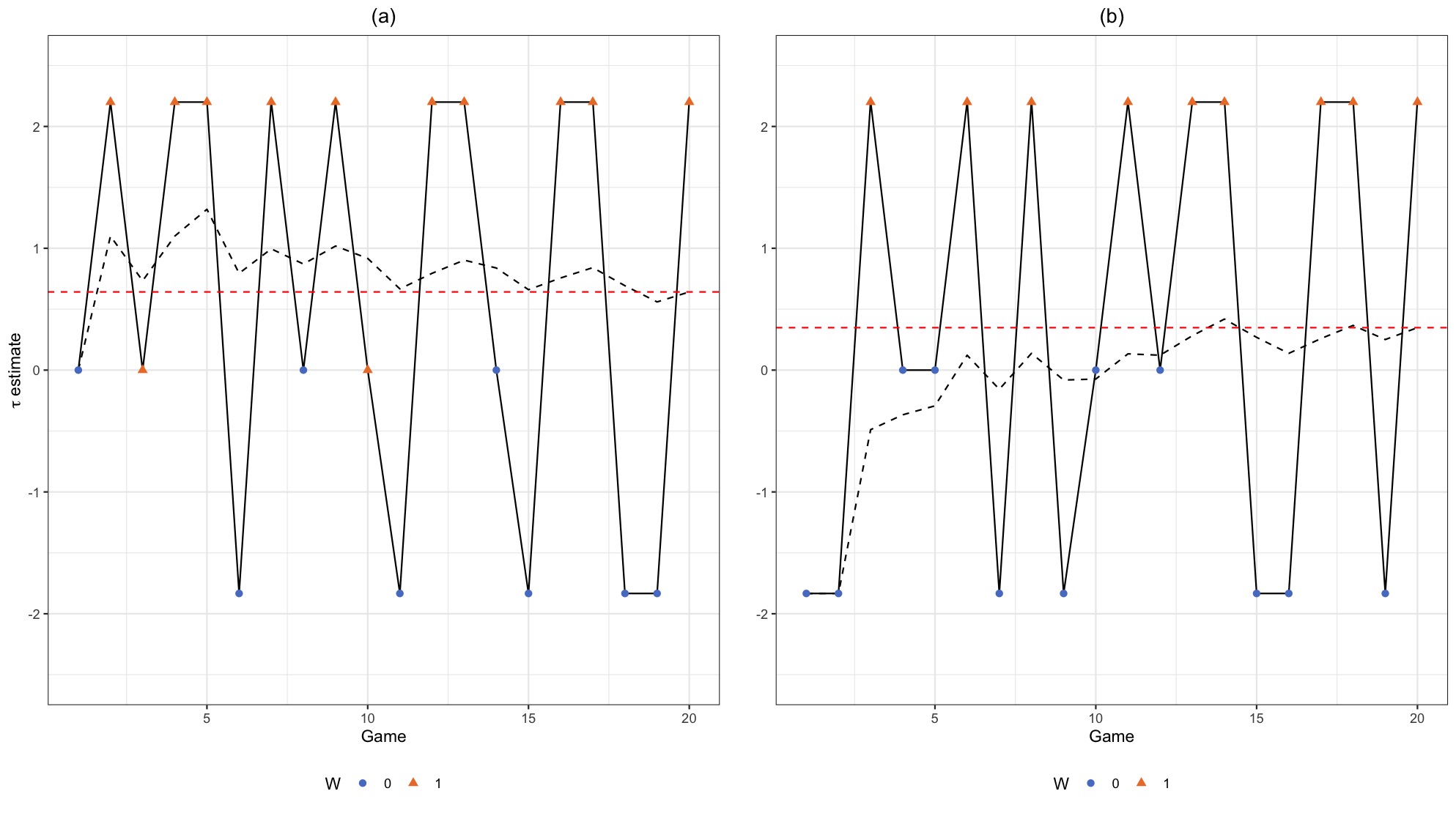}
    \caption{Estimates of the weighted average $i,t$-th lag-$0$ dynamic causal effect (Definition \ref{defn:pq-causal-effect}) of $W = \mathbbm{1}\{\lambda \geq 0.6\}$ on cooperation in period one for two units in the experiment of \cite{AndreoniSamuelson(06)}. The solid black line plots the nonparametric estimator $\hat{\tau}_{i,t}(1,0;0)$ given in Remark \ref{remark: HT weighted average causal effect}. The dashed black line plots the running average of the period-specific estimator for each unit: for each $t \in [T]$, $\frac{1}{t} \sum_{s=1}^{t} \hat{\tau}_{i,s}(1,0;0)$. The dashed red line plots the estimated weighted average unit-$i$ lag-$0$ dynamic causal effect, $\hat{\bar{\tau}}_{i \cdot}(1, 0;0) = \frac{1}{T} \sum_{t=1}^{T} \hat{\tau}_{i,t}(1,0;0)$.}
    \label{fig: unit plots, cooperation}
\end{figure}

We next estimate period-specific, weighted average dynamic causal effects that pool information across units in order to gain precision. For each time period $t \in [T]$, we construct estimates based on the nonparametric estimator of the weighted average time-$t$, lag-$p$ dynamic causal effect $\bar{\tau}_{\cdot t}^{\dagger}(1, 0; p) = \frac{1}{N} \sum_{i=1}^{N} \tau_{i,t}^{\dagger}(1, 0; p) $ for $p = 0, 1, 2, 3$. For each value of $p$, the dashed black line in Figure \ref{fig: cross section plots, cooperation} plots the estimates $\hat{\bar{\tau}}_{\cdot t}^{\dagger}(1, 0; p)$ and the grey region plots a 95\% pointwise conservative confidence band for the period-specific weighted average dynamic causal effects. For each value of $p$, there appears to be some heterogeneity in the period-specific weighted causal dynamic causal effects across time periods.

To further investigate these dynamic causal effects, the solid blue line in Figure \ref{fig: cross section plots, cooperation} plots the nonparametric estimator the total lag-$p$ weighted average causal effect $\bar{\tau}^{\dagger}(1, 0; p)$ for $p = 0, 1, 2, 3$, which further pools information across all units and time periods. The dashed blue lines plot the conservative confidence interval for the total lag-$p$ weighted average causal effect. See the main text for further discussion of the total lag-$p$ weighted average causal effect estimates.

\begin{figure}[t]
    \centering
    \includegraphics[width = 5in, height = 3in]{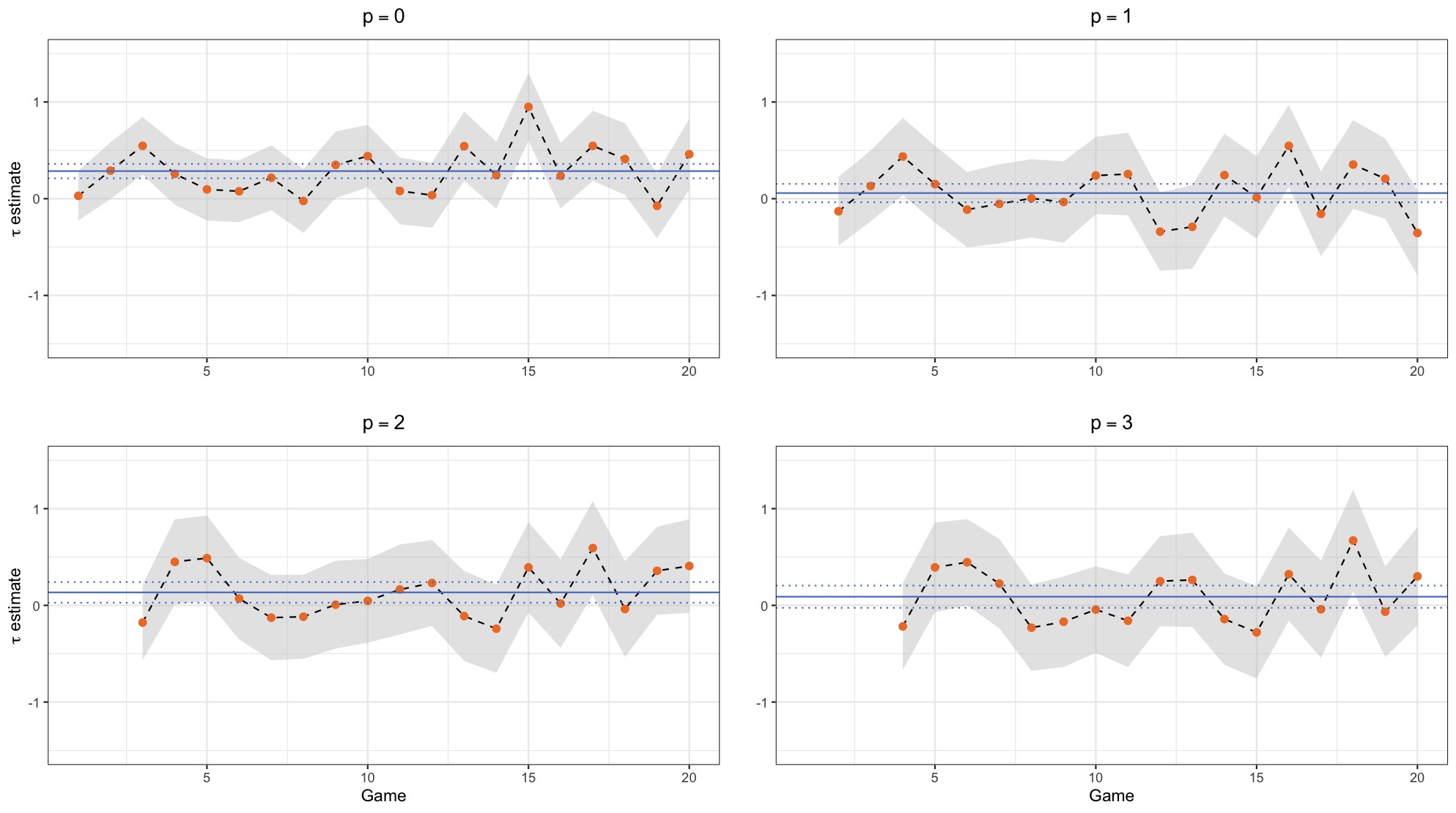}
    \caption{Estimates of the time-$t$ lag-$p$ weighted average dynamic causal effect, $\bar{\tau}^{\dagger}_{\cdot t}(1,0;p)$ of $W = \mathbbm{1}\{\lambda \geq 0.6\}$ on cooperation in period one based on the experiment of \cite{AndreoniSamuelson(06)} for each time period $t \in [T]$ and $p = 0,1,2,3$. The black dashed line plots the nonparametric estimator of the time-$t$ lag-$p$ weighted average dynamic causal effect, $\hat{\bar{\tau}}^{\dagger}_{\cdot t}(1,0;p)$, for each period $t \in [T]$. The grey region plots the 95\% point-wise confidence band for $\bar{\tau}^{\dagger}_{\cdot t}(1,0;p)$ based on the conservative estimator of the asymptotic variance of the nonparametric estimator (Theorem \ref{thm:clts}). The solid blue line plots the nonparametric estimator of the total lag-$p$ weighted average dynamic causal effect, $\hat{\bar{\tau}}^{\dagger}(1,0;p)$ and the dashed blue lines plot the 95\% confidence interval for $\bar{\tau}^{\dagger}(1,0;p)$ based on the conservative estimator of the asymptotic variance of the nonparametric estimator.}
    \label{fig: cross section plots, cooperation}
\end{figure}

\end{document}